\documentclass[reprint,amsmath,amssymb,aps,prl]{revtex4-2}

\usepackage{graphicx}% Include figure files
\usepackage{dcolumn}% Align table columns on decimal point
\usepackage{bm}% bold math
\usepackage{xcolor}
\usepackage{physics}
\usepackage{amsthm}
\usepackage{comment}
\usepackage{dsfont}
\usepackage[colorlinks=true,bookmarks=false,citecolor=blue,urlcolor=blue,linkcolor=blue]{hyperref}
\usepackage{bibunits}
\usepackage{bbold}
\usepackage{enumitem}
\usepackage{amssymb}

\theoremstyle{definition}
\newtheorem{theorem}{Theorem}
\newtheorem{lemma}{Lemma}

\newtheorem{definition}{Definition}

\newtheorem{observation}{Observation}

\makeatletter
\def\maketitle{
\@author@finish
\title@column\titleblock@produce
\suppressfloats[t]}
\makeatother

\begin{document}

\title{One pure steered state implies Einstein-Podolsky-Rosen steering for two-qubit entangled states}

\author{Qiu-Cheng Song$^{1}$}
\email{songqiucheng@kaist.ac.kr}
\author{Joonwoo Bae$^{2}$}
\email{joonwoo.bae@kaist.ac.kr}
\affiliation{$^{1}$Information $\&$ Electronics Research Institute, Korea Advanced Institute of Science and Technology (KAIST), 291 Daehak-ro, Yuseong-gu, Daejeon 34141, Republic of Korea}
\affiliation{$^{2}$School of Electrical Engineering, Korea Advanced Institute of Science and Technology (KAIST), 291 Daehak-ro, Yuseong-gu, Daejeon 34141, Republic of Korea}

\date{\today}

\begin{abstract}
In this work, we show that a two-qubit entangled state admitting at least one \emph{pure} steered state is Einstein-Podolsky-Rosen (EPR) steerable from Alice to Bob. Pure steered states signifies that the quantum steering ellipsoid of Bob is tangent to his Bloch sphere at least at a single point. Furthermore, we prove that for a two-qubit entangled state, Bob’s quantum steering ellipsoid is tangent to his Bloch sphere at exactly $N$ points, for $N\in \{ 0,1,2,\infty\}$, if and only if Alice’s quantum steering ellipsoid is tangent to her Bloch sphere at exactly $N$ points. For any two-qubit entangled state, therefore, if one party can steer the other to at least one pure state, the state is two-way EPR steerable. We also present several examples of two-qubit entangled states such that the EPR steering can be verified in terms of pure steered states. Our result addresses the Gisin theorem in a EPR steering scenario: at least a single pure steered state implies two-way steering.
\end{abstract}
\maketitle

\emph{Introduction}. Einstein-Podolsky-Rosen (EPR) steering can be formulated in a two-party scenario, for instance, involving Alice and Bob, where Alice, by performing local measurements, can remotely prepare an assemblage of conditional states for Bob in a way that cannot be reproduced by any local hidden-state (LHS) model~\cite{Einstein35,Schrodinger1935,Schrodinger1936,Wiseman2007,Jones2007}. In constrat to entanglement~\cite{Einstein35,Schrodinger1935,Schrodinger1936,Horodeckis09} and Bell nonlocality~\cite{Bell64,Brunner14}, EPR steering contains intrinsic asymmetry since it corresponds to a one-sided device-independent scenario; EPR steering relies on which party realizes a measurement in a device-independent manner. Consequently, one party's steerability does not necessarily imply steerability by other parties~\cite{Wiseman2007,Bowles14,Sekatski23}. For practical applications, EPR steering is essential for realizing one-sided device-independent quantum key distribution~\cite{Branciard2012,Walk16,RMP20,XiangYu22}.

Detection of EPR steering, being of both fundamental and practical interest, is a challenging problem in quantum information theory~\cite{RMP20}. Despite extensive efforts, including steering inequalitie~\cite{Reid89, CavalcantiE09, Saunders10, CavalcantiE15, Schneeloch13, Kogias15}, assemblage-based methods~\cite{Pusey2013, Skrzypczyk14, CavalcantiD17, Bowles16, CavalcantiD16, Rutkowski25}, density-matrix formalisms~\cite{Sania15, Nguyen16, Nguyen19, zhang25}, machine-learning techniques~\cite{Ren19, Zhang21, Wang24, Tsai25}, and quantum steering ellipsoids (QSEs)~\cite{Sania15, McCloskey17, KuHuanYu18, Song2023}, the problem is still far from being fully understood, even for two-qubit states. The main difficulty arises because arbitrary projective measurements are allowed. Bell-diagonal states~\cite{Horodecki96} are so far the only case for which a necessary and sufficient condition has been derived analytically \cite{Sania15, Nguyen16}.

Otherwise, there exist entangled states for which even known numerical methods, e.g., Ref.~\cite{Nguyen19}, fail to determine the { steering bounds}~\cite{Songqc2023}. Interestingly, these states admit a pure steered state, which also makes the construction of an LHS model {exceptionally difficult}. All these naturally raise the converse question of whether the existence of one pure steered state implies EPR steering for two-qubit entangled states, or not. QSEs are a useful tool to tackle the question~\cite{verstraete02,Jevtic2014, Sania15}. Notice that a physically realizable QSE may be tangent to the Bloch sphere at zero, one, two, or infinitely many points, where each point corresponds to a pure steered state~\cite{Braun14}.

Several known steering results for two-qubit states can be viewed geometrically in terms of the number of tangency points between the QSE and the Bloch sphere. For instance, the necessary and sufficient criteria~\cite{Sania15,Nguyen16} have been established for Bell-diagonal states~\cite{Horodecki96}, whose QSEs can have zero, two, or infinitely many tangency points. Their subclass of Werner states has either zero or infinitely many tangency points~\cite{Werner89,Wiseman2007,ZhangYujie24,Renner24,JingLing24}. Two-qubit pure entangled states, that are steerable, have infinitely many tangency points~\cite{Wiseman2007}. There are cases of two tangency points, in which Alice may construct two measurements to demonstrate EPR steering~\cite{chenjl2013,Nguyen2017}. For a single tangency point, EPR steering is possible when the probability of the pure steered state is large enough~\cite{Song2023}. Although the usefulness of a pure steered state has been evidenced, it remains valid only in limited cases.

In this work, we establish a general sufficient condition for two-qubit EPR steering in terms of a pure steered state. We show that a two-qubit state is two-way EPR steerable whenever the assemblage generated by all projective measurements contains at least one pure steered state. Equivalently, the condition for two-way steering can be rephrased as a QSE being tangent to the Bloch sphere at least at a single point. We prove this result by exploiting the one-to-one correspondence between the assemblage and Bob’s QSE, together with his Bloch vector. We show that no LHS model can reproduce the assemblage in this scenario.

\emph{Preliminaries}.  
We consider a two-party steering scenario in which Alice and Bob share an unknown state $\rho_{AB}$.
If Alice chooses a possible measurement setting $s\in\mathbf{S}$ and makes a corresponding measurement $\{M_{r|s}\}_r$ on her local system, where $r\in\mathbf{R}$ is one of the possible measurement results, then Bob's state is transformed into $\rho_{r|s}={\sigma_{r|s}}/{p_{r|s}}$
with probability $p_{r|s}=\mathrm{Tr} [(M_{r|s} \otimes I) \rho_{AB}]$,
where $\sigma_{r|s}=\mathrm{Tr}_\mathrm{A} [(M_{r|s} \otimes I) \rho_{AB}]$. The set $\{\{\sigma_{r|s}\}_{r\in \mathbf{R}}\}_{s\in\mathbf{S}}$
is called an \textit{assemblage}~\cite{Pusey2013}. The state $\rho_{AB}$ is said to be EPR steerable from Alice to Bob if there does not exist an local hidden–state (LHS) model that can reproduce the assemblage~\cite{Wiseman2007,Jones2007}, 
%i.e., if there is no decomposition of the form 
Equivalently, no decomposition of the following form exists:
\begin{align}
    \sigma_{r|s}=\int d \lambda \mu_\lambda p(r|s,\lambda)\rho_\lambda 
\end{align}
for all assemblage elements, where $\lambda$ a classical hidden variable with probability distribution $\mu_{\lambda}$, while $p(r|s,\lambda)$ and $\rho_{\lambda}$ denote the response function and the $\lambda$-indexed LHS, respectively.

A quantum steering ellipsoid (QSE) provides a geometric representation of the set of all quantum states to which Bob’s system can be steered by Alice’s local measurements on a shared two-qubit state $\rho_{AB}$~\cite{Jevtic2014}.
Any $\rho_{AB}$ can be expressed in terms of 
the standard Pauli basis
$\boldsymbol{\sigma}=(\sigma_x,\sigma_y,\sigma_z)$ and the identity operator $I$ as
$\rho_{AB}= 1/4(I\otimes I+ \boldsymbol{a}\cdot\boldsymbol{\sigma}\otimes I+ I\otimes \boldsymbol{b}\cdot\boldsymbol{\sigma} + \sum_{i,j} T_{ij}\,\sigma_i\otimes\sigma_j)$,
where $i,j\in \{x,y,z\}$, ${\boldsymbol{a}}=(a_x,a_y,a_z)^\top$ and ${\boldsymbol{b}}=(b_x,b_y,b_z)^\top$ are the Bloch vectors of Alice’s and Bob’s reduced states, respectively. These vectors and the spin-correlation matrix are given by
${a}_{i}=\operatorname{Tr}\left[(\sigma_{i} \otimes I)\rho_{AB}\right]$, ${b}_{j}=\operatorname{Tr}\left[(I \otimes \sigma_{j})\rho_{AB}\right]$, and ${T}_{ij}=\operatorname{Tr}\left[(\sigma_{i} \otimes \sigma_{j})\rho_{AB}\right]$.
If Alice performs projective measurements, the corresponding effects are denoted by $E_{\pm|s}=(I\pm\hat{\boldsymbol{n}}_{s} \cdot \boldsymbol{\sigma})/2$
with an unit vector $\hat{\boldsymbol{n}}_{s}$ specifying the measurement direction.
The resulting ensemble prepared for Bob consists of two conditional states 
\begin{align}\label{eq:conds}
    \rho_{\pm|s}=\frac{1}{2}\left[{I}+ \frac{(\boldsymbol{b}\pm {T}^{\top} \hat{\boldsymbol{n}}_{s}) \cdot \boldsymbol{\sigma}}{(1\pm \hat{\boldsymbol{n}}_{s}\cdot \boldsymbol{a})}\right],
\end{align}
with probabilities 
\begin{align}\label{eq:condp}
    p_{\pm|s}=\frac{1}{2}(1\pm\hat{\boldsymbol{n}}_{s} \cdot \boldsymbol{a}).
\end{align}
The Bloch vectors of 
$\{\rho_{\pm|s}\}_{s\in\mathbf{S}}$
form an ellipsoidal surface in his Bloch sphere $\mathcal{S}_B$, known as the quantum steering ellipsoid (QSE)~\cite{verstraete02}:
\begin{equation}\label{qellipsoid}
\partial\mathcal{E}_B=\left\{\frac{\boldsymbol{b}\pm T^{\top} \hat{\boldsymbol{n}}_{s}}{1\pm \hat{\boldsymbol{n}}_{s} \cdot \boldsymbol{a}}\right\}_{\hat{\boldsymbol{n}}_{s}}.
\end{equation}
Here we assume $|\boldsymbol{a}|<1$, as is necessarily the case for any entangled two-qubit state. If $|\boldsymbol{a}|=1$, then Alice's reduced state is pure, which implies $\rho_{AB}=\rho_A\otimes\rho_B$. Consequently, all conditional states of Bob occurring with nonzero probability coincide with $\rho_B$, and the QSE reduces to the single point $\boldsymbol{b}$. Thus, for an entangled two-qubit state, $|\boldsymbol{a}|<1$, and Eq.~\eqref{eq:condp} immediately yields
\begin{equation}\label{eq:pmin}
p_{\pm|s}\geq \frac{1-|\boldsymbol{a}|}{2}
=:p_{\min}>0.
\end{equation}
We write $\mathcal E=\mathcal E_B$,
$\mathcal B=\{\boldsymbol x\in\mathbb R^3:|\boldsymbol x|\le1\}$,
and $\mathcal S=\mathcal S_B=\partial\mathcal B$. Here $\mathcal E$
denotes the solid ellipsoid of normalized conditional Bloch vectors;
its surface is $\partial\mathcal E$.
The center of QSE is $\boldsymbol{c}_B=\gamma^2(\boldsymbol{b}-T^{\top}\boldsymbol{a})$, and its semiaxis lengths are the square roots of the eigenvalues of the ellipsoid matrix 
\begin{equation}\label{ellipsoidmatrix }
Q_B=\gamma^2{\left(T^{\top}-\boldsymbol{b a}^{\top}\right)}\left(I+\gamma^2{\boldsymbol{a a}^{\top}}\right)\left(T-\boldsymbol{a} \boldsymbol{b}^{\top}\right),
\end{equation}
where $\gamma^2:=1/(1-|\boldsymbol{a}|^{2})$~\cite{Jevtic2014}.
The eigenvectors of $Q_B$ determine the orientation of the ellipsoid. In particular, 
\begin{equation}\label{eq:ellipsoid}
\mathcal E
=\left\{\boldsymbol x:
(\boldsymbol x-\boldsymbol c_B)^\top Q_B^{-1}
(\boldsymbol x-\boldsymbol c_B)\le1\right\}
\subseteq\mathcal B.
\end{equation}
The QSE $\mathcal{E}_A$ for Alice is obtained by exchanging the roles of Alice and Bob, 
i.e., $\boldsymbol{a}\leftrightarrow\boldsymbol{b}$ and $T\to T^{\top}$~\cite{Jevtic2014}. Further properties of QSEs can be found in Refs.~\cite{Jevtic2014,Milne14,Milne15,FanHeng15,Shuming16,Zhang19,Divyamani23,XuKai24}.

\begin{figure}[t]
    \centering
    \includegraphics[width=0.21\textwidth]{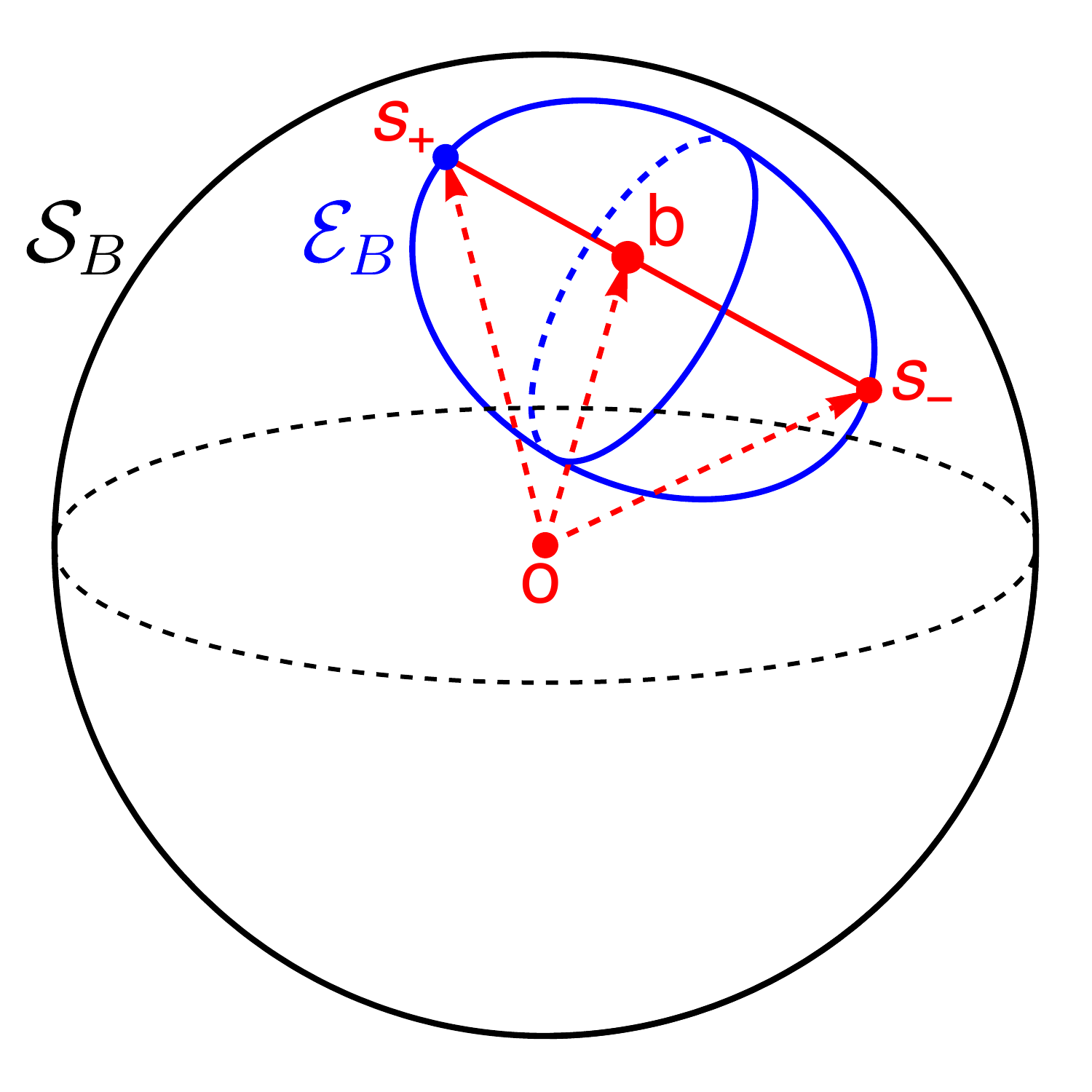}
    \includegraphics[width=0.23\textwidth]{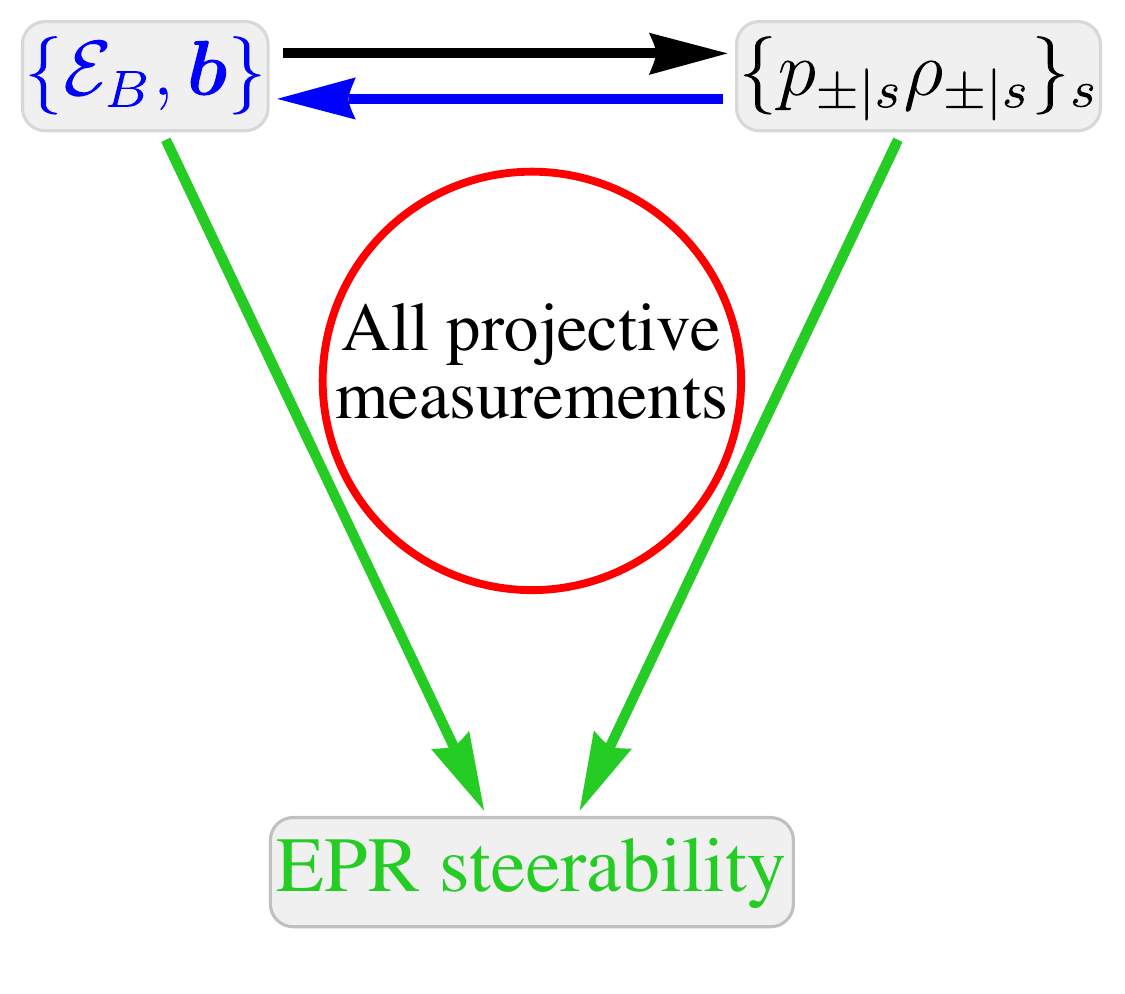}
    \caption{One-to-one correspondence between the pair $\{\mathcal{E}_B,\boldsymbol{b}\}$ and the assemblage $\{p_{\pm|s}\rho_{\pm|s}\}_{s}$ obtained from all possible projective measurements. 
    The point $\mathsf{b}$ inside QSE $\mathcal{E}_B$ and the points $\mathsf{s}_{\pm}$ on its surface represent the Bloch vectors of Bob's reduced state $\rho_B$ and the steered states $\rho_{\pm | s}$, respectively. 
    $\{\mathcal{E}_B, \boldsymbol{b}\}$ and $\{p_{\pm|s}\rho_{\pm|s}\}_s$ determine each other and hence determine EPR steerability from Alice to Bob.}
    \label{fig:steeredstate}
\end{figure}

Then, EPR steerability, an assemblage, and a QSE are related as follows. For any two-qubit state $\rho_{AB}$, if Alice performs all possible projective measurements,
EPR steerability is fully characterized by the assemblage $\Sigma_B:=\{p_{\pm|s}\rho_{\pm|s}\}_{s\in\mathbf{S}}$. 
Moreover, $\Sigma_B$ determines both Bob's QSE $\mathcal{E}_B$ 
and the Bloch vector $\boldsymbol{b}$ of his reduced state $\rho_B$ of $\rho_{AB}$; we denote this pair by $\{\mathcal{E}_B, \boldsymbol{b}\}$. 
Specifically, the set of Bloch vectors of all normalized conditional states in $\Sigma_B$ forms the surface of $\mathcal{E}_B$, while $\boldsymbol{b}$ is fixed by any ensemble in $\Sigma_B$. 
By the no-signaling principle, every measurement setting $s$ yields an ensemble satisfying
$\rho_B=p_{+|s}\rho_{+|s}+p_{-|s}\rho_{-|s}$.
Equivalently, 
%$\mathsf{b}=p_{+|s}\mathsf{s}_{+|s}+p_{-|s}\mathsf{s}_{-|s}$ 
\begin{align}\label{eq:nosignaling}
   \mathsf{b}=p_{+|s}\mathsf{s}_{+|s}+p_{-|s}\mathsf{s}_{-|s} 
\end{align}
as illustrated in Fig.~\ref{fig:steeredstate}, where
\begin{align}\label{ssprob}
    p_{\pm |s }=\frac{|\mathsf{bs_{\mp}}|}{|\mathsf{s_{+}s_{-}}|},
\end{align}
points $\mathsf{b}$ and $\mathsf{s}_{\pm}$ correspond to the Bloch vectors of states $\rho_B$ and $\rho_{\pm|s}$, respectively~\cite{Song2023}. 

Conversely, the set $\{\mathcal{E}_B, \boldsymbol{b}\}$ determines the assemblage $\Sigma_B$, and hence also determines EPR steerablity from Alice to Bob.
More specifically,  any line passing through the point $\mathsf{b}$ intersects the surface of $\mathcal{E}_B$ at two points $\mathsf{s}_{\pm}$, which correspond to two conditional states $\rho_{\pm|s}$ in an ensemble in $\Sigma_B$. The probabilities associated with  $\rho_{\pm|s}$ are given by Eq.~\eqref{ssprob}.
For a nonzero-volume QSE, the pair $\{\mathcal E_B,\boldsymbol b\}$ determines the assemblage generated by all projective measurements up to relabeling of the measurement settings and outcomes. It specifies the two-outcome ensembles, including the conditional states and their probabilities, but not Alice's measurement directions in a fixed reference frame. Such relabeling preserves the existence of an LHS model and therefore leaves steerability unchanged.
Hence, there exists a one-to-one correspondence between the pair $\{\mathcal{E}_B,\boldsymbol{b}\}$ and the assemblage $\Sigma_B$ obtained from all possible projective measurements.

Let us clarify the notion of a pure steered state and its geometric counterpart, a tangency point, both of which are essential for stating our main results. We summarize the scenario considered throughout this work, namely, the case in which at least one pure steered state exists.

\begin{definition}\label{Scenario} (Scenario)
Let Alice and Bob share a two-qubit entangled state. Alice is allowed to perform arbitrary projective measurements on her subsystem, thereby generating an assemblage, denoted by $\Sigma_B$, for Bob. The assemblage contains at least one pure state with nonzero probability. Equivalently, Bob’s nonzero-volume quantum steering ellipsoid, denoted by $\mathcal{E}_B$, is tangent to his Bloch sphere $\mathcal{S}_B$ at least at a single point.
\end{definition}

Note that the number of pure steered states in $\Sigma_B$ coincides with the number of tangency points between $\mathcal{E}_B$ and $\mathcal{S}_B$. This number, denoted by $N$, can take only the values $0$, $1$, $2$, or $\infty$~\cite{Braun14}.

\emph{Main results}. Let us now consider the scenario summarized above and present the main results, tangency points as a sufficient condition for EPR steering of two-qubit entangled states. We first present one of our principal results.

\begin{theorem}\label{onepsteering}
For a quantum steering ellipsoid of Bob with nonzero volume, exactly one tangency point between the ellipsoid and a Bloch sphere implies EPR steering from Alice to Bob.
\end{theorem}

A mixed entangled two-qubit state may yield either one or two distinct pure steered states on either subsystem. Previous work established EPR steering when a pure steered state occurs with sufficiently high probability~\cite{Song2023}. Theorem~\ref{onepsteering} resolves the single-state case by removing this probability requirement.
When two distinct pure steered states exist, they can arise from the two outcomes of a single projective measurement or from outcomes of two different projective measurements. In the former case, the measurement yielding the two pure states, together with a second different projective measurement, suffices to demonstrate steering~\cite{chenjl2013,Nguyen2017}. The latter case has remained unresolved and is settled by the following theorem.

\begin{theorem}\label{twopsteering}
For a quantum steering ellipsoid of Bob with nonzero volume, exactly two tangency points between the ellipsoid and his Bloch sphere imply EPR steering from Alice to Bob.
\end{theorem}

The proofs of Theorems~\ref{onepsteering} and~\ref{twopsteering} are presented in the Methods section. 
These two theorems provide criteria for EPR steering from Alice to Bob when Bob's QSE has one or two tangency points with the Bloch sphere. Steering in the reverse direction, from Bob to Alice, can be found as follows.

\begin{theorem}\label{ellippn}
For a two-qubit entangled state, a quantum steering ellipsoid of Alice is tangent to her Bloch sphere at exactly $N$ points for $(N\in\{0,1,2,\infty \})$, so is the ellipsoid of Bob. 
\end{theorem}

\begin{proof}
A physical QSE may touch the Bloch sphere at $N$ points where $(N\in{0,1,2,\infty})$~\cite{Braun14}.
For cases $N=1,2$, we provide separate proofs in Lemmas~\ref{onetangpoint} and~\ref{twotangpoints} in the Supplemental Material~\cite{Suppm}, respectively. 
For the case $N=\infty$, the corresponding two-qubit entangled state is pure. Consequently, the QSEs of both Alice and Bob coincide with their Bloch spheres, and the statement holds trivially. 
Since the result holds for $N=1,2,\infty$, it follows that it also holds for case $N=0$.  
\end{proof}

\begin{figure}[t]
    \centering
    \includegraphics[width=0.47\textwidth]{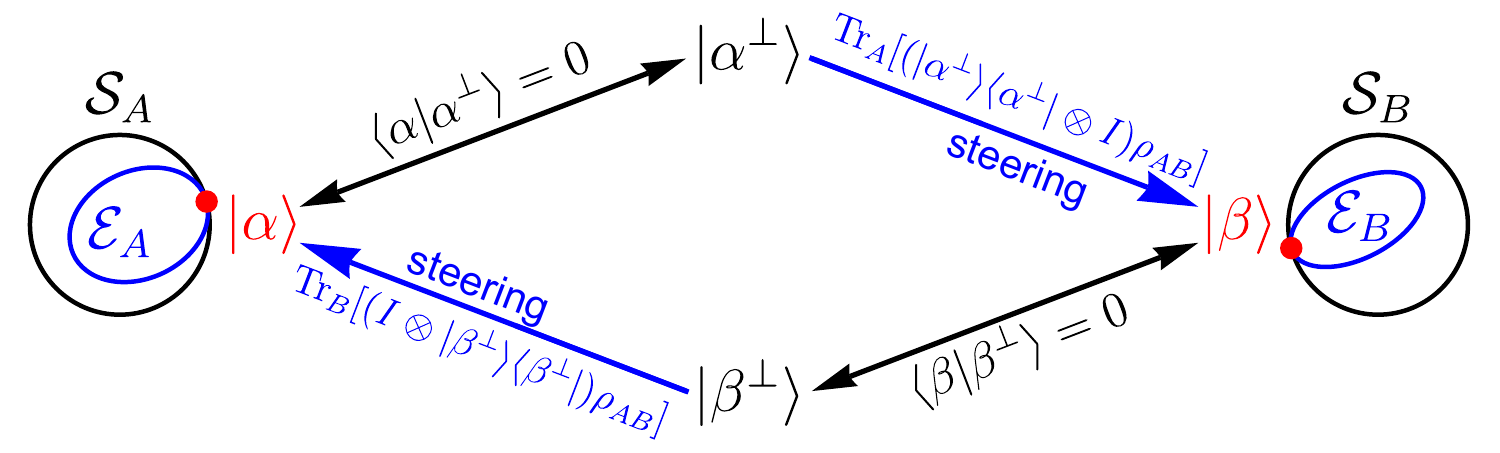}
    \caption{The one-to-one correspondence between Alice's pure steered state $|\alpha\rangle$ and Bob's pure steered state $|\beta\rangle$.}
    \label{fig:twoellip}
\end{figure}

We emphasize the following key property identified in the proofs of Lemmas~\ref{onetangpoint} and~\ref{twotangpoints}, as illustrated in Fig.~\ref{fig:twoellip}.
\begin{observation}\label{Observation}
Let $\rho_{AB}$ be a two-qubit entangled state admitting at least one pure steered state on either subsystem. 
There exists a one-to-one correspondence between Alice's pure steered state $|\alpha\rangle_A$ and Bob's pure steered state $|\beta\rangle_B$. Specifically, Alice can steer Bob's system to $|\beta\rangle_B$ via the measurement effect $|\alpha^\perp\rangle\langle \alpha^\perp|_A$, while Bob can steer Alice's system to $|\alpha\rangle_A$ via $|\beta^\perp\rangle\langle \beta^\perp|_B$, where $\langle \alpha|\alpha^\perp\rangle=0$ and $\langle \beta|\beta^\perp\rangle=0$.
\end{observation}

We also provide an alternative proof of this observation in the Supplemental Material.

The case $N=\infty$ corresponds to a two-qubit pure state. By Gisin’s theorem, all two-qubit pure entangled states are Bell nonlocal~\cite{Gisin91}. Since Bell nonlocality implies EPR steering, such states are necessarily steerable~\cite{Wiseman2007}. Combining Theorems~\ref{onepsteering}, \ref{twopsteering}, and~\ref{ellippn}, we obtain the following principal result.

\begin{theorem}\label{twowaysteering}
Let $\rho_{AB}$ be a two-qubit entangled state shared by Alice and Bob. If a local projective measurement by either
party prepares a pure conditional state for the other with nonzero probability, then $\rho_{AB}$ is two-way EPR
steerable.
\end{theorem}

This theorem is equivalent to the following statement:
For any QSE with nonzero volume, the existence of at least one point of tangency between the QSE and the Bloch sphere implies two-way EPR steering.

Since projective measurements are a subset of POVMs, the conclusion also holds when general POVMs are allowed. In Scenario~\ref{Scenario}, where the nonzero-volume QSE touches the Bloch sphere, the entanglement assumption in Theorem~\ref{twowaysteering} is equivalent to the nonzero-volume assumption in Theorem~\ref{onepsteering}. Indeed, every entangled two-qubit state has a nonzero-volume QSE. Conversely, a nonzero-volume QSE that touches the Bloch sphere cannot be contained in a tetrahedron contained within the Bloch ball and therefore represents an entangled state, by the nested-tetrahedron criterion for separability~\cite{Jevtic2014}.

These assumptions cannot be omitted. For example, the product state $\rho_{AB}=|0\rangle\langle0|\otimes I/2$
allows Bob to prepare the pure conditional state $|0\rangle\langle0|$ on Alice's side, yet is unsteerable in either direction. Here, $|\boldsymbol a|=1$, Alice's QSE degenerates to a single point on her Bloch sphere, and Bob's QSE degenerates to the center of his Bloch ball.

For $N=0$, neither party can prepare a pure conditional state on the other subsystem by a local measurement. 
This does not imply that the shared state is unsteerable. For entangled two-qubit states, the existence of a pure steered state is a sufficient, but not necessary, condition for EPR steering. Our criterion is therefore inconclusive in this case.

\emph{Examples}.---We consider a  class of entangled two-qubit states that admit a unique pure steered state, given by
\begin{align}\label{examples}
    \rho_{AB} = q |\psi_-\rangle\langle \psi_-| +(1-q)|0\rangle\langle 0|\otimes\frac{I}{2},
\end{align}
where $|\psi_-\rangle=(|01\rangle-|10\rangle)/\sqrt{2}$ and $0<q\leq 1$. The state is entangled for all $q>0$ and admits a local hidden-variable (LHV) model for projective measurements when $q\leq 1/2$~\cite{Brunner13}. For $0<q<1$, either party, by performing the Pauli measurement $\sigma_z$, can steer the other party's state to a unique pure state, namely $|0\rangle$, with nonzero probability. According to Theorem~\ref{twowaysteering}, the state $\rho_{AB}$~\eqref{examples} is therefore two-way steerable, despite its asymmetry between Alice and Bob. Further examples are constructed and discussed in the Supplemental Material~\cite{Suppm}.

Finally, the results so far also provide a sufficient condition for measurement incompatibility. EPR steering is closely related to measurement incompatibility~\cite{Quintino14, Uola14, Kiukas17, Uola2015, Uola25, porto26}. In fact, there exists a one-to-one correspondence between steering and incompatibility~\cite{Uola2015}.
An assemblage $\{\sigma_{r|s}\}_{r,s}$ is steerable if and only if the corresponding steering-equivalent observables (the pretty-good measurements) $B_{r|s}=\rho_{B}^{-1/2} \sigma_{r|s} \rho_{B}^{-1/2}$ are incompatible~\cite{Uola2015}. Hence, under the map $\sigma_{r|s} \mapsto B_{r|s}$, the state assemblage in Scenario~\ref{Scenario} can be transformed into a steering-equivalent measurement assemblage (POVMs), which is necessarily incompatible.  Conversely, if a measurement assemblage maps to a state assemblage containing at least one pure state in Scenario~\ref{Scenario}, it is incompatible.

\emph{Discussions}. In summary, we have shown a general sufficient condition for EPR steering of two-qubit states. The existence of a pure steered state rules out any LHS model, i.e., EPR steering can be detected if an assemblage resulting from arbitrary projective measurements on a shared two-qubit entangled state contains at least one pure state. We further proved that Alice and Bob have the same number of tangency points between a QSE and a Bloch sphere, implying the same number of pure states in both assemblages. Hence, Alice can steer at least one pure state on Bob's side, and so can Bob, leading to two-way EPR steering. In addition, the proofs we presented here are also constructive. For instance, in the Supplemental Material~\cite{Suppm}, we construct five classes of two-qubit entangled states that are steerable.

For a two-qubit state, a pure steered state by a projective measurement, together with two distinct mixed conditional states produced by another projective measurement on the same subsystem, suffices to certify both entanglement and two-way steering. Entanglement follows from the nested-tetrahedron criterion, which states that a two-qubit state is separable if and only if its steering ellipsoid is contained in a tetrahedron that is itself contained in the Bloch sphere~\cite{Jevtic2014}. A separable state admitting a pure steered state must have a steering ellipsoid that is either a point or a line segment with a pure endpoint. Neither case permits another projective measurement to produce two distinct mixed conditional states. These observations therefore establish entanglement, after which our theorem implies two-way steering.

Our results resolve the limitation of the numerical method in Refs.~\cite{Nguyen19, Songqc2023}. Two-qubit entangled states admitting at least one pure steered state, hence steerable from the presented result, may not be detected by the aforementioned numerical approach. In addition, our approach with tangency points of a QSE in a Bloch sphere goes beyond the previous results in terms of the volume of a QSE~\cite{McCloskey17,KuHuanYu18}. For instance, our results show that a QSE with arbitrary nonzero volume in a Bloch sphere implies EPR steering, provided that the QSE is tangent to the Bloch sphere.

We remark that our results establish the Gisin theorem~\cite{Gisin91} for EPR steering. We recall the Gisin theorem for nonlocality addresses that all two-qubit pure states are Bell-nonlocal; the result has been generalized to arbitrary bipartite states and also to multipartite systems~\cite{GISIN92,Chen04,Choudhary10,LiFei10,YuSixia12}. Our results here show that any two-qubit entangled state that admits a pure steered state implies EPR steering. Hence, a pure steered state is a sufficient condition for EPR steering. As the Gisin theorem generalizes to arbitrary pure states, it would also be interesting to determine whether our results with a pure steered state hold in high-dimensional and multipartite quantum systems. 
The connection to Gisin’s theorem naturally suggests an interesting future direction: certifying single-party pure states via quantum steering.
The question also aligns with the application of the Gisin theorem for the certification of entangled states. We reiterate that Bell-nonlocality has been a tool to certify entangled states, e.g.~\cite{PhysRevLett.121.180503}. 

In addition, we observe that a class of Bell-diagonal states with two pure steered states is not only two-way steerable but also Bell-nonlocal, see Supplemental Material~\cite{Suppm}. While we have shown here that one pure steered state implies EPR steering, it would be interesting to determine whether two pure steered states imply Bell nonlocality.

\emph{Methods}.

Proof of Theorem~\ref{onepsteering}
\begin{proof}
Let $\mathsf p$ be the tangency point. Without loss of generality, choose Bob's coordinates so that $\mathsf p=(0,0,-1)$ as shown in Fig.~\ref{fig:Threeplhs}. Suppose, for a contradiction, that an LHS model for the assemblage generated by all projective measurements has the form
\begin{equation}\label{eq:lhs}
\sigma_{\pm|s}
=\int_{\mathcal B}p(\pm|\hat{\boldsymbol n}_s,\boldsymbol\lambda)
\rho_{\boldsymbol\lambda}\,\dd\omega(\boldsymbol\lambda),
\quad
\rho_{\boldsymbol\lambda}
=\frac12(I+\boldsymbol\lambda\cdot\boldsymbol\sigma).
\end{equation}
Here $\omega$ is a fixed, countably additive Borel probability measure on $\mathcal B$, shared by all measurement settings. The hidden states may be pure or mixed. For each fixed setting, the response functions are measurable in $\boldsymbol\lambda$, satisfy $0\le p(\pm|\hat{\boldsymbol n}_s,\boldsymbol\lambda)\le1$ and
\begin{equation}
p(+|\hat{\boldsymbol n}_s,\boldsymbol\lambda)
+p(-|\hat{\boldsymbol n}_s,\boldsymbol\lambda)=1
\end{equation}
for $\omega$-almost every $\boldsymbol\lambda$. Using Bloch vectors as hidden-state labels entails no restriction. No continuity or deterministic-response assumption is imposed. Representatives can be chosen to obey these bounds and normalization everywhere by redefining the responses on a null set. Such changes leave all assemblage elements unchanged.

Mathematically, $\omega$ is a Borel probability measure on $\mathcal B$. This notation does not restrict the class of LHS models. The measure may include finite or countably infinite discrete contributions, volume or surface distributions, singular continuous components, and mixtures. No density with respect to volume is assumed, and no joint measurability in the setting and the hidden variable is needed below. In particular, Alice's responses need not have a Born-rule or affine dependence on the measurement direction. Summing the two outcomes of any setting gives
\begin{equation}\label{eq:lhs-marginal}
\int_{\mathcal B}\rho_{\boldsymbol\lambda}\,\dd\omega(\boldsymbol\lambda)
=\rho_B.
\end{equation}

\begin{figure}[t]
    \centering
    \includegraphics[width=0.22\textwidth]{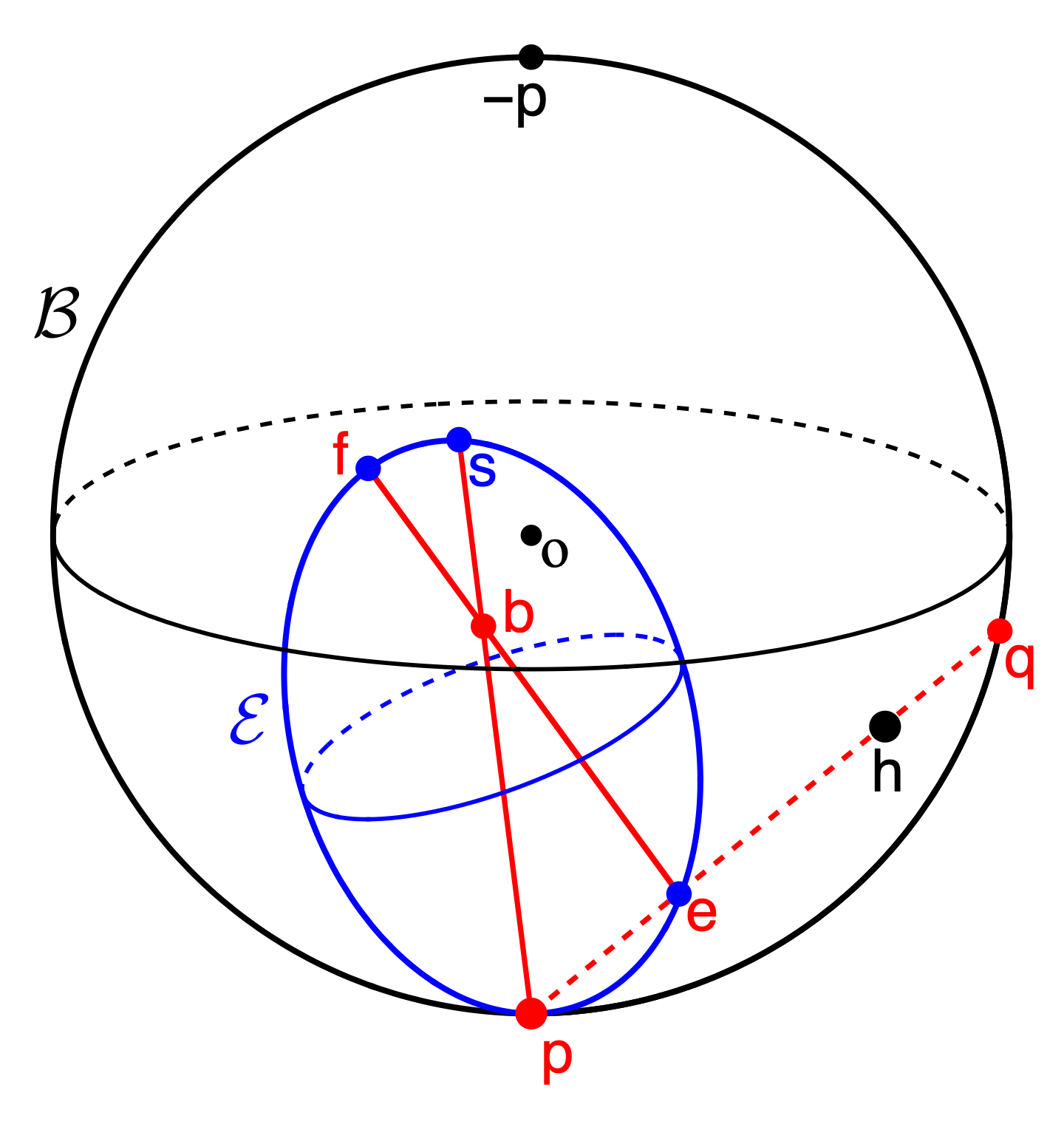}
    \hspace{0.01\textwidth}
    \includegraphics[width=0.22\textwidth]{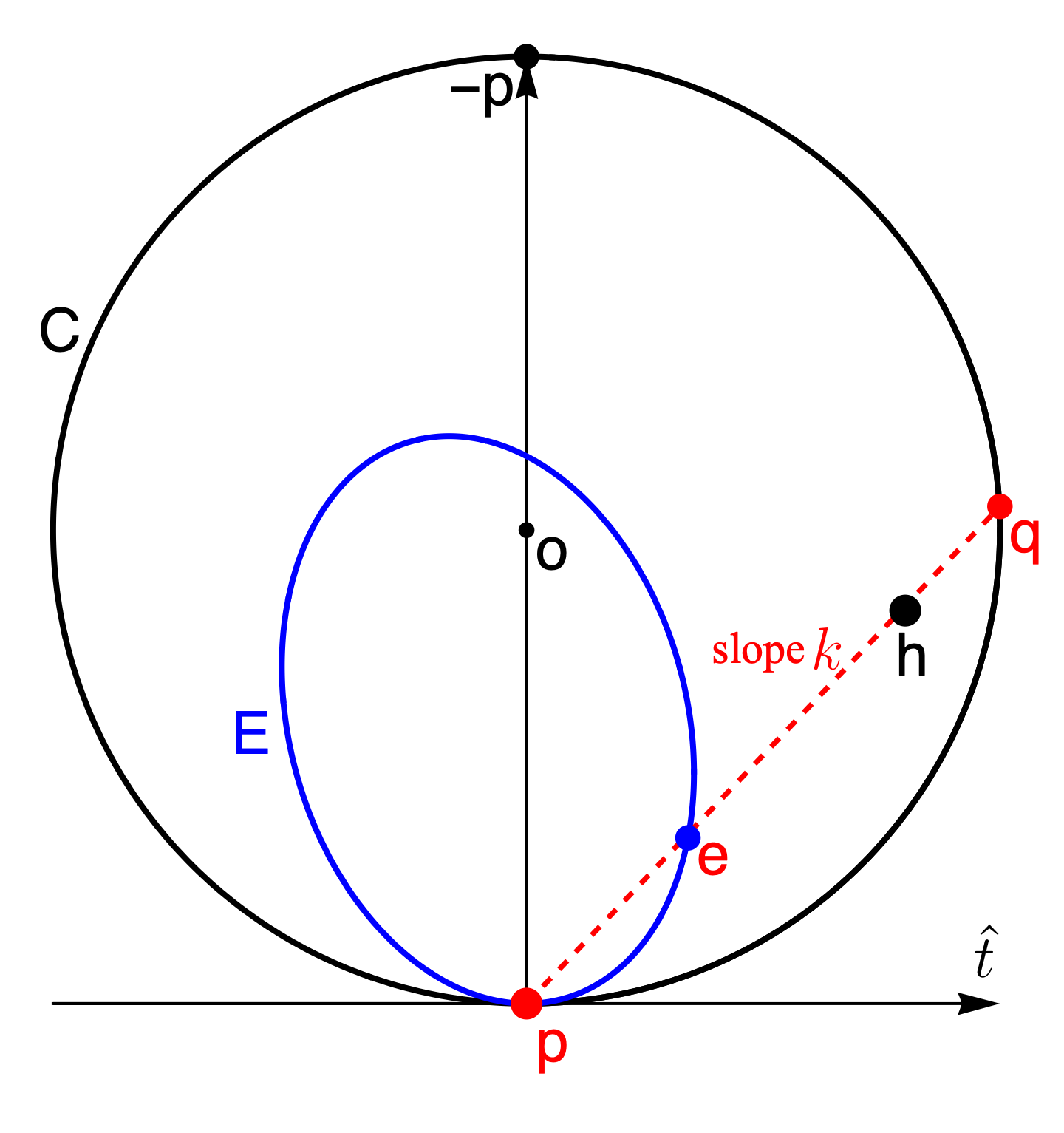}
    \caption{Illustration of Alice’s response function in a hypothetical LHS model.  Here, $\mathsf{C}=\partial\mathcal{B}\cap \Pi $, $\mathsf{E}=\partial\mathcal{E}\cap \Pi $, where  $\Pi=\mathsf p+\operatorname{span}\{\hat{\boldsymbol t},-\mathsf p\}$ is the fixed plane determined by $\mathsf{p}$ and $\hat{\boldsymbol t}$. If an LHS model exists, it must include the pure state $\rho_\mathsf{p}$ at the tangency point $\mathsf{p}$ with sufficient probability, together with other LHSs distributed over the Bloch ball $\mathcal{B}$ with probability density $\mu$. Any steered state $\sigma_\mathsf{e}$ at point $\mathsf{e}$ on the QSE $\mathcal{E}$ would then be reproduced by $\rho_\mathsf{p}$ with some probability, together with some LHSs specified by the response function $p(+|\hat{\boldsymbol{n}}_\mathsf{e},{\boldsymbol{\lambda}})$. We show that no
    LHS model can reproduce the assemblage in this scenario.}
    \label{fig:Threeplhs}
\end{figure}

Choose a projective measurement direction $\hat{\boldsymbol n}_{\mathsf p}$ whose $+$ outcome produces the pure state $\rho_{\mathsf p}$. Write $\rho_{\mathsf s}$ for the complementary conditional state. 
By Eq.~\eqref{eq:nosignaling},
\begin{equation}\label{eq:pure-ensemble}
\mathsf b=p_{\mathsf p}\mathsf p+p_{\mathsf s}\mathsf s,
\quad
p_{\mathsf p}=\frac{|\mathsf{bs}|}{|\mathsf{ps}|}>0,
\quad
p_{\mathsf s}=1-p_{\mathsf p}>0.
\end{equation}
Here points such as $\mathsf p$, $\mathsf b$, and $\mathsf s$ denote Bloch vectors as well as their geometric locations; in particular, $\rho_{\mathsf p}=(I+\mathsf p\cdot\boldsymbol\sigma)/2$. Indeed, both probabilities are positive by Eq.~\eqref{eq:pmin}, so $\mathsf b$ lies strictly between $\mathsf p$ and $\mathsf s$; subtracting $\mathsf s$ gives the distance ratio above. The hypothetical LHS model must reproduce
\begin{equation}\label{asse01}
\sigma_{\mathsf p}=p_{\mathsf p}\rho_{\mathsf p}
=\int_{\mathcal B}
 p(+|\hat{\boldsymbol n}_{\mathsf p},\boldsymbol\lambda)
 \rho_{\boldsymbol\lambda}\,\dd\omega(\boldsymbol\lambda).
\end{equation}
Taking the trace gives the total weight $p_{\mathsf p}$, while comparing Pauli coefficients gives the weighted Bloch vector $p_{\mathsf p}\mathsf p$. Taking its scalar product with $\mathsf p$ and subtracting from the trace identity yields
\begin{equation}\label{eq:pure-support}
\int_{\mathcal B}
 p(+|\hat{\boldsymbol n}_{\mathsf p},\boldsymbol\lambda)
 (1-\mathsf p\cdot\boldsymbol\lambda)
 \,\dd\omega(\boldsymbol\lambda)=0.
\end{equation}
The integrand is nonnegative, and its second factor vanishes only at $\boldsymbol\lambda=\mathsf p$. Thus $p(+|\hat{\boldsymbol n}_{\mathsf p},\boldsymbol\lambda)=0$ for $\omega$-almost every $\boldsymbol\lambda\ne\mathsf p$. Taking the trace in Eq.~\eqref{asse01} consequently gives
\begin{equation}\label{eq:atom-response}
p_{\mathsf p}=\tilde p_{\mathsf p}
 p(+|\hat{\boldsymbol n}_{\mathsf p},\mathsf p),
\qquad
\tilde p_{\mathsf p}:=\omega(\{\mathsf p\})\ge p_{\mathsf p}>0.
\end{equation}
Since $\{\mathsf p\}$ has positive $\omega$-measure, the response bounds and normalization hold at $\mathsf p$ for every setting. Separate the entire probability weight of this hidden state:
\begin{equation}\label{eq:split}
\omega=\tilde p_{\mathsf p}\delta_{\mathsf p}+\mu,
\qquad
\mu(D)=\omega(D\setminus\{\mathsf p\}),
\end{equation}
where $D\subseteq\mathcal B$ is any Borel set and $\delta_{\mathsf p}$ is the probability measure concentrated at $\mathsf p$. The residual measure is fixed, finite, and nonnegative:
\begin{equation}
\mu(\{\mathsf p\})=0,\qquad
M:=\mu(\mathcal B)=1-\tilde p_{\mathsf p}\le1.
\end{equation}
It is not renormalized and may still assign nonzero weights to infinitely many distinct hidden states approaching $\mathsf p$. No density with respect to volume is assumed. If such a density $\mu_{\boldsymbol\lambda}$ exists, then $\mathrm d\mu(\boldsymbol\lambda)=\mu_{\boldsymbol\lambda}\,\mathrm dV$. For completeness, response normalization gives $p(-|\hat{\boldsymbol n}_{\mathsf p},\boldsymbol\lambda)=1$ for $\mu$-almost every $\boldsymbol\lambda$. Hence the complementary outcome and the marginal satisfy
\begin{align}
&\sigma_{\mathsf s}
=\bar p_{\mathsf p}\rho_{\mathsf p}
+\int_{\mathcal B}\rho_{\boldsymbol\lambda}\,\dd\mu,
\quad \bar p_{\mathsf p}:=\tilde p_{\mathsf p}-p_{\mathsf p}\ge0,
\label{asse02}\\
&\tilde p_{\mathsf p}+\mu(\mathcal B)=1,\label{probcond}\\
&\tilde p_{\mathsf p}\rho_{\mathsf p}
+\int_{\mathcal B}\rho_{\boldsymbol\lambda}\,\dd\mu=\rho_B.
\label{statecond}
\end{align}
The last identity follows by summing the conditional states, not from scalar normalization alone.

For every $\boldsymbol x\in\mathcal E\subseteq\mathcal B$, $\mathsf p\cdot\boldsymbol x\le1$, with equality at $\boldsymbol x=\mathsf p$. Thus the plane $\{\boldsymbol x:\mathsf p\cdot\boldsymbol x=1\}$ supports both the ellipsoid and the Bloch ball at $\mathsf p$. Since $Q_B\succ0$, the ellipsoid has a unique outward normal direction at $\mathsf p$, given by $Q_B^{-1}(\mathsf p-\boldsymbol c_B)$. Consequently,
\begin{equation}\label{eq:normal}
Q_B^{-1}(\mathsf p-\boldsymbol c_B)=\eta\mathsf p,
\qquad
\eta=\frac{1}{1-\mathsf p\cdot\boldsymbol c_B}>0.
\end{equation}
To obtain $\eta$, take the scalar product with $\mathsf p-\boldsymbol c_B$ and use the ellipsoid boundary equation, which gives $1=\eta(1-\mathsf p\cdot\boldsymbol c_B)$. The center $\boldsymbol c_B$ is an interior point of $\mathcal E\subseteq\mathcal B$, so $|\boldsymbol c_B|<1$ and hence $\eta>0$.

Choose a unit tangent vector $\hat{\boldsymbol t}$ with $\hat{\boldsymbol t}\cdot\mathsf p=0$ and define, for $k>0$,
\begin{equation}\label{eq:direction}
\hat{\boldsymbol d}(k)
=\frac{\hat{\boldsymbol t}-k\mathsf p}{\sqrt{1+k^2}}.
\end{equation}
Consider the ray $\mathsf p+\ell\hat{\boldsymbol d}(k)$, where $\ell\ge0$ is the distance from $\mathsf p$ along the ray. It lies in the fixed plane $\Pi=\mathsf p+\operatorname{span}\{\hat{\boldsymbol t},-\mathsf p\}$. This plane contains the Bloch-ball center $O=\boldsymbol0$, so its intersection with $\mathcal S$ is a unit circle. It need not contain $\boldsymbol b$ or $\boldsymbol c_B$; the section ellipse need not be centered at $\boldsymbol c_B$ or have its principal axes aligned with $\hat{\boldsymbol t}$ and $-\mathsf p$. 
In local coordinates on $\Pi$, with origin at $\mathsf p$ and axes along $\hat{\boldsymbol t}$ and $-\mathsf p$, the ray has slope $k$. These local coordinates only explain the slope: all Bloch vectors remain expressed in the original coordinates centered at $O$, with $\mathsf p=(0,0,-1)$.

Denote the second intersections of the ray with $\mathcal S$ and $\partial\mathcal E$ by $\mathsf q(k)$ and $\mathsf e(k)$, respectively. They are
\begin{align}
\mathsf q(k)
&=\mathsf p+\ell_{\mathcal B}(k)\hat{\boldsymbol d}(k),
&
\ell_{\mathcal B}(k)
&=\frac{2k}{\sqrt{1+k^2}},
\label{eq:sphere-chord}\\
\mathsf e(k)
&=\mathsf p+\ell_{\mathcal E}(k)\hat{\boldsymbol d}(k),
&
\ell_{\mathcal E}(k)
&=\frac{2\eta k}
{\sqrt{1+k^2}\,
\hat{\boldsymbol d}(k)^\top Q_B^{-1}\hat{\boldsymbol d}(k)}.
\label{eq:ellipsoid-chord}
\end{align}
Indeed, substituting the ray into the sphere equation gives
\begin{align}
|\mathsf p+\ell\hat{\boldsymbol d}(k)|^2-1
=\ell^2-\frac{2k\ell}{\sqrt{1+k^2}}=0, 
\end{align}
whose nonzero root is $\ell_{\mathcal B}(k)$. Similarly, substitution into the ellipsoid boundary equation, together with Eq.~\eqref{eq:normal}, gives
\begin{equation}
2\eta\ell\,\mathsf p\cdot\hat{\boldsymbol d}(k)
+\ell^2\hat{\boldsymbol d}(k)^\top
Q_B^{-1}\hat{\boldsymbol d}(k)=0.
\end{equation}
Since $\mathsf p\cdot\hat{\boldsymbol d}(k)=-k/\sqrt{1+k^2}$, the roots are $\ell=0$ and $\ell=\ell_{\mathcal E}(k)>0$. For $0<\ell<\ell_{\mathcal E}(k)$, the ray lies strictly inside the ellipsoid and exits at $\mathsf e(k)\ne\mathsf p$.

Because $\mathcal E_B\subseteq\mathcal B$, $0<\ell_{\mathcal E}(k)\le\ell_{\mathcal B}(k)$. Moreover,
\begin{equation}\label{eq:chord-ratio}
\frac{|\mathsf e(k)-\mathsf p|}{|\mathsf q(k)-\mathsf p|}
=\frac{\eta}
{\hat{\boldsymbol d}(k)^\top Q_B^{-1}\hat{\boldsymbol d}(k)}
\ge\frac{\eta}{\lambda_{\max}(Q_B^{-1})}>0,
\end{equation}
where $\lambda_{\max}(Q_B^{-1})$ is the largest eigenvalue of $Q_B^{-1}$. The inequality follows from $|\hat{\boldsymbol d}(k)|=1$ and provides a positive lower bound independent of $k$. Finally, $\ell_{\mathcal B}(k)\to0$ as $k\to0^+$, so both $\mathsf q(k)$ and $\mathsf e(k)$ approach $\mathsf p$.

\paragraph{The weight required by the LHS model.}
In the fixed plane $\Pi$, the point $\mathsf e(k)$ lies on the boundary of the section ellipse and hence on $\partial\mathcal E_B$ for every $k>0$. Since $\mathsf e(k)\in\partial\mathcal E$, there exists a projective measurement direction $\hat{\boldsymbol n}_{\mathsf e}(k)$ whose $+$ outcome prepares $\rho_{\mathsf e(k)}$ with probability $p_{\mathsf e}(k)$. Eq.~\eqref{eq:pmin} gives $p_{\mathsf e}(k)\ge p_{\min}>0$. 
Although the geometric construction applies to all $k>0$, the contradiction requires only the settings corresponding to the sequence $k_n=1/n$, $n\ge1$, together with the original pure-state setting $\hat{\boldsymbol n}_{\mathsf p}$. 
This family depends only on the shared state and the chosen tangent direction, not on a hypothetical LHS model. From now on, $k$ ranges over this sequence and every limit as $k\to0^+$ is understood along it. The union of the exceptional $\omega$-null sets for this countable family has measure zero; redefining the responses there, for example to
$p(+|\hat{\boldsymbol n}_s,\boldsymbol\lambda)=1$ and $p(-|\hat{\boldsymbol n}_s,\boldsymbol\lambda)=0$, leaves all
assemblage elements unchanged and makes their probability bounds and normalization hold simultaneously. Since $\mu\le\omega$, these changes also leave every integral against $\mu$ unchanged.

For compactness, we suppress the argument $k$ where unambiguous. The LHS decomposition is
\begin{equation}\label{eq:sequence-lhs}
\sigma_{\mathsf e}=p_{\mathsf e}\rho_{\mathsf e}
=p_{\mathsf p}'\rho_{\mathsf p}
+\int_{\mathcal B}
 p(+|\hat{\boldsymbol n}_{\mathsf e},\boldsymbol\lambda)
 \rho_{\boldsymbol\lambda}\,\dd\mu(\boldsymbol\lambda),
\end{equation}
where
\begin{equation}\label{conditions1}
p_{\mathsf p}'=
\tilde p_{\mathsf p}
 p(+|\hat{\boldsymbol n}_{\mathsf e},\mathsf p),
\qquad 0\le p_{\mathsf p}'\le\tilde p_{\mathsf p}.
\end{equation}
The weight $\tilde p_{\mathsf p}>0$ is forced by the pure outcome, but its contribution $p_{\mathsf p}'$ to this different setting may vanish. We allow $p_{\mathsf p}'=0$ throughout. Define the total weight of the remaining contribution by
\begin{equation}\label{pointh}
p_{\mathsf h}
=\int_{\mathcal B}
 p(+|\hat{\boldsymbol n}_{\mathsf e},\boldsymbol\lambda)
 \,\dd\mu(\boldsymbol\lambda).
\end{equation}
Here $p_{\mathsf h}>0$. Otherwise the residual matrix integral, which is positive semidefinite with trace $p_{\mathsf h}$, would vanish, and taking its trace would give $p_{\mathsf p}'=p_{\mathsf e}>0$. The same equation would then imply $\mathsf e=\mathsf p$, contrary to Eq.~\eqref{eq:ellipsoid-chord}. The normalized average Bloch vector of the residual contribution is therefore well defined:
\begin{equation}\label{eq:mean-hidden-state}
\mathsf h
=\frac{1}{p_{\mathsf h}}\int_{\mathcal B}
 p(+|\hat{\boldsymbol n}_{\mathsf e},\boldsymbol\lambda)
 \boldsymbol\lambda\,\dd\mu(\boldsymbol\lambda).
\end{equation}
Taking the trace and comparing the Pauli coefficients gives
\begin{align}
p_{\mathsf e}&=p_{\mathsf p}'+p_{\mathsf h},\label{conditions2}\\
p_{\mathsf e}\mathsf e
&=p_{\mathsf p}'\mathsf p+p_{\mathsf h}\mathsf h.
\label{ideal010}
\end{align}
Together, these identities give
\begin{equation}\label{ideal04}
p_{\mathsf h}(\mathsf h-\mathsf p)
=p_{\mathsf e}(\mathsf e-\mathsf p),
\qquad
\frac{p_{\mathsf h}}{p_{\mathsf e}}
=\frac{|\mathsf{pe}|}{|\mathsf{ph}|}.
\end{equation}
Because $p_{\mathsf p}'\ge0$, the point $\mathsf h$ lies on the ray from $\mathsf p$ through $\mathsf e$, at or beyond $\mathsf e$. It is an average of Bloch vectors with normalized nonnegative weights, so $|\mathsf h|\le1$ and it cannot lie beyond $\mathsf q$. This average need not itself be one of the hidden states in the model. Consequently,
\begin{equation}\label{cond3}
\mathsf h\in[\mathsf e,\mathsf q],
\qquad
\frac{p_{\mathsf h}}{p_{\mathsf e}}
=\frac{|\mathsf{pe}|}{|\mathsf{ph}|}
\ge\frac{|\mathsf{pe}|}{|\mathsf{pq}|}.
\end{equation}
The endpoint $\mathsf h=\mathsf e$ corresponds to $p_{\mathsf p}'=0$ and $p_{\mathsf h}=p_{\mathsf e}$, so no strictly positive contribution from $\rho_{\mathsf p}$ is required. Combining Eqs.~\eqref{eq:pmin}, \eqref{eq:chord-ratio}, and
\eqref{cond3}, we obtain the uniform lower bound
\begin{equation}\label{limithk1}
p_{\mathsf h}(k)
\ge p_{\min}\frac{\eta}{\lambda_{\max}(Q_B^{-1})}
=:c_0>0.
\end{equation}
The constant $c_0$ depends on the fixed shared state and the selected tangency point, but not on $k$ or on the hypothetical LHS model. The same segment constraint gives
\begin{equation}\label{eq:h-convergence}
|\mathsf h(k)-\mathsf p|
\le|\mathsf q(k)-\mathsf p|
=\frac{2k}{\sqrt{1+k^2}}\longrightarrow0.
\end{equation}
More precisely, writing
$\mathsf h(k)=\mathsf p+\ell_{\mathsf h}(k)\hat{\boldsymbol d}(k)$,
where $0<\ell_{\mathsf h}(k)\le\ell_{\mathcal B}(k)$, gives
\begin{equation}\label{eq:normal-deficit}
0\le 1-\mathsf p\cdot\mathsf h(k)
=\frac{k\,\ell_{\mathsf h}(k)}{\sqrt{1+k^2}}
\le\frac{2k^2}{1+k^2}.
\end{equation}

The fixed residual measure satisfies $\mu(\{\mathsf p\})=0$. For $\delta>0$, define the measurable set
\begin{equation}\label{eq:cap}
C_\delta=\{\boldsymbol\lambda\in\mathcal B:
1-\mathsf p\cdot\boldsymbol\lambda<\delta\}.
\end{equation}
This cap depends only on $\mathsf p$ and $\delta$, not on the measurement setting or its response function.
Since the response is at most one,
\begin{equation}\label{eq:inside-cap}
\int_{C_\delta}
 p(+|\hat{\boldsymbol n}_{\mathsf e}(k),\boldsymbol\lambda)\,\dd\mu(\boldsymbol\lambda)
\le\mu(C_\delta).
\end{equation}
On its complement, $1-\mathsf p\cdot\boldsymbol\lambda\ge\delta$. Using nonnegativity and the definitions of $p_{\mathsf h}$ and $\mathsf h$, we obtain
\begin{equation}\label{eq:outside-cap}
\begin{aligned}
&\delta\int_{\mathcal B\setminus C_\delta}
 p(+|\hat{\boldsymbol n}_{\mathsf e}(k),\boldsymbol\lambda)\,\dd\mu(\boldsymbol\lambda)\\
&\le\int_{\mathcal B} p(+|\hat{\boldsymbol n}_{\mathsf e}(k),\boldsymbol\lambda)
 [1-\mathsf p\cdot\boldsymbol\lambda]\,\dd\mu(\boldsymbol\lambda)\\
&=p_{\mathsf h}(k)[1-\mathsf p\cdot\mathsf h(k)].
\end{aligned}
\end{equation}
Adding the contributions from the two regions and using $p_{\mathsf h}(k)\le M$ gives
\begin{equation}\label{eq:cap-bound}
p_{\mathsf h}(k)\le\mu(C_\delta)
+\frac{M[1-\mathsf p\cdot\mathsf h(k)]}{\delta}.
\end{equation}
By Eq.~\eqref{eq:normal-deficit},
\begin{equation}\label{eq:uniform-bound}
p_{\mathsf h}(k)\le\mu(C_\delta)
+\frac{2Mk^2}{(1+k^2)\delta}.
\end{equation}
For fixed $\delta$, taking the limit along $k_n=1/n$ yields
\begin{equation}
\limsup_{n\to\infty}p_{\mathsf h}(k_n)\le\mu(C_\delta).
\end{equation}
The sets $C_{1/j}$ decrease to $\{\mathsf p\}$ as $j\to\infty$, because $\mathsf p\cdot\boldsymbol\lambda=1$ holds only when $\boldsymbol\lambda=\mathsf p$. Continuity from above of the finite measure $\mu$ therefore gives
\begin{equation}\label{eq:cap-prob-limit}
\lim_{j\to\infty}\mu(C_{1/j})
=\mu(\{\mathsf p\})=0.
\end{equation}
Combining the last two equations with $p_{\mathsf h}(k_n)\ge0$ gives
\begin{equation}\label{limithk2}
\lim_{n\to\infty}p_{\mathsf h}(k_n)=0,
\end{equation}
contradicting the uniform lower bound in Eq.~\eqref{limithk1}. Only the selected settings were used, so no common LHS model reproduces even this countable family. In particular, the state is steerable from Alice to Bob under projective measurements.

The estimate~\eqref{eq:uniform-bound} holds for every selected setting, including discontinuous responses. No pointwise convergence of responses or exchange of a limit with an integral is used. Continuity from above follows from countable additivity and finiteness; it requires neither a density nor a smoothness assumption. The hypothetical model is fixed before the limits are taken. Equivalently, first choose $j$ so that $\mu(C_{1/j})<c_0/2$, and then choose $n$ so that the second term in Eq.~\eqref{eq:uniform-bound}, with $\delta=1/j$ and $k=k_n$, is less than $c_0/2$. These indices may depend on the model, whereas the measurement family does not. No convergence rate uniform over all models is required.

\end{proof}

Proof of Theorem~\ref{twopsteering}
\begin{proof}
Let the contact points be $\mathsf p_1$ and $\mathsf p_2$. Equation~\eqref{eq:pure-support} forces a nonzero probability weight at each point in a hypothetical LHS model. Apply the preceding construction with $\mathsf p=\mathsf p_1$ and write
\begin{equation}
\omega=\tilde p_{\mathsf p_1}\delta_{\mathsf p_1}+\mu,
\qquad \mu(\{\mathsf p_1\})=0.
\end{equation}
The weight at $\mathsf p_2$ remains included in $\mu$. Neither the chord estimate nor the weight argument requires uniqueness of the contact point or mixedness of the complementary conditional state. The same estimates therefore give both
\begin{equation}
p_{\mathsf h}(k_n)\ge c_0>0,
\qquad p_{\mathsf h}(k_n)\longrightarrow0,
\end{equation}
a contradiction. In particular, setting $\Delta=1-\mathsf p_1\cdot\mathsf p_2=|\mathsf p_1-\mathsf p_2|^2/2>0$, we have
$\mathsf p_2\notin C_\delta$ for $0<\delta<\Delta$ when the caps are centered at $\mathsf p_1$. Thus the second pure hidden state causes no exception to the limit. The two pure conditional states may arise from different PVMs or from the two outcomes of one PVM; both possibilities are covered.
\end{proof}

\emph{Acknowledgments}. We thank Howard Wiseman and Travis Baker for useful discussions. This work was supported by the National Research Foundation of Korea (RS-2025-00561467) and the Institute for Information \& Communication Technology Promotion (RS-2023-00229524, RS-2025-02304540, RS-2025-25464876, RS-2025-25464616).

\emph{Note added}. After the completion of this work, we became aware of an independent related work by Zhang and Chen, which studies how boundary geometry turns entanglement into steering~\cite{zhang26boundary}. The two works were developed independently and share some results in common.

\bibliographystyle{apsrev4-2}
\bibliography{puresteering-ref.bib}

\clearpage

\title{Supplemental Material: One pure steered state implies Einstein-Podolsky-Rosen steering }

\maketitle
\onecolumngrid

\tableofcontents

\vspace{1.0\baselineskip}

\section{One point of tangency}
\label{app:onetangpoint}

To prove Theorem~\ref{ellippn}, we first establish the following lemma.

\begin{lemma}\label{onetangpoint}
For an arbitrary two-qubit entangled state, Alice's quantum steering ellipsoid $\mathcal{E}_A$ is tangent to her Bloch sphere $\mathcal{S}_A$ at exactly one point if and only if Bob's quantum steering ellipsoid $\mathcal{E}_B$ is tangent to his Bloch sphere $\mathcal{S}_B$ at exactly one point.   
\end{lemma}

For the proof, it suffices to show that, for an arbitrary two-qubit entangled state, the existence of a rank-one measurement effect on subsystem $B$ that steers subsystem $A$ to a unique pure state is equivalent to the existence of a rank-one measurement effect on subsystem $A$ that steers subsystem $B$ to a unique pure state.

\begin{proof}
We prove the forward implication (Alice-side tangency) $\Rightarrow$ (Bob-side tangency); the reverse follows by exchanging $A$ and $B$. 
Let a two-qubit state $\rho_{AB}$ be given and choose a purification $|\Psi\rangle_{ABC}$ such that
\begin{align}
    \rho_{AB}=\Tr_{C}(|\Psi_{ABC}\rangle\langle \Psi_{ABC}|).
\end{align} 
Since $\mathrm{dim}\ \mathcal{H}_A=2$, the Schmidt decomposition of $|\Psi\rangle_{ABC}$ across $A|BC$ involves at most two terms. 
Thus, we may write
\begin{align}
    |\Psi\rangle_{ABC}=\sqrt{\lambda_1}|\alpha_1\rangle_A|B_1\rangle_{BC}
    +\sqrt{\lambda_2}|\alpha_2\rangle_A|B_2\rangle_{BC}
\end{align}
with $\{|\alpha_1\rangle, |\alpha_2\rangle\}$ orthonormal in $\mathcal{H}_A$, $\{|B_1\rangle, |B_2\rangle\}$ orthonormal in $\mathcal{H}_{BC}$, and $\lambda_1,\lambda_2 \geq 0$, $\lambda_1+\lambda_2=1$.

Assume Alice’s QSE $\mathcal{E}_A$ touches her Bloch sphere $\mathcal{S}_A$ at exactly one point. 
This means there exists a rank-one projector
\begin{align}\label{projector1}
    \Pi_B=|b\rangle\langle b|_B
\end{align}
such that the (normalized) conditional state of system $A$ after outcome $\Pi_B$ is a pure state, and this pure conditional state is unique (no other rank-one outcome on $B$ produces a different pure conditional state on $A$). 
Choose the orthonormal basis $\{|b\rangle,|b^\perp\rangle\}$ of $\mathcal{H}_B$, and expand the two $BC$-vectors in this basis:
\begin{subequations}\label{eq:2}
\begin{align}
|B_1\rangle_{BC} &= |b\rangle|u_1\rangle+|b^\perp\rangle|v_1\rangle, \label{eq:2a} \\
|B_2\rangle_{BC} &= |b\rangle|u_2\rangle+|b^\perp\rangle|v_2\rangle, \label{eq:2b}
\end{align}
\end{subequations}
where $|u_i\rangle$ and $|v_i\rangle$ $(i=1,2)$ are unnormalized states. The orthogonality condition $\langle B_1|B_2\rangle$=0 implies that
\begin{align}\label{orthogcond}
   \langle u_1|u_2\rangle+ \langle v_1|v_2\rangle=0.
\end{align}
Then $|\Psi\rangle_{ABC}$ takes the form
\begin{align}\label{Schmidtds}    |\Psi\rangle_{ABC}=\sqrt{\lambda_1}|\alpha_1\rangle_A(|b\rangle|u_1\rangle+|b^\perp\rangle|v_1\rangle)
    +\sqrt{\lambda_2}|\alpha_2\rangle_A(|b\rangle|u_2\rangle+|b^\perp\rangle|v_2\rangle).
\end{align}
In particular, for the vector $|b\rangle_B$,
the unnormalized postmeasurement state on $AC$, 
\begin{align}
    |\psi_1\rangle_{AC}:=(I_A\otimes \langle b|_B\otimes I_C)|\Psi\rangle_{ABC}
\end{align}
is a nonzero product vector. We may therefore write
\begin{align}\label{ACalpha}
    |\psi_1\rangle_{AC}:=\mu|\alpha\rangle_{A}|\chi\rangle_C, \quad \mu\neq 0,
\end{align}
with unit $|\alpha\rangle_{A}$ and nonzero $|\chi\rangle_C$.
By applying $\langle b|_B$ and $\langle b^\perp|_B$ to the state $|\Psi\rangle_{ABC}$ in Eq.~\eqref{Schmidtds}, we obtain
\begin{subequations}\label{phiAC}
\begin{align}
    |\psi_1\rangle_{AC}
    &=\sqrt{\lambda_1}|\alpha_1\rangle|u_1\rangle
    +\sqrt{\lambda_2}|\alpha_2\rangle|u_2\rangle,\\
    |\psi_2\rangle_{AC}
    &=\sqrt{\lambda_1}|\alpha_1\rangle|v_1\rangle
    +\sqrt{\lambda_2}|\alpha_2\rangle|v_2\rangle.
\end{align}
\end{subequations}
Since $|\psi_1\rangle_{AC}$ is a product vector in $A\otimes C$ whose marginal on $A$ is proportional to the unique pure state $|\alpha\rangle_A$,
it follows that the reduced state of Alice associated with $|\psi_2\rangle_{AC}$ must be mixed and satisfy the conditions
\begin{equation}\label{mixcond0}
\begin{aligned}
    &{\lambda_1} > 0,\ {\lambda_2} > 0, \\
    &|v_1\rangle \neq 0,\ |v_2\rangle \neq 0, \\
    &\nexists\, k_{v}\in \mathbb{C}\ \text{s.\ t.}\  |v_1\rangle = k_v |v_2\rangle.
\end{aligned}
\end{equation}
To ensure that the reduced state of Alice associated with $|\psi_1\rangle_{AC}$ remains pure, one needs to consider the following two cases:
\begin{enumerate}
\item \textbf{Case $|u_1\rangle \neq 0$ and $|u_2\rangle = 0$.}\label{twocases1}
In this case, the condition~\eqref{orthogcond} reduces to 
\begin{align}\label{orthogcondr}
    \langle v_1|v_2\rangle=0,
\end{align}
and the state $|\psi_1\rangle_{AC}$ takes the form
\begin{align}\label{psiAC1}
    |\psi_1\rangle_{AC} \propto |\alpha_1\rangle |u_1\rangle.
\end{align}
The corresponding pure steered state on subsystem $A$ is
\begin{align}\label{alpha1}
    |\alpha\rangle_A = |\alpha_1\rangle .
\end{align}
The purification can then be written as
\begin{equation}\label{purification1}
\begin{aligned}
    |\Psi\rangle_{ABC}=   &\sqrt{\lambda_1}|\alpha_1\rangle\left(|b\rangle|u_1\rangle+|b^\perp\rangle|v_1\rangle\right)  +\sqrt{\lambda_2}|\alpha_2\rangle|b^\perp\rangle|v_2\rangle.
\end{aligned}
\end{equation}
Consider an arbitrary nonzero vector $|b'\rangle_B$, distinct from both $|b\rangle$ and $|b^\perp\rangle$, of the form
\begin{align}\label{arbbv}
    |b'\rangle_B = \eta |b\rangle +\gamma |b^\perp\rangle, \quad |\eta|^2+|\gamma|^2=1,\quad \eta,\gamma \neq 0.
\end{align}
Applying $\langle b'|_B$ to the purification in Eq.~\eqref{purification1} yields
\begin{equation}\label{phiAC3}
\begin{aligned}
   |\psi_3\rangle_{AC}=&\sqrt{\lambda_1}|\alpha_1\rangle(\eta^* |u_1\rangle+\gamma^* |v_1\rangle) +\gamma^*\sqrt{\lambda_2}|\alpha_2\rangle |v_2\rangle,
\end{aligned}
\end{equation}
where $\eta^*$ denotes the complex conjugate of $\eta$, and the others are defined similarly.
The requirement that the state can be steered to exactly one pure steered state implies that all other steered states must be mixed. In particular, the following conditions must hold:
\begin{align}
    &\nexists\, k_{1}\in \mathbb{C}\ \text{s.\ t.}\ |u_1\rangle = k_{1} |v_1\rangle,\label{mixcond1:a}\\
    &\nexists\, k_{2}\in \mathbb{C}\ \text{s.\ t.}\ \eta^* |u_1\rangle+\gamma^* |v_1\rangle = k_2 |v_2\rangle,\label{mixcond1:b}
\end{align}
for all $\eta,\gamma$ satisfying Eq.~\eqref{arbbv}.
To see this, suppose that there exists $k_{1}\in \mathbb{C}$, such that $|u_1\rangle = k_{1} |v_1\rangle$. Defining $\langle \tilde{b}|_B \propto \langle b| - k_{1} \langle b^\perp|$, and applying $\langle \tilde{b}|_B$ to the purification~\eqref{purification1}, one finds that the resulting conditional state on Alice’s side reduces to the pure state $|\alpha_2\rangle$. 
Notice that, if $|u_1\rangle=k_{1}|v_1\rangle$, then a straightforward calculation shows that $\rho_{AB}$ is separable.
Similarly, if there exist $\eta,\gamma$ satisfying Eq.~\eqref{arbbv} and $k_{2}\in \mathbb{C}$, such that $\eta^* |u_1\rangle+\gamma^* |v_1\rangle = k_2 |v_2\rangle$, then the postmeasurement state on Alice’s side is proportional to $k_2\sqrt{\lambda_1}|\alpha_1\rangle+\gamma^*\sqrt{\lambda_2}|\alpha_2\rangle$, which is again pure.

\vspace{1.0\baselineskip}

We now show that the condition of steering from Bob to Alice to exactly one pure state implies steering from Alice to Bob to exactly one pure state. 
%in Case~\ref{twocases1}. 
To this end, we take an arbitrary normalized vector
\begin{align}\label{arbvec}
    |a\rangle_A = m |\alpha_1\rangle +n |\alpha_2\rangle,\quad |m|^2+|n|^2=1,
\end{align}
in system $A$ and act with $\langle a |_A$ on the purification $|\Psi\rangle_{ABC}$ in Eq.~\eqref{purification1}, we have
\begin{align}
  |\phi_1\rangle_{BC} 
   =&  m^*\sqrt{\lambda_1}|b\rangle|u_1\rangle +|b^\perp\rangle\left(m^*\sqrt{\lambda_1}|v_1\rangle+n^*\sqrt{\lambda_2}|v_2\rangle\right). \label{phiBC1}
\end{align}
When $m=0$, the corresponding measurement effect is $\Pi_A=|\alpha_2\rangle\langle\alpha_2 |$, and the conditional state on Bob’s side reduces to the pure state  
\begin{align}
    |\beta\rangle_B=|b^\perp\rangle_B,
\end{align}
which is the unique pure state obtainable in this case. 
Note that 
\begin{align}\label{relation}
    \Pi_A|\alpha\rangle_A=0,\quad \Pi_B|\beta\rangle_B=0.
\end{align}

When $m\neq0$, the conditions in Eqs.~\eqref{mixcond0}, ~\eqref{mixcond1:a}, and~\eqref{mixcond1:b} ensure that the corresponding conditional states on Bob’s side are mixed. 
Indeed, suppose, for the sake of contradiction, that there exists a complex number $k_3\in\mathbb{C}$ such that
\begin{equation}\label{existedk3}
\begin{aligned}
  m^*\sqrt{\lambda_1}|u_1\rangle
  = k_3\left(m^*\sqrt{\lambda_1}|v_1\rangle+n^*\sqrt{\lambda_2}|v_2\rangle\right).
\end{aligned}
\end{equation}
Dividing both sides by $m^*\sqrt{\lambda_1}$ and multiplying by $\eta^*$, one obtains
\begin{equation}\label{existedk32}
\begin{aligned}
   \eta^*|u_1\rangle
  = k_3 \eta^*|v_1\rangle+k_3 \eta^*\frac{n^*\sqrt{\lambda_2}}{m^*\sqrt{\lambda_1}}|v_2\rangle.
\end{aligned}
\end{equation}
For any given $k_3\in\mathbb{C}$, there exists a suitable parameter $\eta$ such that
$|\eta|^2+|k_3 \eta|^2=1$,
with the corresponding coefficient $\gamma$ determined by
\begin{align}\label{k3relation}
    \gamma^* =-k_3 \eta^*.
\end{align}
Moreover, for $n\neq 0$, one has 
\begin{align}\label{k2k3relation}
    k_2 = k_3 \eta^*\frac{n^*\sqrt{\lambda_2}}{m^*\sqrt{\lambda_1}},
\end{align}
while for $n=0$, it follows that 
\begin{align}\label{k1k3relation}
    k_1=k_3.
\end{align} 
In either case, conditions~\eqref{mixcond1:a} and~\eqref{mixcond1:b}  are violated, leading to a contradiction. Therefore, no such $k_3$ can exist.

\vspace{1.0\baselineskip}

It should be noted that no entanglement assumption is imposed on the two-qubit state $\rho_{AB}$; thus, $\rho_{AB}$ may be either entangled or separable. 
The conditions~\eqref{mixcond0}, \eqref{mixcond1:a}, and \eqref{mixcond1:b} require Alice to have exactly one pure steered state and one mixed steered state.
However, certain separable states may also satisfy these conditions. Hence, the set of states satisfying these conditions contains all such entangled states, but may also include a subset of separable states.
Nevertheless, the QSEs associated with the entangled and separable cases are distinct. Specifically, the QSEs corresponding to the entangled states are three dimensional, whereas those corresponding to the separable states are one dimensional, i.e., steering needles. This distinction can be understood as follows.
For an two-qubit entangled state, the volumes of the QSEs associated with both Alice and Bob are nonzero~\cite{Jevtic2014,Milne14}.
Moreover, a two-qubit state is separable if and only if its QSE fits inside a tetrahedron that is itself contained in the Bloch sphere~\cite{Jevtic2014}. For separable states, the QSEs can also be two dimensional, corresponding to a steering pancake, one dimensional, corresponding to a steering needle, or trivially zero dimensional~\cite{Jevtic2014}. 
Hence, for a two-qubit state admitting exactly one pure steered state, the QSEs of the separable states and the entangled states considered here can only be one dimensional and three dimensional, respectively.  
This distinction enables us to distinguish the entangled states from the separable states considered here.
Specifically, $\rho_{AB}$ is entangled if and only if there exists a steered state that is not collinear, in the Bloch sphere, with the two steered states obtained from projective measurement effects $|b\rangle\langle b|_B$ and $|b^\perp\rangle\langle b^\perp|_B$. The other cases in the proofs of Lemmas~\ref{onetangpoint} and~\ref{twotangpoints} are analogous and will not be discussed further.

\item \textbf{Case $|u_1\rangle = 0$ and $|u_2\rangle \neq 0$.}\label{twocases11}
In this case, the unique pure steered states on systems $A$ and $B$ are given by
\begin{align}
|\alpha\rangle_A &= |\alpha_2\rangle ,\\
|\beta\rangle_B &= |b^\perp\rangle ,
\end{align}
respectively. The argument then proceeds analogously to that in Case~\ref{twocases1}, and the details are omitted.

\item \textbf{Case $|u_1\rangle \neq 0$ and $|u_2\rangle \neq 0$.}\label{twocases2}
In this case, there exists $k_u \in \mathbb{C}$ such that
\begin{align}\label{ku12}
    |u_1\rangle = k_u |u_2\rangle .
\end{align}
It follows that
\begin{align}
    |\psi_1\rangle_{AC}
= \left( k_u \sqrt{\lambda_1}|\alpha_1\rangle
+ \sqrt{\lambda_2}|\alpha_2\rangle \right) |u_2\rangle .
\end{align}
Accordingly, the corresponding pure steered state on subsystem $A$ is
\begin{align}\label{alpha2}
    |\alpha\rangle_A
    &\propto k_u \sqrt{\lambda_1}|\alpha_1\rangle
    + \sqrt{\lambda_2}|\alpha_2\rangle.
\end{align}
The purification can therefore be written as
\begin{equation}\label{purif2}
\begin{aligned}
    |\Psi\rangle_{ABC}=   &\sqrt{\lambda_1}|\alpha_1\rangle\left(k_u|b\rangle|u_2\rangle+|b^\perp\rangle|v_1\rangle\right)   +\sqrt{\lambda_2}|\alpha_2\rangle\left(|b\rangle|u_2\rangle+|b^\perp\rangle|v_2\rangle\right).
\end{aligned}
\end{equation}
Applying the vector $\langle b'|_B$ defined in Eq.~\eqref{arbbv} 
 to the purification in Eq.~\eqref{purif2} yields
\begin{equation}\label{phiAC4}
\begin{aligned}
   |\psi_4\rangle_{AC}
   =&\sqrt{\lambda_1}|\alpha_1\rangle(\eta^* k_u |u_2\rangle + \gamma^* |v_1\rangle)
   +\sqrt{\lambda_2}|\alpha_2\rangle(\eta^* |u_2\rangle+\gamma^* |v_2\rangle).
\end{aligned}
\end{equation}
The requirement that steering yields exactly one pure conditional state implies that Alice's reduced state obtained from Eq.~\eqref{phiAC4} must be mixed. This holds only if the following conditions are satisfied simultaneously:
\begin{subequations}\label{mixcond2}
\begin{align}
     &\nexists\, g_{1}\in \mathbb{C}\ \text{s.\ t.}\ |u_2\rangle = g_1|v_1\rangle, \label{mixcond2:a}\\
     &\nexists\, g_{2}\in \mathbb{C}\ \text{s.\ t.}\ |u_2\rangle = g_2|v_2\rangle,\label{mixcond2:b} \\
     &\nexists\, g_{3}\in \mathbb{C}\setminus\{k_u\}\ \text{s.\ t.}\ (g_3-k_u) \eta^* |u_2\rangle = \gamma^*( |v_1\rangle - g_3 |v_2\rangle),\label{mixcond2:c}
\end{align}
\end{subequations}
for all $\eta,\gamma$ satisfying Eq.~\eqref{arbbv}.
Indeed, if any one of these conditions were violated, an additional pure steered state could be obtained, contradicting the assumption of uniqueness.
Specifically, if $|u_2\rangle=g_1|v_1\rangle$ for some $g_1\in\mathbb{C}$, then choosing $|b'\rangle = |b\rangle - g_1^*k_u^* |b^\perp\rangle$ and applying $\langle b'|_B$ to the purification steers system $A$ to the pure state $|\alpha_2\rangle$.
If $|u_2\rangle=g_2|v_2\rangle$ for some $g_2\in\mathbb{C}$, then choosing $|b'\rangle = |b\rangle - g_2^* |b^\perp\rangle$ steers system $A$ to the pure state $|\alpha_1\rangle$. Finally, if there exists $g_3\in\mathbb{C}$ such that $(g_3-k_u) \eta^* |u_2\rangle = \gamma^*( |v_1\rangle - g_3 |v_2\rangle)$, then the resulting conditional state of system $A$ is pure and proportional to $g_3\sqrt{\lambda_1}|\alpha_1\rangle+\sqrt{\lambda_2}|\alpha_2\rangle$.

\vspace{1.0\baselineskip}

We now show that the conditions in Eqs.~\eqref{mixcond0} and~\eqref{mixcond2} imply that Alice can steer Bob to exactly one pure conditional state. To this end, consider an arbitrary normalized vector $\langle a |_A$ as defined in Eq.~\eqref{arbvec} on subsystem $A$ and apply it to the purification $|\Psi\rangle_{ABC}$ given in Eq.~\eqref{purif2}. This leads to
\begin{equation}\label{phiBC}
\begin{aligned}
    |\phi_2\rangle_{BC}= 
    &|b\rangle\left(m^*k_u\sqrt{\lambda_1}+n^*\sqrt{\lambda_2}\right)|u_2\rangle   +|b^\perp\rangle\left(m^*\sqrt{\lambda_1}|v_1\rangle+n^*\sqrt{\lambda_2}|v_2\rangle\right).
\end{aligned}
\end{equation}
When $m^*k_u\sqrt{\lambda_1}+n^*\sqrt{\lambda_2}=0$, 
the corresponding measurement effect is 
$\Pi_A = |\alpha^\perp\rangle\langle \alpha^\perp|$, where $|\alpha^\perp\rangle\propto 
    \sqrt{\lambda_2}|\alpha_1\rangle - k_u^* \sqrt{\lambda_1}|\alpha_2\rangle$,
and Bob’s conditional state reduces to the pure state
\begin{align}
    |\beta\rangle_B=|b^\perp\rangle_B.
\end{align}
Similarly, relation~\eqref{relation} follows.
We proceed to prove that this pure state is unique.
On the other hand, when $m^*k_u\sqrt{\lambda_1}+n^*\sqrt{\lambda_2} \neq 0$, then Bob’s conditional state is necessarily mixed, as guaranteed by the conditions~\eqref{mixcond0} and~\eqref{mixcond2}.
Suppose, for the sake of contradiction, that there exists a complex number $g_4\in\mathbb{C}$ such that
\begin{equation}\label{opureg4}
\begin{aligned}
    &\left(m^*k_u\sqrt{\lambda_1}+n^*\sqrt{\lambda_2}\right)|u_2\rangle
    =g_4\left(m^*\sqrt{\lambda_1}|v_1\rangle+n^*\sqrt{\lambda_2}|v_2\rangle\right).
\end{aligned}
\end{equation}
When $m=0$, Eq.~\eqref{opureg4} reduces to 
\begin{align}
    |u_2\rangle = g_4 |v_2\rangle,
\end{align}
which contradicts to the condition~\eqref{mixcond2:b}.
When $n=0$, Eq.~\eqref{opureg4} reduces to 
\begin{align}
   k_u |u_2\rangle = g_4 |v_1\rangle,
\end{align}
which contradicts to the condition~\eqref{mixcond2:a}.
When $m\neq 0$ and $n\neq 0$, dividing both sides of Eq.~\eqref{opureg4} by the nonzero scalar $m^*k_u\sqrt{\lambda_1}+n^*\sqrt{\lambda_2}$ and multiplying $(g_3-k_u) \eta^*$ yields
\begin{align}\label{opureg44}
    (g_3-k_u) \eta^*|u_2\rangle = G_1 |v_1\rangle + G_2 |v_2\rangle,
\end{align}
where
\begin{align}
    G_1 = \frac{(g_3-k_u) \eta^*g_4 m^*\sqrt{\lambda_1}}{m^*k_u\sqrt{\lambda_1}+n^*\sqrt{\lambda_2}},\quad
    G_2 = \frac{(g_3-k_u) \eta^*g_4 n^*\sqrt{\lambda_2}}{m^*k_u\sqrt{\lambda_1}+n^*\sqrt{\lambda_2}}. \notag
\end{align}
We choose $g_3=-\frac{n^*\sqrt{\lambda_2}}{m^*\sqrt{\lambda_1}}$ and define $\gamma^*=G_1$, so that $G_2=-\gamma^* g_3$, and we then obtain
\begin{align}
     (g_3-k_u) \eta^* |u_2\rangle = \gamma^* ( |v_1\rangle - g_3 |v_2\rangle)
\end{align}
in contradiction with the condition~\eqref{mixcond2:c}. Hence, no such $g_4$ exists.
\end{enumerate}

\end{proof}

Notice that the above proof covers all two-qubit entangled states admitting exactly one pure steered state, as well as certain two-qubit separable states with the same property. For these separable states, the corresponding QSEs degenerate into line segments.
It is worth noting, however, that a two-qubit separable state admitting exactly one pure steered state on one side does not necessarily admit exactly one pure steered state on the other side. For example, for the state
\begin{align}
\rho_{AB}=|0\rangle\langle 0|\otimes \frac{I}{2},
\end{align}
Alice’s side admits exactly one pure steered state, whereas Bob’s side admits no pure steered state. The QSE of this state degenerates into a single point rather than a line segment. The requirement in the above proof that there be one pure steered state and one mixed steered state excludes the case in which the QSE degenerates into a single point. The same exclusion also applies to the proof below; hence, we will not repeat an analogous analysis.

\section{Two points of tangency}
\label{app:twotangpoints}

To prove Theorem~\ref{ellippn}, we next establish the following lemma.

\begin{lemma}\label{twotangpoints}
For an arbitrary two-qubit entangled state, Alice's quantum steering ellipsoid $\mathcal{E}_A$ is tangent to her Bloch sphere $\mathcal{S}_A$ at exactly two points if and only if Bob's quantum steering ellipsoid $\mathcal{E}_B$ is tangent to his Bloch sphere $\mathcal{S}_B$ at exactly two points.   
\end{lemma}

\begin{proof}
The proof follows the same strategy as that of Lemma~\ref{onetangpoint}. We therefore only highlight the differences. 
Assume Alice’s QSE $\mathcal{E}_A$ touches Alice's Bloch sphere $\mathcal{S}_A$ at exactly two different points.  Without loss of generality, we consider that one pure conditional state is steered by the rank-one projector 
\begin{align}
    \Pi_B=|b\rangle\langle b|_B
\end{align} 
given in Eq.~\eqref{projector1}, another pure conditional state is steered by the rank-one projector
\begin{align}
    &\Pi_B^\prime =|b^\prime\rangle\langle b^\prime|_B,
\end{align}
where 
\begin{align}\label{arbbv0}
    |b'\rangle_B = \eta |b\rangle +\gamma |b^\perp\rangle, \quad |\eta|^2+|\gamma|^2=1,\quad \gamma \neq 0,
\end{align}
such that no other rank-one outcome on $B$ produces a different pure conditional state on $A$. 
Hence, a purification $|\Psi\rangle_{ABC}$ of the state $\rho_{AB}$ can be chosen as in Eq.~\eqref{Schmidtds}, subject to the condition~\eqref{orthogcond}.
The proof splits into two distinct parts according to whether $\eta = 0$ or $\eta \neq 0$.

\subsection{Analysis for $\eta = 0$.}

Two distinct pure states are obtained from the rank-one measurement effects
\begin{align}
\Pi_B=|b\rangle\langle b|_B,\quad \Pi_B^\perp =|b^\perp\rangle\langle b^\perp|,
\end{align}
corresponding to the orthonormal basis $\{|b\rangle,|b^\perp\rangle\}$ of $\mathcal{H}_B$. Hence, two distinct pure states are obtained from a single projective measurement.
These two pure conditional states are unique in the following sense: no other rank-one measurement outcome on system $B$ can steer system $A$ to a pure conditional state distinct from either of these two states.
In particular, for the vectors $|b\rangle_B$ and $|b^\perp\rangle_B$, the unnormalized postmeasurement states on $AC$ are given by Eq.~\eqref{ACalpha} and
\begin{align}
|\psi_2\rangle_{AC}:=\nu|\alpha^\prime\rangle_{A}|\chi^\prime\rangle_{A}, \quad \nu\neq 0,
\end{align}
where $|\alpha^\prime\rangle_{A}$ is normalized and $|\chi^\prime\rangle_{A}\neq 0$.
Applying $\langle b |_B$ and $\langle b^\perp |_B$ to the state $|\Psi\rangle_{ABC}$~\eqref{Schmidtds} yields the states $|\psi_1\rangle_{AC}$ and $|\psi_2\rangle_{AC}$ given in Eqs.~\eqref{phiAC}, respectively.
The existence of two pure steered states implies that $|\psi_1\rangle_{AC}$ and $|\psi_2\rangle_{AC}$ are single $A\otimes C$ product vector proportional to $|\alpha\rangle_A$, $|\alpha^\prime\rangle_{A}$, respectively. 

If $|u_2\rangle = |v_2\rangle = 0$ or $|u_1\rangle = |v_1\rangle = 0$ in Eq.~\eqref{phiAC}, both steered states reduce to the same pure state, $|\alpha_1\rangle$ or $|\alpha_2\rangle$, respectively, contradicting the assumption that two distinct pure steered states exist. We therefore restrict our attention to the following three cases.
\begin{enumerate}
\item \textbf{Case $|u_2\rangle = 0$ and $|v_2\rangle \neq 0$.}\label{threecases1}
In this case, there exists $k_v \in \mathbb{C}$ such that 
\begin{align}\label{kv12}
   &|v_1\rangle = k_v |v_2\rangle.
\end{align}
Combining Eqs.~\eqref{orthogcond} and~\eqref{kv12} yields $k_v=0$ and $|v_1\rangle = 0$. The states $|\psi_1\rangle_{AC}$ and $|\psi_2\rangle_{AC}$ then take the form
\begin{subequations}
\begin{align}
  |\psi_1\rangle_{AC}
\propto&|\alpha_1\rangle |u_1\rangle,\\
|\psi_2\rangle_{AC}
\propto & |\alpha_2\rangle  |v_2\rangle,
\end{align}
\end{subequations}
with $|u_1\rangle \neq 0$. It follows that the corresponding pure steered states on system $A$ are
\begin{subequations}
\begin{align}
  &|\alpha\rangle_A = |\alpha_1\rangle ,\\
  &|\alpha^\prime\rangle_{A}  =|\alpha_2\rangle .  
\end{align} 
\end{subequations}
The purification $|\Psi\rangle_{ABC}$ can therefore be written as follows
\begin{equation}\label{purification2a}
\begin{aligned}
    |\Psi\rangle_{ABC}
    =\sqrt{\lambda_1}|\alpha_1\rangle|b\rangle|u_1\rangle
    +\sqrt{\lambda_2}|\alpha_2\rangle|b^\perp\rangle|v_2\rangle.
\end{aligned}
\end{equation}
Consider an arbitrary nonzero vector $|b'\rangle_B$ defined in Eq.~\eqref{arbbv}. Applying $\langle b'|_B$ to the purification~\eqref{purification2a} yields
\begin{equation}
\begin{aligned}
   |\psi_3\rangle_{AC}=&\sqrt{\lambda_1}\eta^*|\alpha_1\rangle |u_1\rangle +\sqrt{\lambda_2}\gamma^*|\alpha_2\rangle |v_2\rangle.
\end{aligned}
\end{equation}
The requirement that the state can steer to exactly two distinct pure states implies that all other steered states on Alice’s side must be mixed. In particular, the following conditions must be satisfied:
\begin{equation}\label{p2mixcond}
\begin{aligned}
    & \lambda_1>0,\ \lambda_2>0,\\
    &|u_1\rangle \neq 0,\ |v_2\rangle \neq 0,\\
    &\nexists\, k_{1}\in \mathbb{C}\ \text{s.\ t.}\ |u_1\rangle = k_1 |v_2\rangle.
\end{aligned}
\end{equation}

\vspace{0.5\baselineskip}

We now show that the condition of steering from Bob to Alice to exactly two distinct pure states implies steering from Alice to Bob to exactly two distinct pure states. To this end, consider an arbitrary normalized vector $\langle a |_A$ in system $A$, as given in~\eqref{arbvec}. Acting with $\langle a |_A$ on the purification $|\Psi\rangle_{ABC}$~\eqref{purification2a} yields
\begin{equation}
\begin{aligned}
  |\phi_1\rangle_{BC} =  m^*\sqrt{\lambda_1}|b\rangle|u_1\rangle + n^*\sqrt{\lambda_2}|b^\perp\rangle |v_2\rangle.
\end{aligned}
\end{equation}
When $m\neq0$ and $n\neq0$, the conditions in Eqs.~\eqref{p2mixcond} guarantee that the conditional states on Bob’s side are mixed. In contrast, two subcases arise:
\begin{itemize}
\item[(i)] When $m=0$, 
the corresponding measurement effect is 
$\Pi_A = |\alpha_2\rangle\langle \alpha_2|$, and the conditional state reduces to the pure state
\begin{align}
    |\beta\rangle_B=|b^\perp\rangle_B.
\end{align}
\item[(ii)] When $n=0$, the corresponding measurement effect is 
$\widetilde{\Pi}_A = |\alpha_1\rangle\langle \alpha_1|$, and the conditional state reduces to the pure state
\begin{align}
    |\beta^\prime\rangle_B=|b\rangle_B.
\end{align}
\end{itemize}
Similarly, one obtains the following relations
\begin{equation}\label{relation2}
\begin{aligned}
    &\Pi_A|\alpha\rangle_A=0,\quad \Pi_B|\beta\rangle_B=0,\\
    &\widetilde{\Pi}_A|\alpha^\prime\rangle_A=0,\quad \Pi_B^\perp|\beta^\prime\rangle_B=0.
\end{aligned}
\end{equation}

\item \textbf{Case $|u_2\rangle \neq 0$ and $|v_2\rangle = 0$.}\label{threecases2} 
In this case, the corresponding pure conditional states on subsystem $A$ are
\begin{subequations}
\begin{align}
&|\alpha\rangle_A = |\alpha_2\rangle ,\\
&|\alpha^\prime\rangle_A = |\alpha_1\rangle .
\end{align}
\end{subequations}
The corresponding pure conditional states on subsystem $B$ are
\begin{subequations}
\begin{align}
&|\beta\rangle_B = |b\rangle,\\
&|\beta^\prime\rangle_B = |b^\perp\rangle .
\end{align}
\end{subequations}
The argument proceeds analogously to that in Case~\ref{threecases1}, and the details are omitted.

\item \textbf{Case $|u_2\rangle \neq 0$ and $|v_2\rangle \neq 0$.}\label{threecases3} 
In this case, there exist $k_u,k_v \in \mathbb{C}$ such that Eqs.~\eqref{ku12} and~\eqref{kv12} are satisfied. The states $|\psi_1\rangle_{AC}$ and $|\psi_2\rangle_{AC}$ can then be written as
\begin{subequations}
\begin{align}
|\psi_1\rangle_{AC} =& \left( k_u \sqrt{\lambda_1}|\alpha_1\rangle + \sqrt{\lambda_2}|\alpha_2\rangle\right) |u_2\rangle ,\\
|\psi_2\rangle_{AC}
=& \left( k_v \sqrt{\lambda_1}|\alpha_1\rangle + \sqrt{\lambda_2}|\alpha_2\rangle\right) |v_2\rangle.
\end{align}
\end{subequations}
Accordingly, the corresponding pure steered states on system $A$ are
\begin{subequations}
\begin{align}
  &|\alpha\rangle_A  \propto k_u \sqrt{\lambda_1}|\alpha_1\rangle + \sqrt{\lambda_2}|\alpha_2\rangle ,\\
  &|\alpha^\prime\rangle_{A}   \propto k_v \sqrt{\lambda_1}|\alpha_1\rangle + \sqrt{\lambda_2}|\alpha_2\rangle.
\end{align}
\end{subequations}
The requirement that these two steered states be distinct implies
\begin{align}\label{deffk}
    k_u \neq k_v.
\end{align}
The purification can therefore be expressed as follows
\begin{equation}\label{purif2c}
\begin{aligned}
    |\Psi\rangle_{ABC}=   &\sqrt{\lambda_1}|\alpha_1\rangle\left(k_u|b\rangle|u_2\rangle+k_v|b^\perp\rangle|v_2\rangle\right)
    +\sqrt{\lambda_2}|\alpha_2\rangle\left(|b\rangle|u_2\rangle+|b^\perp\rangle|v_2\rangle\right).
\end{aligned}
\end{equation}
Applying the vector $\langle b'|_B$ defined in Eq.~\eqref{arbbv} 
 to the purification~\eqref{purif2c} yields
\begin{equation}
\begin{aligned}
   |\psi_4\rangle_{AC}
   =&\sqrt{\lambda_1}|\alpha_1\rangle(\eta^* k_u |u_2\rangle + \gamma^* k_v |v_2\rangle)
   +\sqrt{\lambda_2}|\alpha_2\rangle(\eta^* |u_2\rangle+\gamma^* |v_2\rangle).
\end{aligned}
\end{equation}
The existence of exactly two distinct pure steered states requires that all other steered states on Alice’s side be mixed. Specifically, the following conditions must be satisfied:
\begin{equation}\label{nonlinearuv2}
\begin{aligned}
    & \lambda_1>0,\ \lambda_2>0,\\
    &|u_2\rangle \neq 0,\ |v_2\rangle \neq 0,\\
    &\nexists\, k_{2}\in \mathbb{C}\ \text{s.\ t.}\ |u_2\rangle = k_2 |v_2\rangle.
\end{aligned}
\end{equation}

\vspace{0.5\baselineskip}

We now show that the conditions~\eqref{deffk} and~\eqref{nonlinearuv2} imply steering from Alice to Bob to exactly two distinct pure states. To this end, consider an arbitrary normalized vector $\langle a |_A$ as given in Eq.~\eqref{arbvec}. Applying $\langle a |_A$ to the purification $|\Psi\rangle_{ABC}$ in Eq.~\eqref{purif2c} yields
\begin{equation}\label{phiBC2c}
\begin{aligned}
    |\phi_2\rangle_{BC}= 
    &|b\rangle\left(m^*k_u\sqrt{\lambda_1}+n^*\sqrt{\lambda_2}\right)|u_2\rangle
    +|b^\perp\rangle\left(m^*k_v\sqrt{\lambda_1}+n^*\sqrt{\lambda_2}\right)|v_2\rangle.
\end{aligned}
\end{equation}
When $m^*k_u\sqrt{\lambda_1}+n^*\sqrt{\lambda_2} \neq 0$ and $m^*k_v\sqrt{\lambda_1}+n^*\sqrt{\lambda_2} \neq 0$, Eqs.~\eqref{deffk} and~\eqref{nonlinearuv2} ensure that the corresponding conditional states on Bob’s side are mixed. 
In contrast, there are two subcases to consider:
\begin{itemize}
\item[(i)] When $m^*k_u\sqrt{\lambda_1}+n^*\sqrt{\lambda_2}=0$, 
The relevant measurement effect can be chosen as
$\Pi_A = |\alpha^\perp_u\rangle\langle \alpha^\perp_u|$, where $|\alpha^\perp_u\rangle\propto  \sqrt{\lambda_2}|\alpha_1\rangle - k_u^* \sqrt{\lambda_1}|\alpha_2\rangle$, and
the conditional state reduces to the pure state  
\begin{align}
    |\beta\rangle_B=|b^\perp\rangle_B.
\end{align}
\item[(ii)] When $m^*k_v\sqrt{\lambda_1}+n^*\sqrt{\lambda_2}=0$, 
the corresponding measurement effect is given by
$\widetilde{\Pi}_A = |\alpha^\perp_v\rangle\langle \alpha^\perp_v|$, where $|\alpha^\perp_v\rangle\propto  \sqrt{\lambda_2}|\alpha_1\rangle - k_v^* \sqrt{\lambda_1}|\alpha_2\rangle$, and
the conditional state reduces to the pure state
\begin{align}
    |\beta^\prime\rangle_B=|b\rangle_B.
\end{align}
\end{itemize}
Similarly, relations~\eqref{relation2} follow.
\end{enumerate}

\subsection{Analysis for $\eta \neq 0$.}

Two distinct pure states are obtained from the rank-one measurement effects
\begin{align}
\Pi_B=|b\rangle\langle b|_B,\quad
\Pi_B^\prime =|b^\prime\rangle\langle b^\prime|_B,
\end{align}
where $|b'\rangle_B$ is defined in Eq.~\eqref{arbbv} as
$|b'\rangle_B = \eta |b\rangle + \gamma |b^\perp\rangle$,
with $|\eta|^2+|\gamma|^2=1$ and $\eta,\gamma \neq 0$.
As in the proof of Lemma~\ref{onetangpoint}, applying $\langle b |_B$ and $\langle b^\perp |_B$ to the state $|\Psi\rangle_{ABC}$ in Eq.~\eqref{Schmidtds} yields the states $|\psi_1\rangle_{AC}$ and $|\psi_2\rangle_{AC}$ given in Eqs.~\eqref{phiAC}, respectively. Since $\Pi_B^\perp$ steers Alice's system to a mixed state, $|\psi_2\rangle_{AC}$ must satisfy the conditions~\eqref{mixcond0}. As $\Pi_B$ steers Alice's system to a pure state, three cases arise, corresponding to Case~\ref{twocases1}, ~\ref{twocases11}, and Case~\ref{twocases2} considered therein.
\begin{enumerate}
\item \textbf{Case $|u_1\rangle \neq 0$ and $|u_2\rangle = 0$.}\label{twopthree1}
In this case, the purification $|\Psi\rangle_{ABC}$ is given in Eq.~\eqref{purification1}, subject to condition~\eqref{orthogcondr}, and the state $|\psi_1\rangle_{AC}$ is given in Eq.~\eqref{psiAC1}. Hence, the pure steered state on system $A$ corresponding to $\Pi_B$ is
\begin{align}
    |\alpha\rangle_A = |\alpha_1\rangle .
\end{align}
Applying an arbitrary nonzero vector $\langle b'|_B$ given in Eq.~\eqref{arbbv} to the purification $|\Psi\rangle_{ABC}$~\eqref{purification1} yields the state $|\psi_3\rangle_{AC}$ of Eq.~\eqref{phiAC3}. 
The requirement that Alice's system can be steered to exactly two pure conditional states implies either
\begin{align}\label{twopurecond1}
    &\exists\, k_{1}\in \mathbb{C}\ \text{s.t.}\ 
    |u_1\rangle = k_{1} |v_1\rangle,
\end{align}
or that there exist $\eta,\gamma$ satisfying Eq.~\eqref{arbbv} such that
\begin{align}\label{twopurecond2}
    &\exists\, k_{2}\in \mathbb{C}\ \text{s.t.}\ 
    \eta^* |u_1\rangle+\gamma^* |v_1\rangle = k_2 |v_2\rangle.
\end{align}
Note that conditions~\eqref{mixcond0} and~\eqref{orthogcondr} ensure that the above two conditions cannot be satisfied simultaneously.
In fact, $\Pi_B$ and $\Pi_B^\perp$ steer Alice's system to a pure state and a mixed state, respectively. This implies that the two-qubit state $\rho_{AB}$ is mixed. Since a two-qubit mixed state admits at most two pure conditional states, there can be at most one additional pure conditional state. Accordingly, at most one of the above two conditions can be satisfied.
We now examine the two conditions separately.
Recall that $|b'\rangle_B = \eta |b\rangle + \gamma |b^\perp\rangle$, 
where $|\eta|^2 + |\gamma|^2 = 1$ and $\eta, \gamma \neq 0$. 
\begin{itemize}
\item[(i)] If the condition~\eqref{twopurecond1} is satisfied, one can choose suitable $\eta$ and $\gamma$ such that
\begin{align}
    k_1=-\frac{\gamma^*}{\eta^*},
\end{align} 
the corresponding measurement effect is given by
$\widetilde{\Pi}_B = |\tilde{b}\rangle\langle \tilde{b} |$, where $|\tilde{b}\rangle \propto  |b\rangle - k_1^* |b^\perp\rangle$, and
the state of system $A$ reduces to the second pure state
\begin{align}
    |\alpha^\prime\rangle_{A} = |\alpha_2\rangle.
\end{align}
\item[(ii)] If the condition~\eqref{twopurecond2} is satisfied, according to Eq.~\eqref{k3relation}, one can choose
\begin{align}
    k_3=-\frac{\gamma^*}{\eta^*},
\end{align}
so that 
the relevant measurement effect can be chosen as 
$\overline{\Pi}_B = |\bar{b}\rangle\langle \bar{b} |$, where $|\bar{b}\rangle \propto  |b\rangle - k_3^* |b^\perp\rangle$, the relation between $k_2$ and $k_3$ is given in Eq.~\eqref{k2k3relation}, and
the state of system $A$ then reduces to the second pure state
\begin{align}
    |\alpha^{\prime\prime}\rangle_{A} \propto k_2\sqrt{\lambda_1}|\alpha_1\rangle + \gamma^*\sqrt{\lambda_2}|\alpha_2\rangle.
\end{align}
\end{itemize}

\vspace{0.5\baselineskip}

We now show that if Bob can steer Alice to exactly two pure states, then Alice can steer Bob to exactly two pure states. Consider an arbitrary normalized vector $\langle a|_A$ in Eq.~\eqref{arbvec}. Applying $\langle a|_A$ to $|\Psi\rangle_{ABC}$ in Eq.~\eqref{purification1} yields $|\phi_1\rangle_{BC}$ in Eq.~\eqref{phiBC1}. 
Recall that $|a\rangle_A = m |\alpha_1\rangle +n |\alpha_2\rangle$ where $|m|^2+|n|^2=1$.
When $m=0$, the corresponding measurement effect is $\Pi_A = | \alpha_2 \rangle\langle \alpha_2 |$,
and the conditional state on Bob’s side reduces to the pure state
\begin{align}
    |\beta\rangle_B=|b^\perp\rangle_B.
\end{align}
When $m\neq 0$, two subcases arise:

\begin{itemize}
\item[(i)] If condition~\eqref{twopurecond1} holds, then for $n=0$, 
the corresponding measurement effect is given by $\widetilde{\Pi}_A = |\alpha_1\rangle\langle\alpha_1|$,
and Bob's conditional state reduces to
\begin{align}
    |\beta^\prime\rangle_B \propto k_1|b\rangle+|b^\perp\rangle.
\end{align}

\item[(ii)] If condition~\eqref{twopurecond2} holds, one can choose 
$m$ and $n$ such that a complex number $k_3 \in\mathbb{C}$ exists 
satisfying Eqs.~\eqref{existedk3} and~\eqref{existedk32}. The corresponding measurement effect is $\overline{\Pi}_A = |\bar{a} \rangle\langle\bar{a}|$, where $|\bar{a} \rangle \propto \gamma\sqrt{\lambda_2}|\alpha_1\rangle - k_2^*\sqrt{\lambda_1}|\alpha_2\rangle$,
and the conditional state then reduces to
\begin{align}
    |\beta^{\prime\prime}\rangle_B \propto k_3|b\rangle+|b^\perp\rangle.
\end{align}
\end{itemize}
Similarly, one obtains the following relations
\begin{equation}\label{relation3}
\begin{aligned}
    &\Pi_A|\alpha\rangle_A=0,\quad \Pi_B|\beta\rangle_B=0,\\
    &\widetilde{\Pi}_A|\alpha^\prime\rangle_A=0,\quad \widetilde{\Pi}_B|\beta^\prime\rangle_B=0,\\
    &\overline{\Pi}_A|\alpha^{\prime\prime}\rangle_A=0,\quad \overline{\Pi}_B|\beta^{\prime\prime}\rangle_B=0.
\end{aligned}
\end{equation}

\item \textbf{Case $|u_1\rangle = 0$ and $|u_2\rangle \neq 0$.}
In this case, the pure steered state on system $A$ corresponding to $\Pi_B$ is
\begin{align}
    |\alpha\rangle_A = |\alpha_2\rangle .
\end{align}
The remainder of the argument proceeds analogously to that of 
Case~\ref{twopthree1}, and we therefore omit the details. 
This case is also a special instance of the following case 
with $k_u=0$.

\item \textbf{Case $|u_1\rangle \neq 0$ and $|u_2\rangle \neq 0$.} In this case, there exists $k_u \in \mathbb{C}$ such that
\begin{align}
    |u_1\rangle = k_u |u_2\rangle .
\end{align}
Eqs.~(\ref{ku12}–\ref{phiAC4}) still hold. In particular, Alice's pure conditional state takes the form
\begin{align}
   |\alpha\rangle_A \propto 
   k_u \sqrt{\lambda_1}|\alpha_1\rangle
   + \sqrt{\lambda_2}|\alpha_2\rangle.
\end{align}
The corresponding measurement effect $\Pi_B$ is given in Eq.~\eqref{alpha2}. The purification $|\Psi\rangle_{ABC}$ is given in Eq.~\eqref{purif2}, and applying an arbitrary vector $\langle b'|_B$ defined in Eq.~\eqref{arbbv} to the purification $|\Psi\rangle_{ABC}$ in Eq.~\eqref{purif2} yields $|\psi_4\rangle_{AC}$ in Eq.~\eqref{phiAC4}. 
The requirement that Alice's system can be steered to exactly two different pure conditional states implies that either
\begin{align}\label{purecond2:a}
     &\exists\, g_{1}\in \mathbb{C}\ \text{s.\ t.}\ 
     |u_2\rangle = g_1|v_1\rangle, 
\end{align}
or
\begin{align}\label{purecond2:b}
     &\exists\, g_{2}\in \mathbb{C}\ \text{s.\ t.}\ 
     |u_2\rangle = g_2|v_2\rangle,
\end{align}
or there exist $\eta,\gamma$ satisfying Eq.~\eqref{arbbv} such that
\begin{align}\label{purecond2:c}
     &\exists\, g_{3}\in \mathbb{C}\setminus\{k_u\}\ \text{s.\ t.}\ 
     (g_3-k_u)\eta^* |u_2\rangle
     = \gamma^*( |v_1\rangle - g_3 |v_2\rangle).
\end{align}

There are three possible subcases:
\begin{itemize}
\item[(i)] If condition~\eqref{purecond2:a} holds, 
one can choose suitable 
$\eta$ and $\gamma$ such that
\begin{align}
    g_1=-\frac{\gamma^*}{\eta^* k_u},
\end{align}
so that 
the required measurement effect is therefore
$\widetilde{\Pi}_B = |\tilde{b}\rangle\langle \tilde{b} |$, where $|\tilde{b}\rangle \propto  |b\rangle - k_u^* g_1^* |b^\perp\rangle$, and
system $A$ reduces to the second pure state
\begin{align}
    |\alpha^\prime\rangle_{A} = |\alpha_2\rangle.
\end{align}

\item[(ii)] If condition~\eqref{purecond2:b} holds, one can choose suitable 
$\eta$ and $\gamma$ such that
\begin{align}
    g_2=-\frac{\gamma^*}{\eta^*},
\end{align}
so that 
the corresponding measurement effect is 
$\overline{\Pi}_B = |\overline{b}\rangle\langle \overline{b} |$, where $|\overline{b}\rangle \propto  |b\rangle -  g_2^* |b^\perp\rangle$, and
system $A$ reduces to the second pure state
\begin{align}
    |\alpha^{\prime\prime}\rangle_{A} = |\alpha_1\rangle.
\end{align}

\item[(iii)] If condition~\eqref{purecond2:c} holds, 
the relevant measurement effect can be chosen as 
$\widehat{\Pi}_B = |\hat{b}\rangle\langle \hat{b} |$, where $|\hat{b}\rangle \propto  |b\rangle -  g_4^* |b^\perp\rangle$, and
system $A$ directly reduces 
to the second pure state
\begin{align}
    |\alpha^{\prime\prime\prime}\rangle_{A} 
    \propto g_3\sqrt{\lambda_1}|\alpha_1\rangle
    +\sqrt{\lambda_2}|\alpha_2\rangle,
    \quad g_3\neq k_u.
\end{align}
\end{itemize}

We now show that if Bob can steer Alice to exactly two pure states, 
then Alice can steer Bob to exactly two pure states. 
Consider an arbitrary normalized vector $\langle a|_A$ 
in Eq.~\eqref{arbvec}. Applying $\langle a|_A$ to 
$|\Psi\rangle_{ABC}$ in Eq.~\eqref{purif2} yields 
$|\phi_2\rangle_{BC}$ in Eq.~\eqref{phiBC}.

When $m^*k_u\sqrt{\lambda_1}+n^*\sqrt{\lambda_2}=0$, 
the corresponding measurement effect is given by
$\Pi_A = |\check{a}\rangle\langle \check{a}|$, where $|\check{a}\rangle\propto  \sqrt{\lambda_2}|\alpha_1\rangle - k_u^* \sqrt{\lambda_1}|\alpha_2\rangle$, and
the conditional state on Bob’s side reduces to
\begin{align}
    |\beta\rangle_B=|b^\perp\rangle_B.
\end{align}

When $m^*k_u\sqrt{\lambda_1}+n^*\sqrt{\lambda_2}\neq 0$, 
three subcases arise:
\begin{itemize}
\item[(i)] If condition~\eqref{purecond2:a} holds, then for $n=0$, the corresponding measurement effect is given by $\widetilde{\Pi}_A = |\alpha_1\rangle\langle\alpha_1|$,
and
Bob’s conditional state reduces to
\begin{align}
    |\beta^\prime\rangle_B 
    \propto k_u g_1 |b\rangle + |b^\perp\rangle.
\end{align}

\item[(ii)] If condition~\eqref{purecond2:b} holds, then for $m=0$, the relevant measurement effect can be chosen as $\overline{\Pi}_A = |\alpha_2\rangle\langle\alpha_2|$, and
Bob’s conditional state reduces to
\begin{align}
    |\beta^{\prime\prime}\rangle_B 
    \propto g_2 |b\rangle + |b^\perp\rangle.
\end{align}

\item[(iii)] If condition~\eqref{purecond2:c} holds, one can choose 
$m$ and $n$ in Eq.~\eqref{arbvec}. 
This condition guarantees the existence of a parameter $g_4$ 
satisfying Eqs.~\eqref{opureg4} and~\eqref{opureg44}. The corresponding measurement effect is $\widehat{\Pi}_A = |\hat{a} \rangle\langle\hat{a}|$, where $|\hat{a} \rangle \propto \sqrt{\lambda_2}|\alpha_1\rangle - g_3^*\sqrt{\lambda_1}|\alpha_2\rangle$,
and Bob’s conditional state reduces to
\begin{align}
    |\beta^{\prime\prime\prime}\rangle_B 
    \propto g_4 |b\rangle+|b^\perp\rangle.
\end{align}
\end{itemize}
Similarly, one obtains the following relations
\begin{equation}\label{relation4}
\begin{aligned}
    &\Pi_A|\alpha\rangle_A=0,\quad \Pi_B|\beta\rangle_B=0,\\
    &\widetilde{\Pi}_A|\alpha^\prime\rangle_A=0,\quad \widetilde{\Pi}_B|\beta^\prime\rangle_B=0,\\
    &\overline{\Pi}_A|\alpha^{\prime\prime}\rangle_A=0,\quad \overline{\Pi}_B|\beta^{\prime\prime}\rangle_B=0,\\
    &\widehat{\Pi}_A|\alpha^{\prime\prime\prime}\rangle_A=0,\quad \widehat{\Pi}_B|\beta^{\prime\prime\prime}\rangle_B=0.
\end{aligned}
\end{equation}
\end{enumerate}

\end{proof}

\section{Correspondence between pure steered states}
\label{app:pure-state-correspondence}

A pure steered state is a normalized pure conditional state prepared with nonzero probability by a local rank-one projective measurement on the other subsystem. Pure states are identified up to an overall phase. For any normalized qubit vector $|\alpha\rangle$, let $|\alpha^\perp\rangle$ denote a normalized orthogonal complement;
its rank-one projector is unique.

\medskip
\noindent\textbf{Theorem (Correspondence between pure steered states).}
\textit{Let $\rho_{AB}$ be an entangled two-qubit state admitting at least one pure steered state on either subsystem. The sets of pure steered states on Alice's and Bob's subsystems are in one-to-one correspondence. The correspondence pairs $|\alpha\rangle_A$ with $|\beta\rangle_B$ precisely when either, and hence both, of the following equivalent relations hold:}
\begin{align}
 {}_A\langle\alpha^\perp|\rho_{AB}|\alpha^\perp\rangle_A
 &=p\,|\beta\rangle\langle\beta|_B,
 \label{eq:psc-forward}\\
 {}_B\langle\beta^\perp|\rho_{AB}|\beta^\perp\rangle_B
 &=q\,|\alpha\rangle\langle\alpha|_A,
 \label{eq:psc-reverse}
\end{align}
\textit{where}
\begin{equation}
 p=\langle\alpha^\perp|\rho_A|\alpha^\perp\rangle>0,
 \qquad
 q=\langle\beta^\perp|\rho_B|\beta^\perp\rangle>0,
 \label{eq:psc-probabilities}
\end{equation}
\textit{with $\rho_A=\operatorname{Tr}_B\rho_{AB}$ and $\rho_B=\operatorname{Tr}_A\rho_{AB}$. Thus Alice prepares $|\beta\rangle_B$ with the measurement effect $|\alpha^\perp\rangle\langle\alpha^\perp|_A$, and Bob prepares $|\alpha\rangle_A$ with the measurement effect $|\beta^\perp\rangle\langle\beta^\perp|_B$.}

\begin{proof}
Both reduced states are full rank. Indeed, a rank-deficient qubit state is pure, and a bipartite state with a pure reduced state is a product state, contrary to the entanglement assumption. Consequently, the probabilities in Eq.~\eqref{eq:psc-probabilities} are strictly positive for all normalized $|\alpha\rangle_A$ and $|\beta\rangle_B$.

Suppose that Eq.~\eqref{eq:psc-forward} holds. Taking the expectation value in $|\beta^\perp\rangle_B$ yields
\begin{equation}
 \langle\alpha^\perp,\beta^\perp|\rho_{AB}
 |\alpha^\perp,\beta^\perp\rangle=0.
\end{equation}
For any positive semidefinite operator $R$, $\langle v|R|v\rangle=\|R^{1/2}|v\rangle\|^2$; hence a vanishing expectation value implies $R|v\rangle=0$. It follows that
\begin{equation}
 \rho_{AB}|\alpha^\perp,\beta^\perp\rangle=0.
 \label{eq:psc-kernel}
\end{equation}
Consider Alice's unnormalized conditional state
\begin{equation}
 \tau_A={}_B\langle\beta^\perp|\rho_{AB}|\beta^\perp\rangle_B
 \succeq0.
\end{equation}
Equation~\eqref{eq:psc-kernel} gives $\tau_A|\alpha^\perp\rangle=0$. Since Alice's Hilbert space is two-dimensional, $\tau_A$ is supported on $\operatorname{span}\{|\alpha\rangle_A\}$. Its trace is $q>0$, so $\tau_A=q|\alpha\rangle\langle\alpha|_A$, proving Eq.~\eqref{eq:psc-reverse}. Exchanging Alice and Bob proves the
converse. Thus Eqs.~\eqref{eq:psc-forward} and \eqref{eq:psc-reverse} are equivalent.

We now show that this relation defines a bijection. For every pure steered state $|\beta\rangle_B$, there is a rank-one projective measurement effect that prepares it. Write this effect as $|\alpha^\perp\rangle\langle\alpha^\perp|_A$.
The equivalence just proved guarantees that $|\alpha\rangle_A$ is itself a pure steered state. Conversely, every pure steered state on Alice's subsystem has a partner on Bob's subsystem by the same argument with the parties exchanged.

The partner of each state is unique. For a fixed $|\beta\rangle_B$, Eq.~\eqref{eq:psc-reverse} fixes
\begin{equation}
 |\alpha\rangle\langle\alpha|_A
 =\frac{{}_B\langle\beta^\perp|\rho_{AB}
 |\beta^\perp\rangle_B}
 {\langle\beta^\perp|\rho_B|\beta^\perp\rangle}.
 \label{eq:psc-unique-partner}
\end{equation}
The denominator is nonzero, and the expression is independent of the phase chosen for $|\beta^\perp\rangle_B$.
Similarly, Eq.~\eqref{eq:psc-forward} fixes a unique $|\beta\rangle\langle\beta|_B$ for each $|\alpha\rangle_A$. Every state on either side therefore has exactly one partner on the other side, establishing the claimed bijection.
\end{proof}

\section{Applications}

\subsection{X-states with $N$ pure steered states: EPR steering and Bell nonlocality}

\begin{figure*}[ht]
    \centering
    \includegraphics[width=0.35\textwidth]{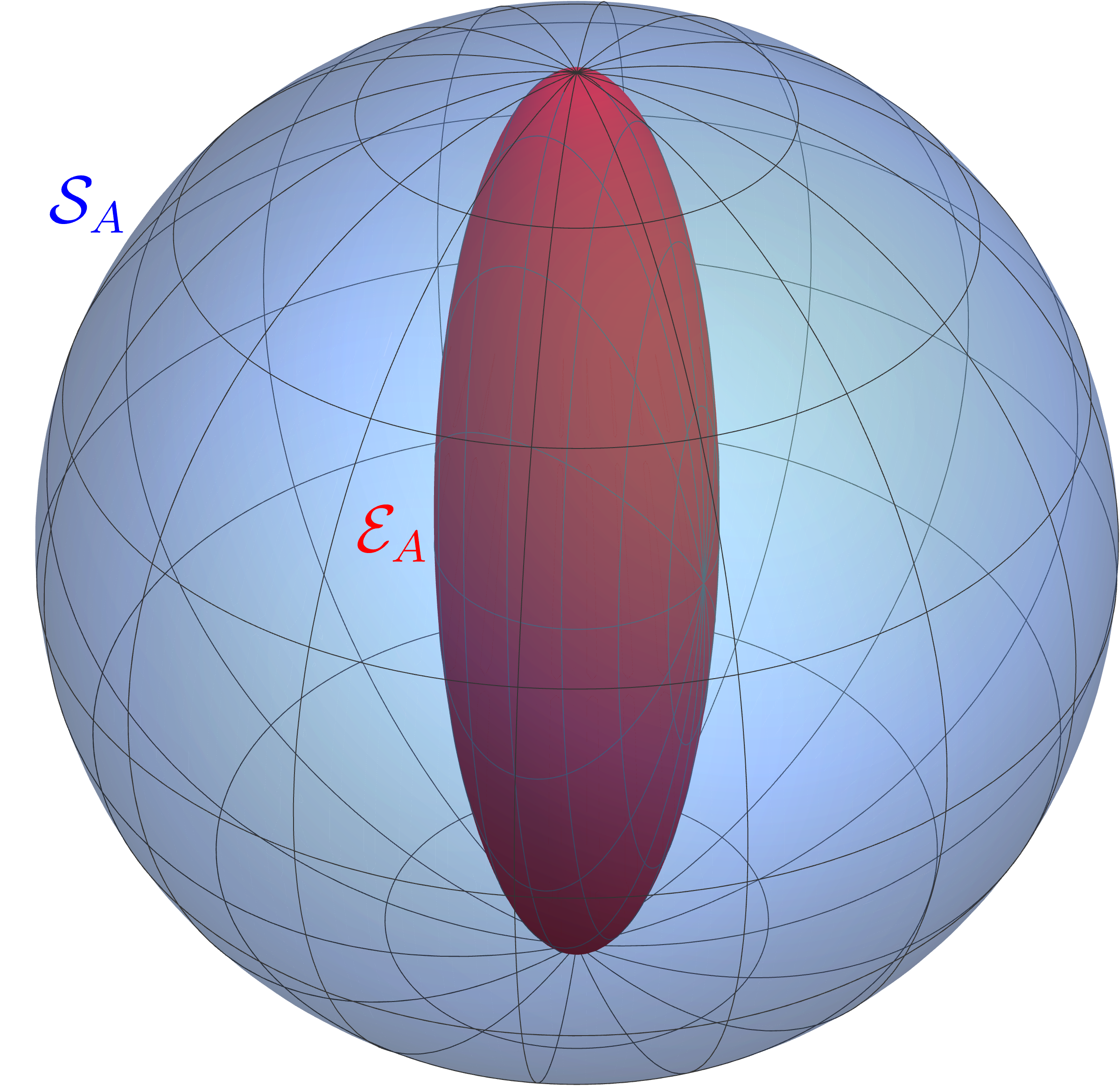}
    \hspace{0.075\textwidth}
    \includegraphics[width=0.35\textwidth]{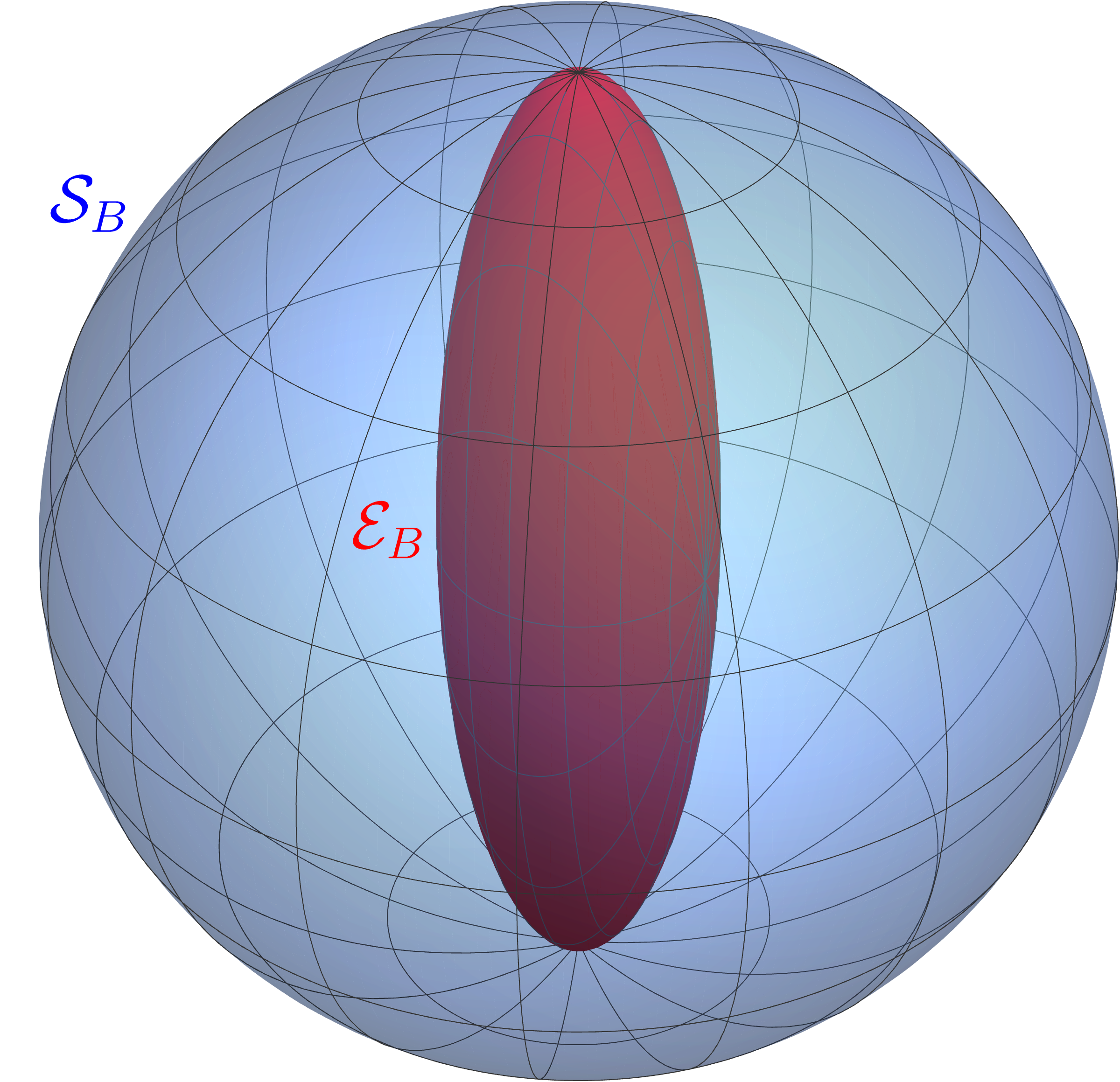}
    \caption{QSEs for the X-states~\eqref{x-state} with $a_z=b_z=1/3$, $t_x=-t_y=1/4$, and $t_z=1$. The two QSEs $\mathcal{E}_A$ and $\mathcal{E}_B$ are identical, with semiaxes $3/(8\sqrt{2})$, $3/(8\sqrt{2})$, and $1$, and are tangent to the Bloch sphere $\mathcal{S}$ at both $|0\rangle$ and $|1\rangle$. The two pure steered states are obtained from a single projective measurement.}
    \label{figexap3}
\end{figure*}

\begin{figure*}[t]
    \centering
    \includegraphics[width=0.35\textwidth]{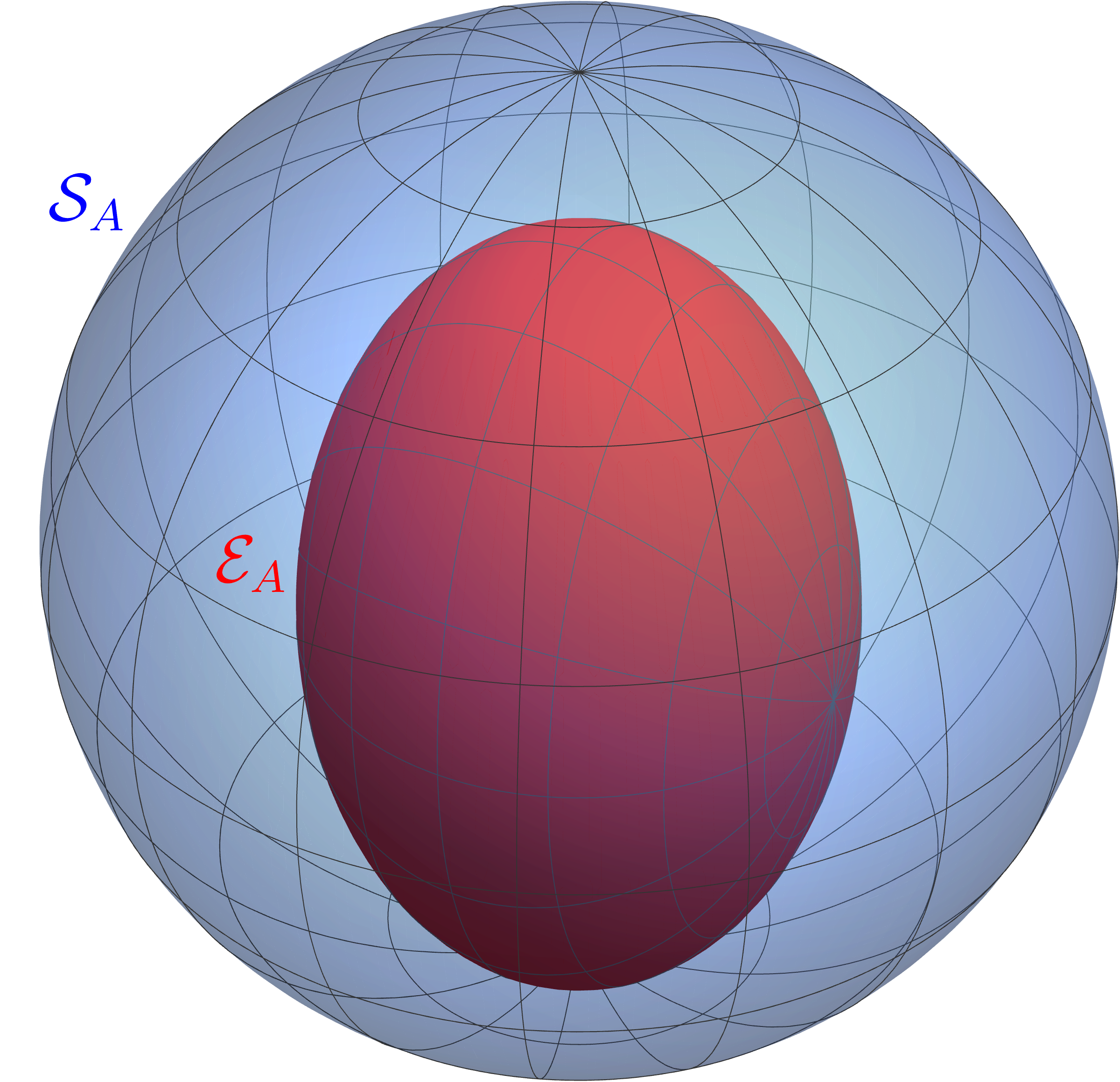}
    \hspace{0.075\textwidth}
    \includegraphics[width=0.35\textwidth]{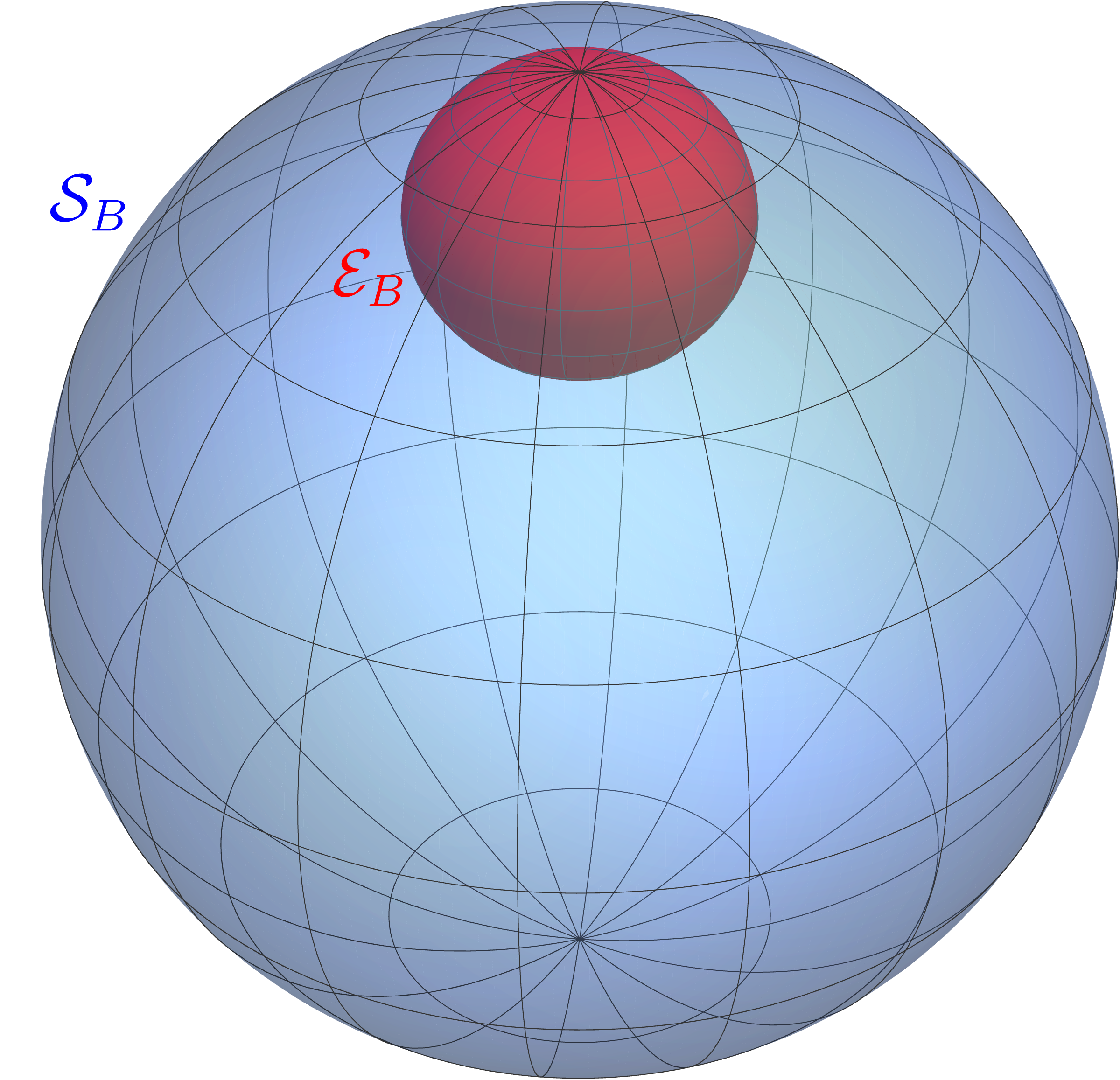}
    \caption{QSEs for the state in Eq.~\eqref{asym} with $q=1/3$, $\eta=1/4$, and $\epsilon=1/5$. Alice's QSE $\mathcal{E}_A$, with semiaxes $\sqrt{2/7}$, $\sqrt{2/7}$, and $23/28$, is tangent to her Bloch sphere $\mathcal{S}_A$ at $|1\rangle$, whereas Bob's QSE $\mathcal{E}_B$, with semiaxes $4/\sqrt{161}$, $4/\sqrt{161}$, and $2/7$, is tangent to his Bloch sphere $\mathcal{S}_B$ at $|0\rangle$.}
    \label{figexap1}
\end{figure*}

\begin{figure*}[ht]
    \centering
    \includegraphics[width=0.35\textwidth]{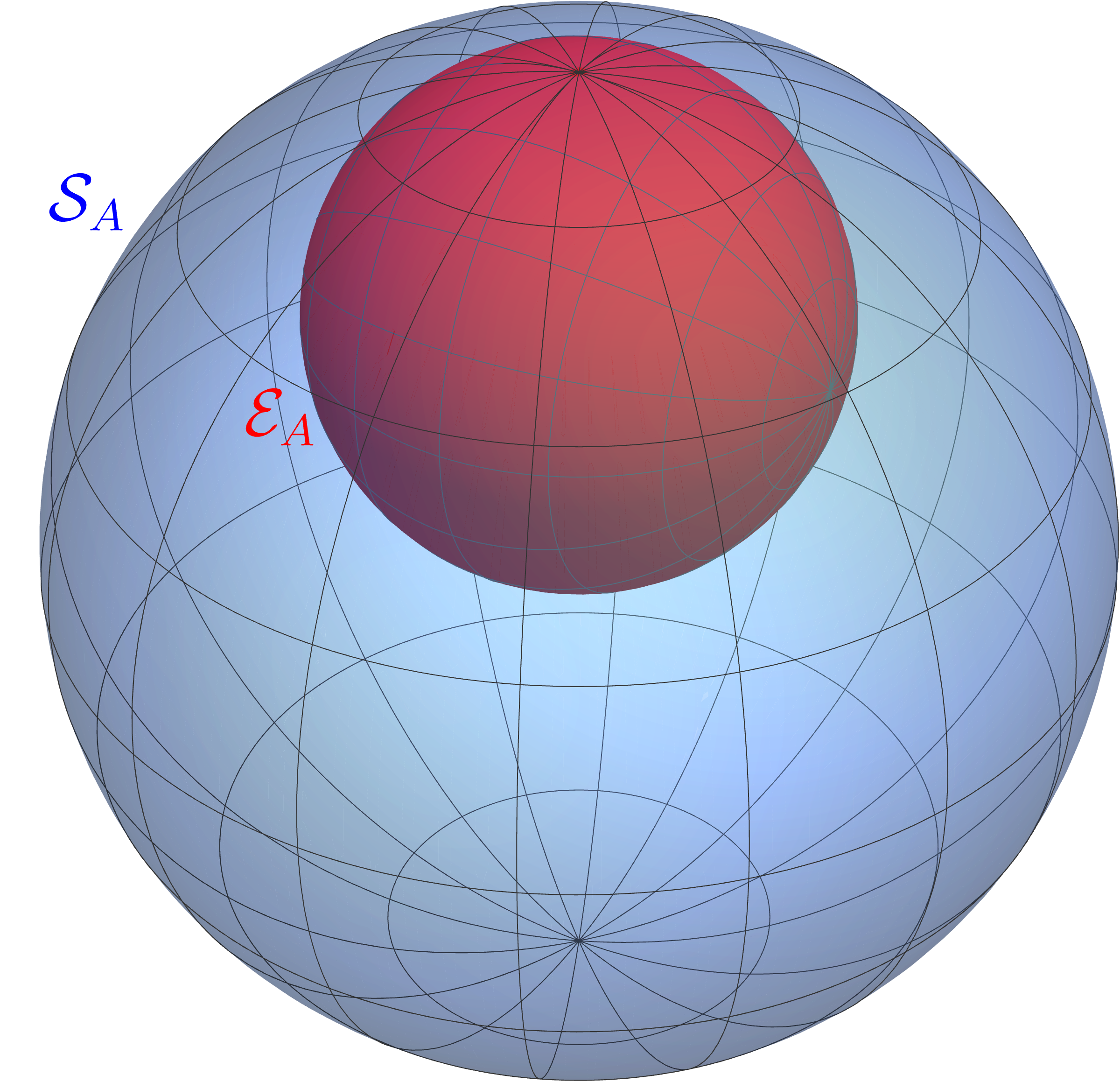}
    \hspace{0.075\textwidth}
    \includegraphics[width=0.35\textwidth]{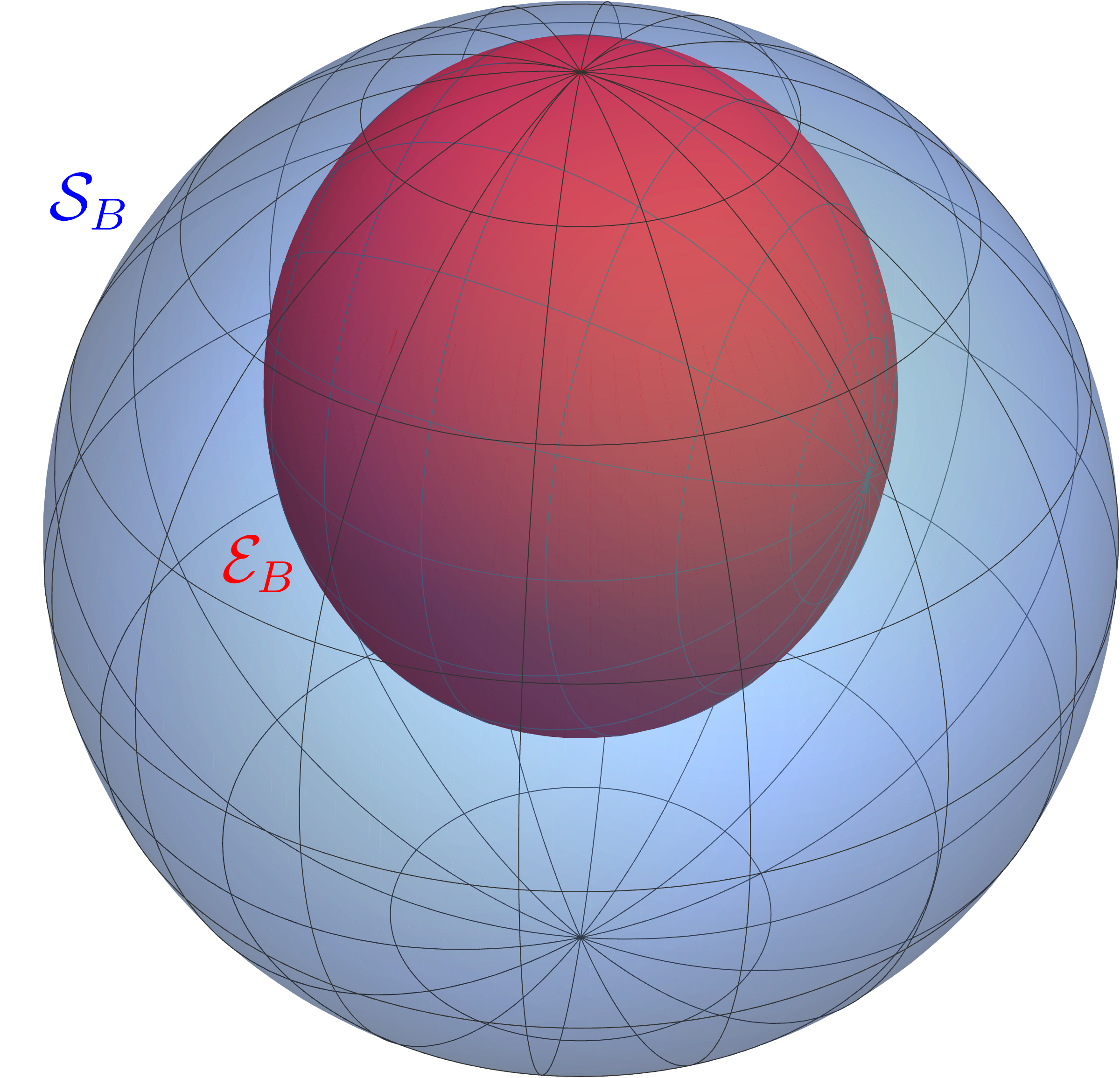}
    \caption{QSEs for the state $\rho_{AB}$~\eqref{examples} with $q=1/2$. The two QSEs are both tangent to the Bloch sphere $\mathcal{S}$ at the same point $|0\rangle$. Alice's QSE $\mathcal{E}_A$ is a sphere with radius $1/2$, while Bob's QSE $\mathcal{E}_B$ has semiaxes $1/\sqrt{3}$, $1/\sqrt{3}$, and $2/3$.}
    \label{figexap2}
\end{figure*}

A representative class of examples is provided by X-states~\cite{xstates}, 
for which the associated QSE 
can exhibit 0, 1, 2, or $\infty$ tangency points with the Bloch sphere 
under different conditions.
In the computational basis 
$\{|00\rangle, |01\rangle, |10\rangle, |11\rangle\}$, 
two-qubit X-states are characterized by a density matrix whose nonzero elements appear only along the main diagonal and anti-diagonal, forming an X-shaped structure~\cite{xstates}. 
Such states can be equivalently expressed in terms of Pauli operators, as
\begin{align}\label{x-state}
    \rho^X_{AB}=\frac{1}{4}\left(I\otimes I+a_z\sigma_z\otimes I+I\otimes b_z\sigma_z+\sum_{i=x,y,z}t_i\sigma_i\otimes \sigma_i\right).
\end{align}
The positivity of X-states requires
\begin{align}\label{positivityXs}
    (1\pm t_z)^2 \geq (a_z\pm b_z)^2+(t_x\mp t_y)^2.
\end{align}
X-states are entangled if and only if either 
\begin{align}
    (1 + t_z)^2 < (a_z + b_z)^2 + (t_x + t_y)^2,
\end{align}
or
\begin{align}
    (1 - t_z)^2 < (a_z - b_z)^2 + (t_x - t_y)^2.
\end{align}

If condition $t_z=1+a_z-b_z$ holds, the two-qubit X-states are called tangent X-states~\cite{Song2023}. 
then Alice’s QSE $\mathcal{E}_A$ 
is tangent to Alice’s Bloch sphere $\mathcal{S}_A$ at $|1\rangle$, 
while Bob’s QSE $\mathcal{E}_B$ 
is tangent to Bob’s Bloch sphere $\mathcal{S}_B$ at $|0\rangle$.
According to Theorem~\ref{twowaysteering}, if the tangent X-states are entangled, then they are two-way steerable.
Similarly, If condition $t_z=1-a_z+b_z$ holds, 
then $\mathcal{E}_B$ 
is tangent to $\mathcal{S}_B$ at $|1\rangle$, 
while $\mathcal{E}_A$ 
is tangent to  $\mathcal{S}_A$ at $|0\rangle$.

If conditions $t_z=1$ and $a_z=b_z$ hold, the positivity condition~\eqref{positivityXs} yields $t_x=-t_y$.
In this case, $\mathcal{E}_B$ is tangent to $\mathcal{S}_B$ 
at both $|0\rangle$ and $|1\rangle$. 
These two pure states can be obtained by a single projective measurement, 
and similarly for Alice.
An illustrative example for $a_z=1/3$ and $t_x=1/4$ is shown in Fig.~\ref{figexap3}.
According to the Horodecki criterion~\cite{Horodecki95}, 
X-states~\eqref{x-state} violate the Clauser-Horne-Shimony-Holt (CHSH) inequality~\cite{CHSH} 
whenever $\sqrt{s_1^2+s_2^2}>1$, 
where $s_j$ $(j=1,2,3)$ denote the singular values, in decreasing order, of the correlation matrix 
$T=\mathrm{diag}(t_x,t_y,t_z)$. 
In this case, $t_z=1$, and therefore entangled X-states necessarily violates the CHSH inequality.
%\end{itemize}

If the X-states are pure, then the number of tangency points is infinite. 
If none of the above conditions are satisfied, then the QSE 
has no tangency points with the Bloch sphere.

A specific example is provided by the following entangled state:
\begin{align}\label{asym}
    \rho^\mathrm{asym}_{AB}=q|\psi_\epsilon\rangle\langle \psi_\epsilon| + (1-q)\rho_\eta^A\otimes |0\rangle\langle 0|,  
\end{align}
where $|\psi_\epsilon\rangle:=\sqrt{1-\epsilon}|00\rangle + \sqrt{\epsilon}|11\rangle$, 
$|\psi_\eta\rangle:=\sqrt{1-\eta}|00\rangle + \sqrt{\eta}|11\rangle,\ (0< q, \epsilon, \eta <1)$ and $\rho^A_\eta=\Tr_B(|\psi_\eta\rangle\langle \psi_\eta|)$.  
The calculation shows that $\mathcal{E}_A$ is tangent to $\mathcal{S}_A$ 
at the state $|1\rangle$, while $\mathcal{E}_B$ is tangent to 
$\mathcal{S}_B$ at the state $|0\rangle$. This asymmetric entangled state $\rho^\mathrm{asym}_{AB}$ is not one-way steerable but instead two-way steerable under all projective measurements. 
As an illustrative example, Fig.~\ref{figexap1} shows the case with 
$q=1/3$, $\eta=1/4$, and $\epsilon=1/5$ in the state defined in Eq.~\eqref{asym}.
When $\eta=\epsilon$, this state reduces to the bipartite reduced semirandom pair entangled (SRPE) state~\cite{Songqc2023}. When $\epsilon=\eta=1/2$, the state $\rho_\eta^A$ reduces to $I/2$, and the state~\eqref{asym} is locally unitarily equivalent to 
the state~\eqref{examples}. We also provide another illustrative example, 
as shown in Fig.~\ref{figexap2}, corresponding to the case 
$q=1/2$ in the state given in Eq.~\eqref{examples}.

\subsection{Examples of two-qubit entangled states admitting exactly one pure steered state}
We now present three examples of two-qubit entangled states parameterized by three real parameters $0<x,y,z<1$.
The first example is given by
\begin{align}\label{exastate1p}
\rho_{AB}=\left(
\begin{array}{cccc}
 K(1-z) & X (1-2 y) (1-z) &  X Y Z & (1-x) Y Z \\
 X (1-2 y) (1-z) & (1-K)(1-z) & - x Y Z & - X Y Z \\
  X Y Z & - x Y Z & x z & X z \\
 (1-x) Y Z & - X Y Z & X z & (1-x)z \\
\end{array}
\right),
\end{align}
where 
%$0<x,y,z<1$ and
\begin{align}
    X:=\sqrt{x(1-x)},\quad Y:=\sqrt{\frac{2y}{3}} ,\quad Z:=\sqrt{z(1-z)},\quad K:=x+y-2xy.
\end{align}
Alice can steer Bob's system to a unique pure state
\begin{align}
|\beta\rangle_B = \sqrt{x}|0\rangle + \sqrt{1-x} |1\rangle
\end{align}
by the measurement effect  $|1\rangle\langle 1|$.
Similarly, Bob can steer Alice's system to a unique pure state
\begin{align}
|\alpha\rangle_A = |0\rangle
\end{align}
by the measurement effect $|\bar{b}\rangle\langle \bar{b}|$, 
where $|\bar{b}\rangle=\sqrt{1-x}|0\rangle - \sqrt{x} |1\rangle$.
Figure~\ref{fig:exastate1p} illustrates the state $\rho_{AB}$~\eqref{exastate1p} for $x=3/4$, $y=1/2$, and $z=1/2$.
\begin{figure*}[t]
    \centering
    \includegraphics[width=0.35\textwidth]{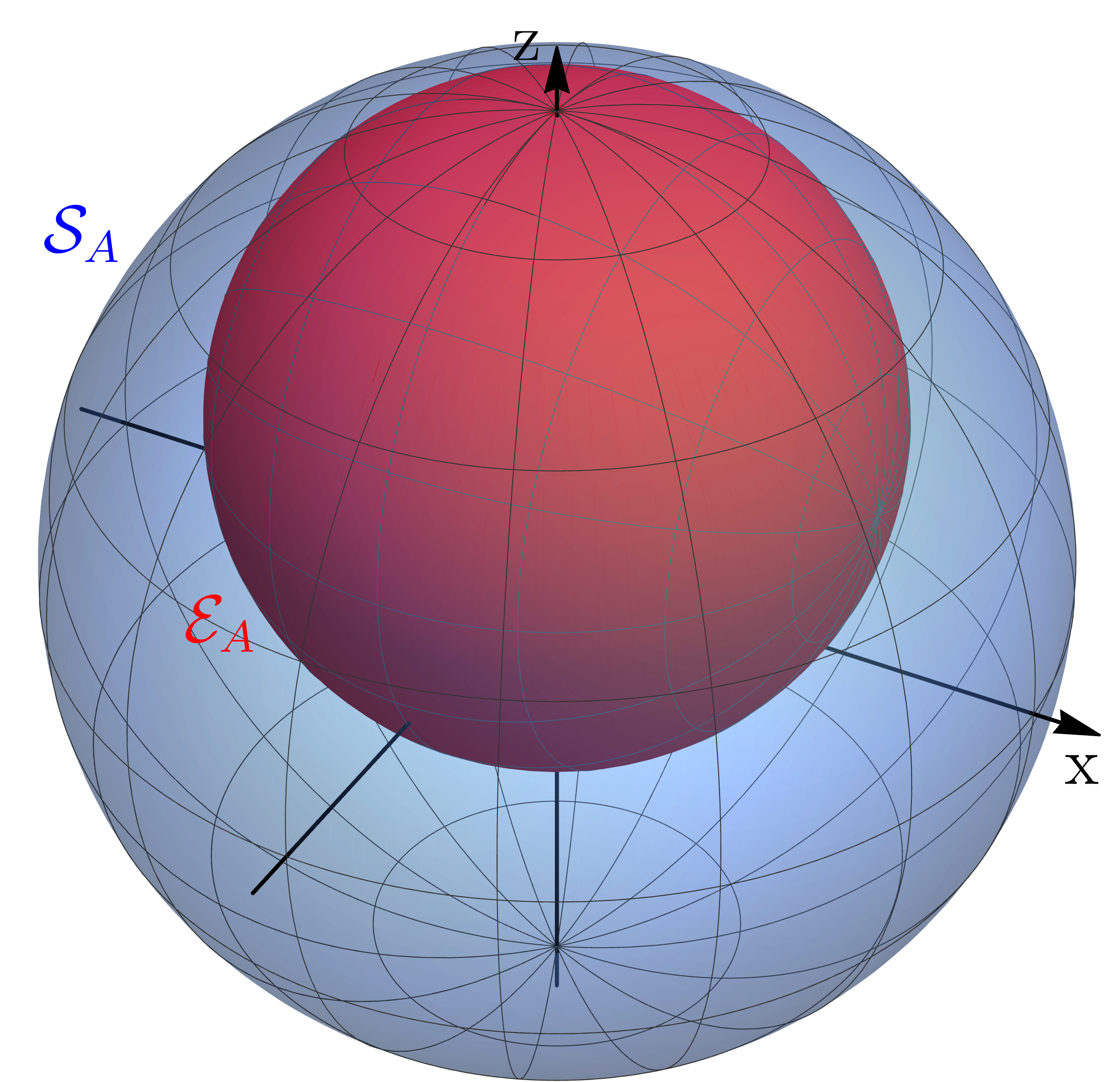}
    \hspace{0.075\textwidth}
    \includegraphics[width=0.35\textwidth]{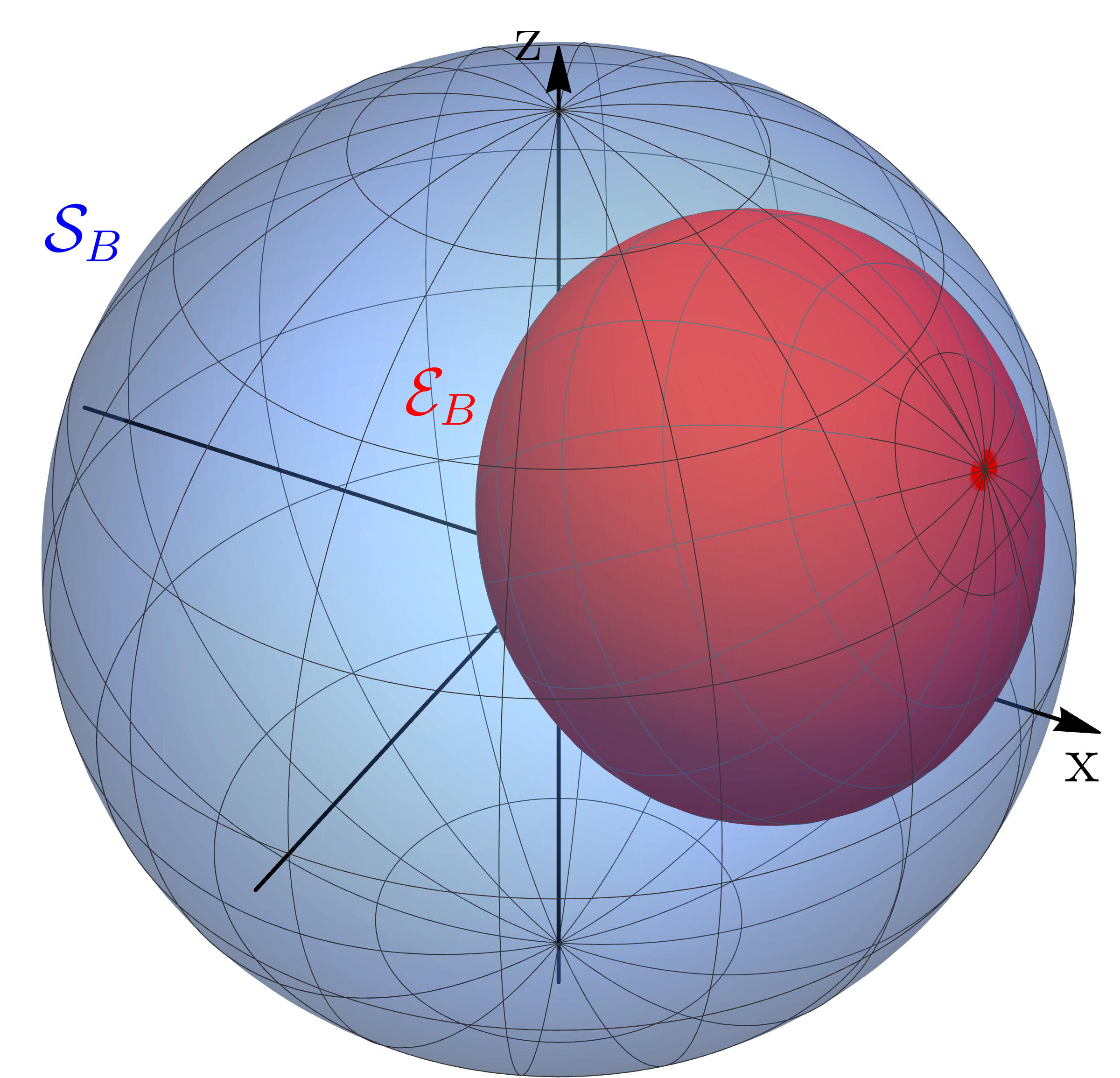}
    \caption{QSEs for the state $\rho_{AB}$~\eqref{exastate1p} with $x=3/4$, $y=1/2$, and $z=1/2$. Alice's QSE $\mathcal{E}_A$ is a sphere with radius $2/3$ and is tangent to her Bloch sphere $\mathcal{S}_A$ at $|0\rangle$. Bob's QSE $\mathcal{E}_B$, with semiaxes $1/\sqrt{3}$, $1/\sqrt{3}$, and $1/2$, is tangent to his Bloch sphere $\mathcal{S}_B$ at the state $(\sqrt{3}|0\rangle + |1\rangle)/2$.}
    \label{fig:exastate1p}
\end{figure*}

\vspace{1.0\baselineskip}

\begin{figure*}[ht]
    \centering
    \includegraphics[width=0.35\textwidth]{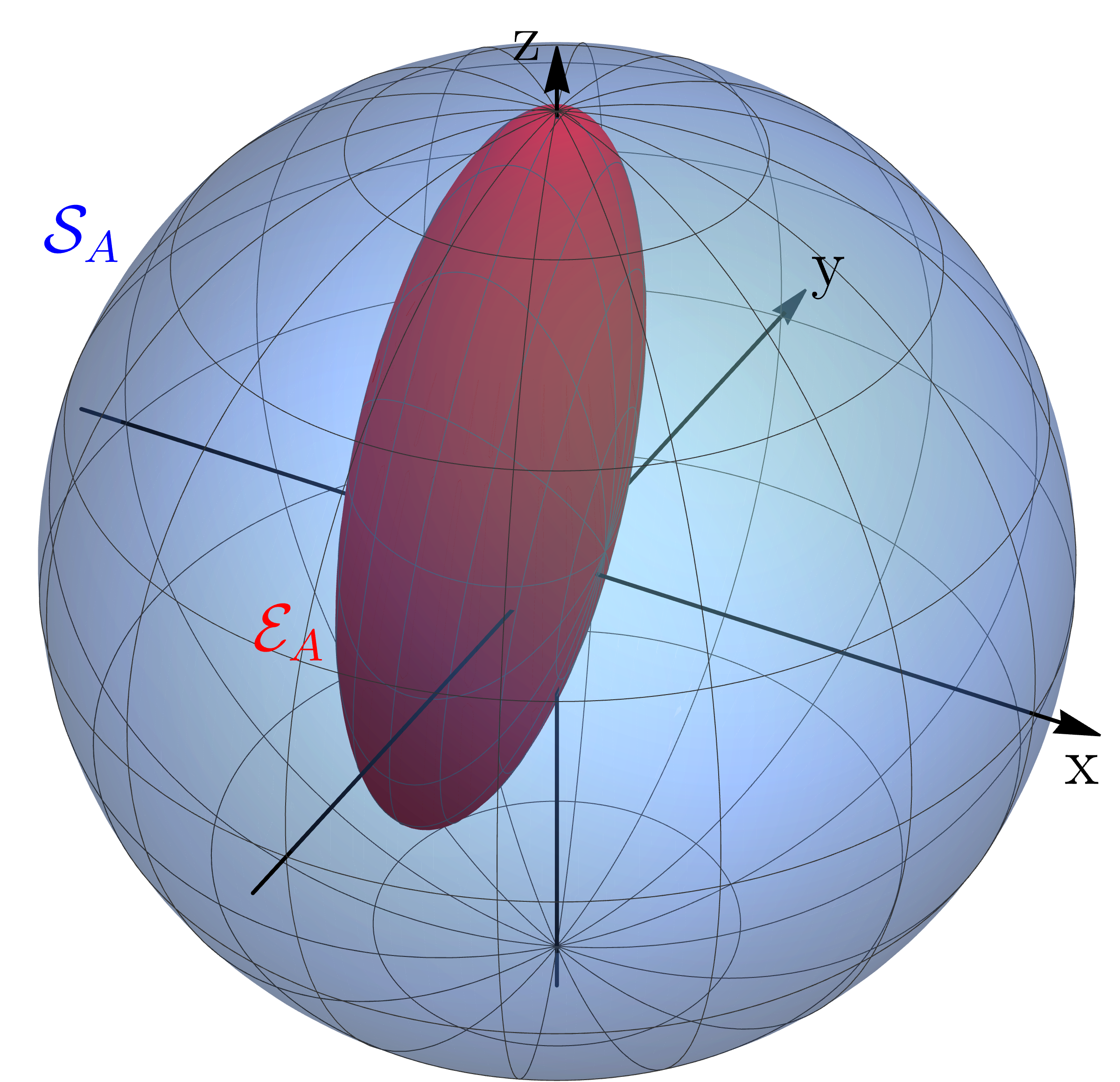}
    \hspace{0.075\textwidth}
    \includegraphics[width=0.35\textwidth]{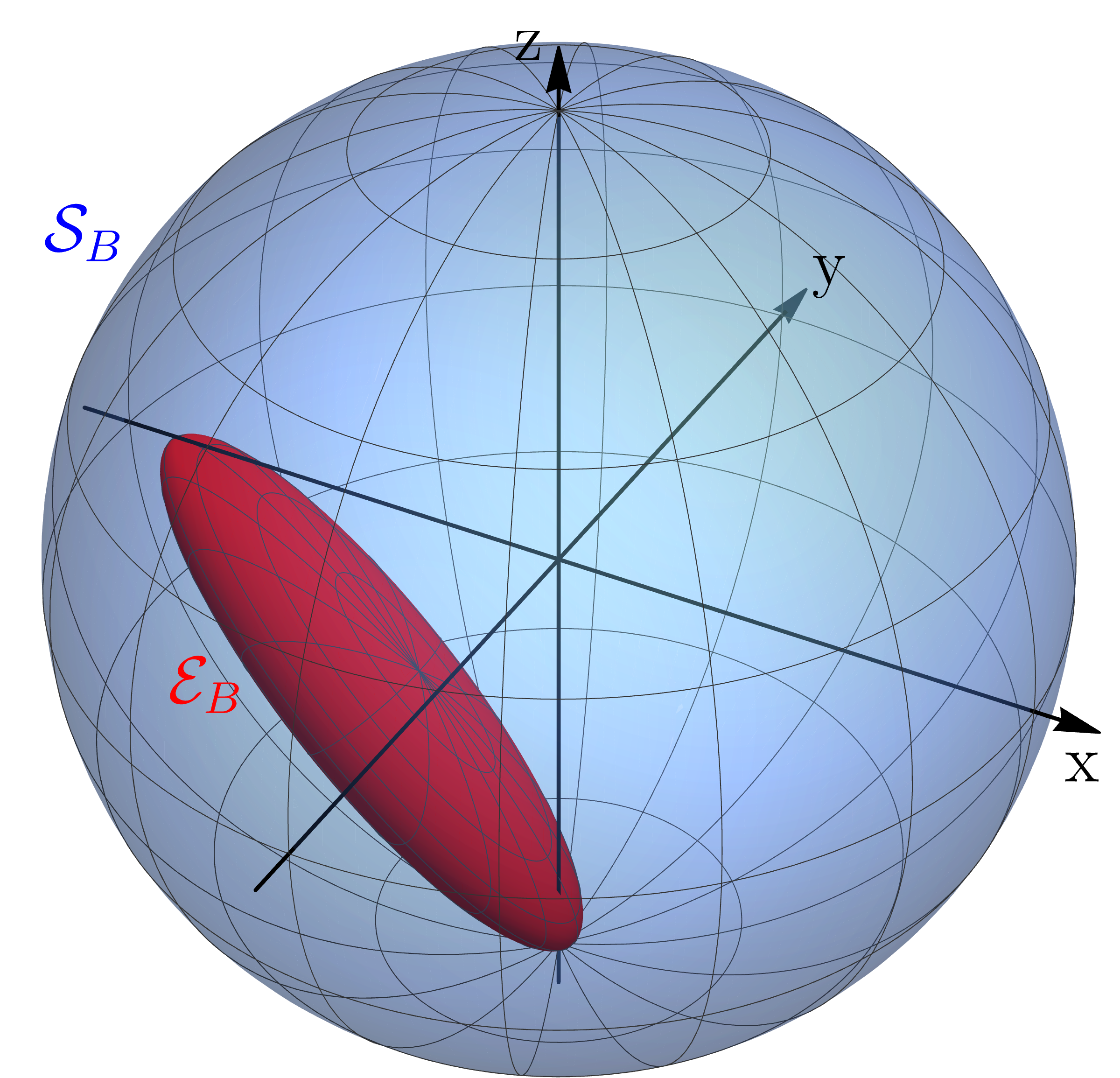}
    \caption{QSEs for the state $\rho_{AB}$~\eqref{exastate1p2} with $x=1/2$, $y=1/2$, and $z=1/3$. Alice's QSE $\mathcal{E}_A$ is tangent to her Bloch sphere $\mathcal{S}_A$ at $|0\rangle$, while Bob's QSE $\mathcal{E}_B$ is tangent to his Bloch sphere $\mathcal{S}_B$ at $|1\rangle$.}
    \label{fig:figp2}
\end{figure*}

The second example is defined as
\begin{align}\label{exastate1p2}
\rho_{AB}=\left(
\begin{array}{cccc}
 x y & - X y Z_{+} & 0 &  \sqrt{x} Y Z_{-} \\
 - X y Z_{+} & y-x y & 0 & 0 \\
 0 & 0 & 0 & 0 \\
  \sqrt{x} Y Z_{-} & 0 & 0 & 1-y \\
\end{array}
\right)
\end{align}
where 
%$0<x,y,z<1$ and
\begin{align}
    X:= \sqrt{x(1-x)},\quad Y:=\sqrt{y(1-y)} ,\quad Z_{\pm}:=\frac{1}{2}\left(\sqrt{z} \pm \sqrt{2-2 z}\right).
\end{align}
Alice can steer Bob's system to exactly one pure state
\begin{align}
|\beta\rangle_B = |1\rangle
\end{align}
by the measurement effect  $|1\rangle\langle 1|$.
Similarly, Bob can steer Alice to exactly one pure state
\begin{align}
|\alpha\rangle_A = |0\rangle
\end{align}
by the measurement effect $|0\rangle\langle 0|$.
We show an example for the state $\rho_{AB}$~\eqref{exastate1p2} with $x=1/2$, $y=1/2$, and $z=1/3$ as shown in Fig.~\ref{fig:figp2}.

\vspace{1.0\baselineskip}

\begin{figure*}[t]
    \centering
    \includegraphics[width=0.35\textwidth]{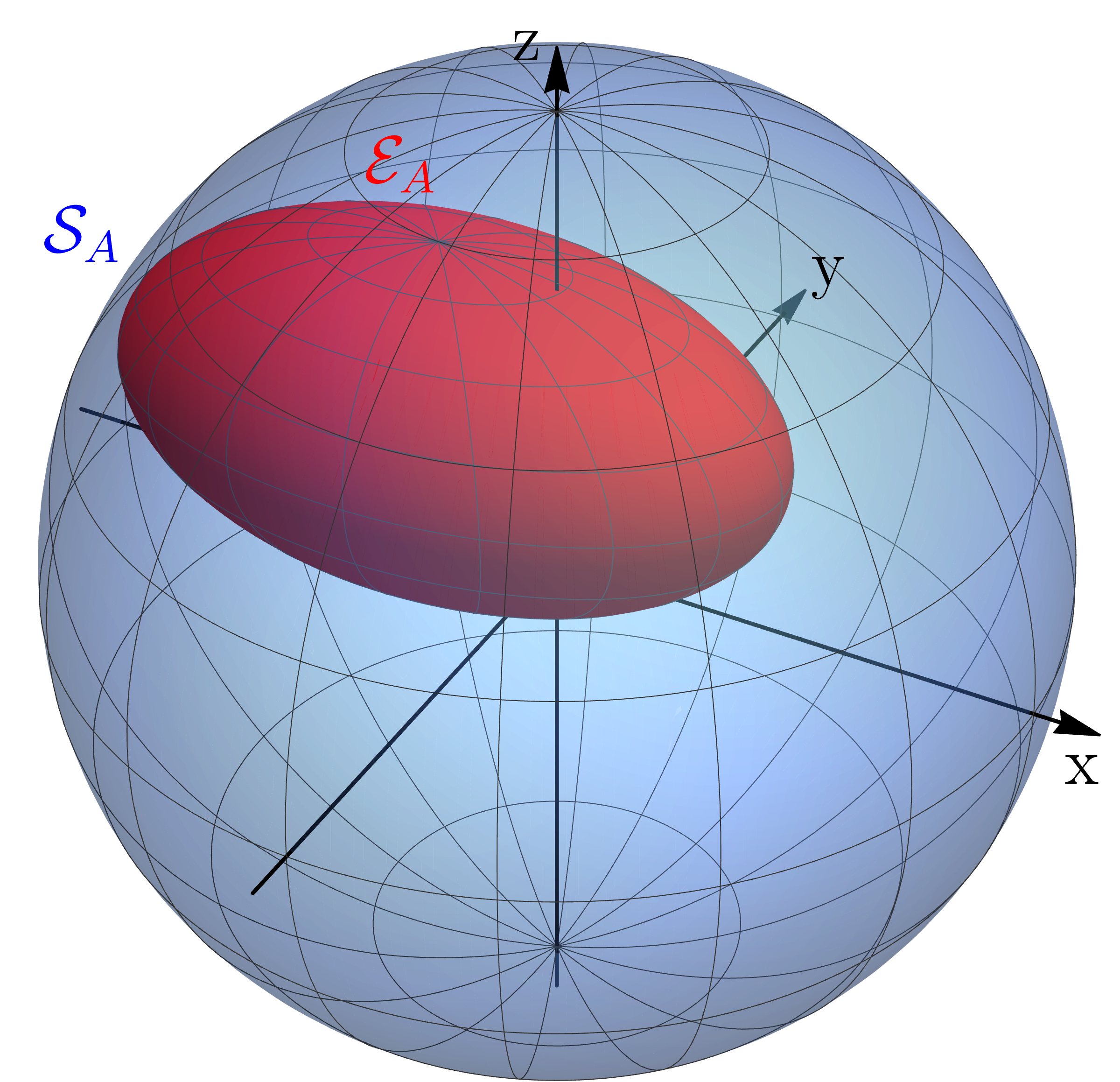}
    \hspace{0.075\textwidth}
    \includegraphics[width=0.35\textwidth]{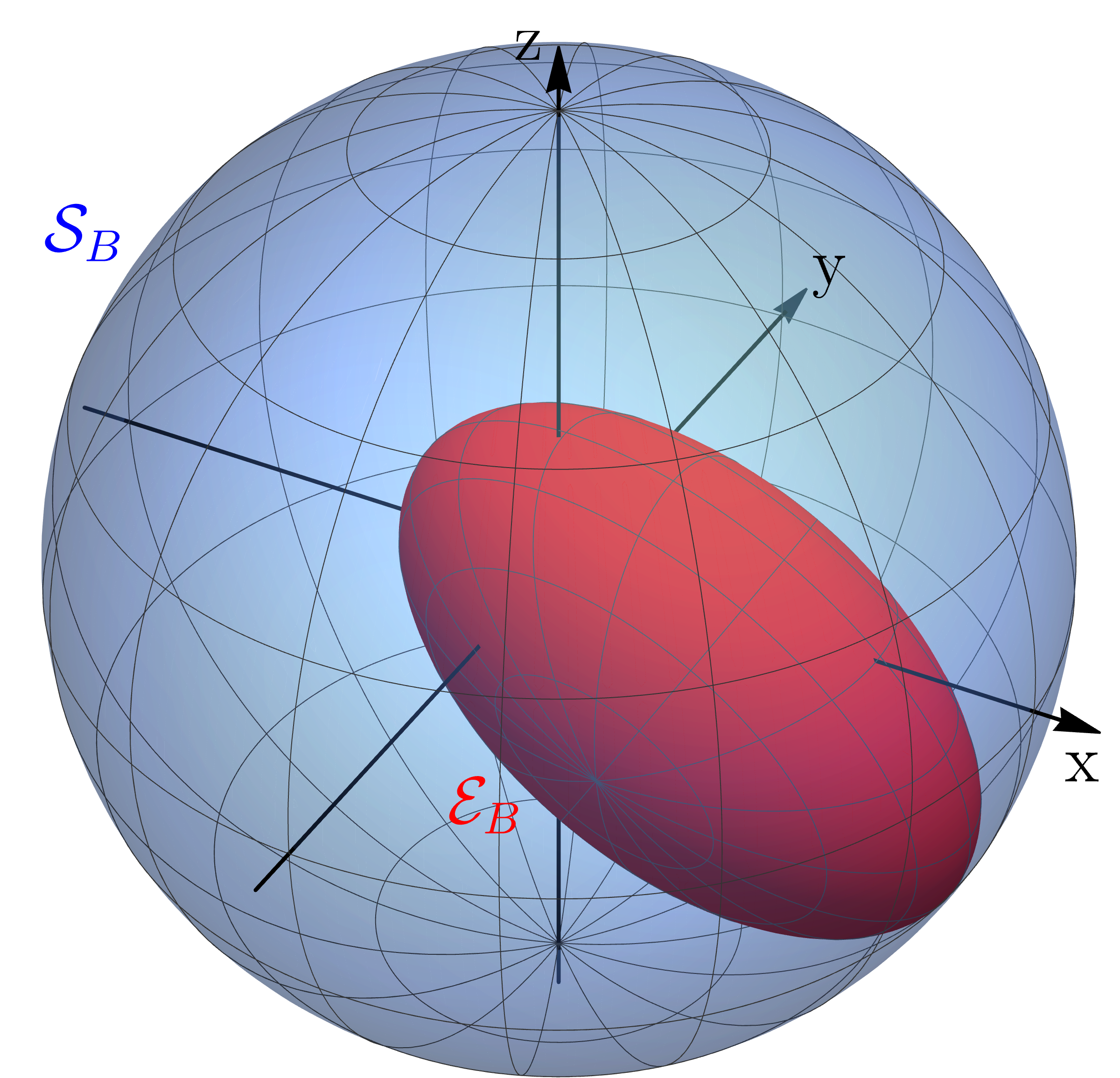}
    \caption{QSEs for the state $\rho_{AB}$~\eqref{exastate1p3} with $x=1/4$, $y=1/3$, and $z=1/3$. Alice's QSE $\mathcal{E}_A$ is tangent to her Bloch sphere $\mathcal{S}_A$ at $(|0\rangle + \sqrt{3}|1\rangle)/2$, while Bob's QSE $\mathcal{E}_B$ is tangent to his Bloch sphere $\mathcal{S}_B$ at $(\sqrt{2}|0\rangle - |1\rangle)/\sqrt{3}$.}
    \label{fig:figp3}
\end{figure*}

The third example takes the form
\begin{align}\label{exastate1p3}
\rho_{AB}=\frac{1}{3}\left(
\begin{array}{cccc}
 (1-z) \left(1+x-2 X \sqrt{y}\right) & (1-z) \left((2 x-1) \sqrt{y}+X\right) & (2 x-1) Z &  Z \left(2 X-\sqrt{y}\right) \\
 (1-z) \left((2 x-1) \sqrt{y}+X\right) &  (z-1) \left(x-2-2 X \sqrt{y}\right) & Z \left(2 X+\sqrt{y}\right) &  (1-2 x) Z \\
  (2 x-1) Z &  Z \left(2 X+\sqrt{y}\right) &  z \left(2 X \sqrt{y}+x+1\right) &  z \left((1-2 x) \sqrt{y}+X\right) \\
  Z \left(2 X-\sqrt{y}\right) &  (1-2 x) Z &  z \left((1-2 x) \sqrt{y}+X\right) & z \left(2-x-2 X \sqrt{y}\right) \\
\end{array}
\right)
\end{align}
where 
%$0<x,y,z<1$ and
\begin{align}
    X:= \sqrt{x(1-x)} ,\quad Z:=\sqrt{z(1-z)}.
\end{align}
Alice can steer Bob's system to a unique pure state
\begin{align}
|\beta\rangle_B = \sqrt{x} |0\rangle + \sqrt{1-x} |1\rangle.
\end{align}
Bob can steer Alice's system to a unique pure state
\begin{align}
|\alpha\rangle_A = \sqrt{1-z} |0\rangle - \sqrt{z} |1\rangle.
\end{align}
We show an illustrative example for the state $\rho_{AB}$~\eqref{exastate1p3} with $x=1/4$, $y=1/3$, and $z=1/3$ as shown in Fig.~\ref{fig:figp3}.

\subsection{Examples of two-qubit entangled states admitting exactly two pure steered states}

We now present two examples of two-qubit entangled states 
parameterized by three real parameters $0<x,y<1$.
\begin{figure*}[t]
    \centering
    \includegraphics[width=0.35\textwidth]{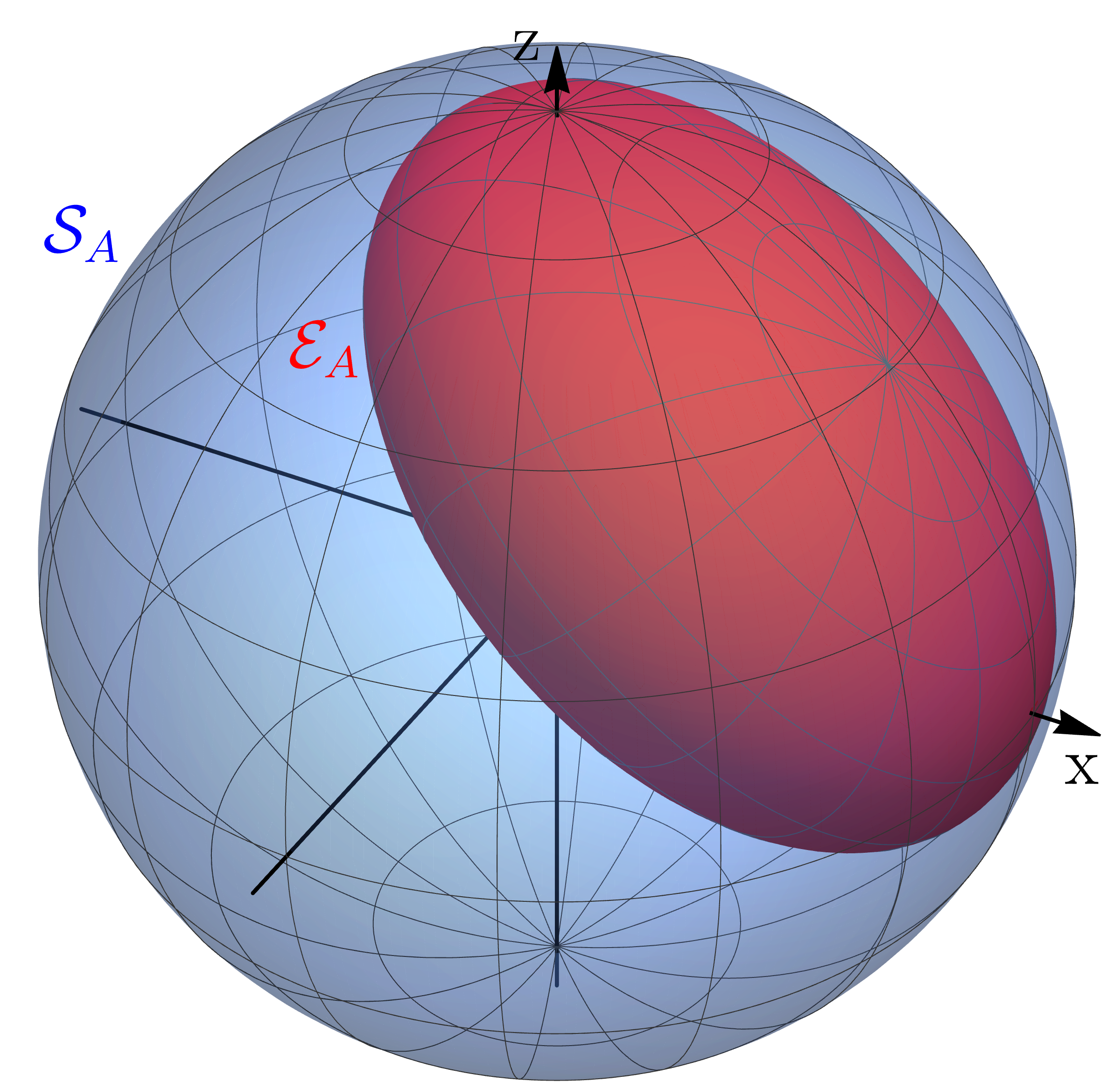}
    \hspace{0.075\textwidth}
    \includegraphics[width=0.35\textwidth]{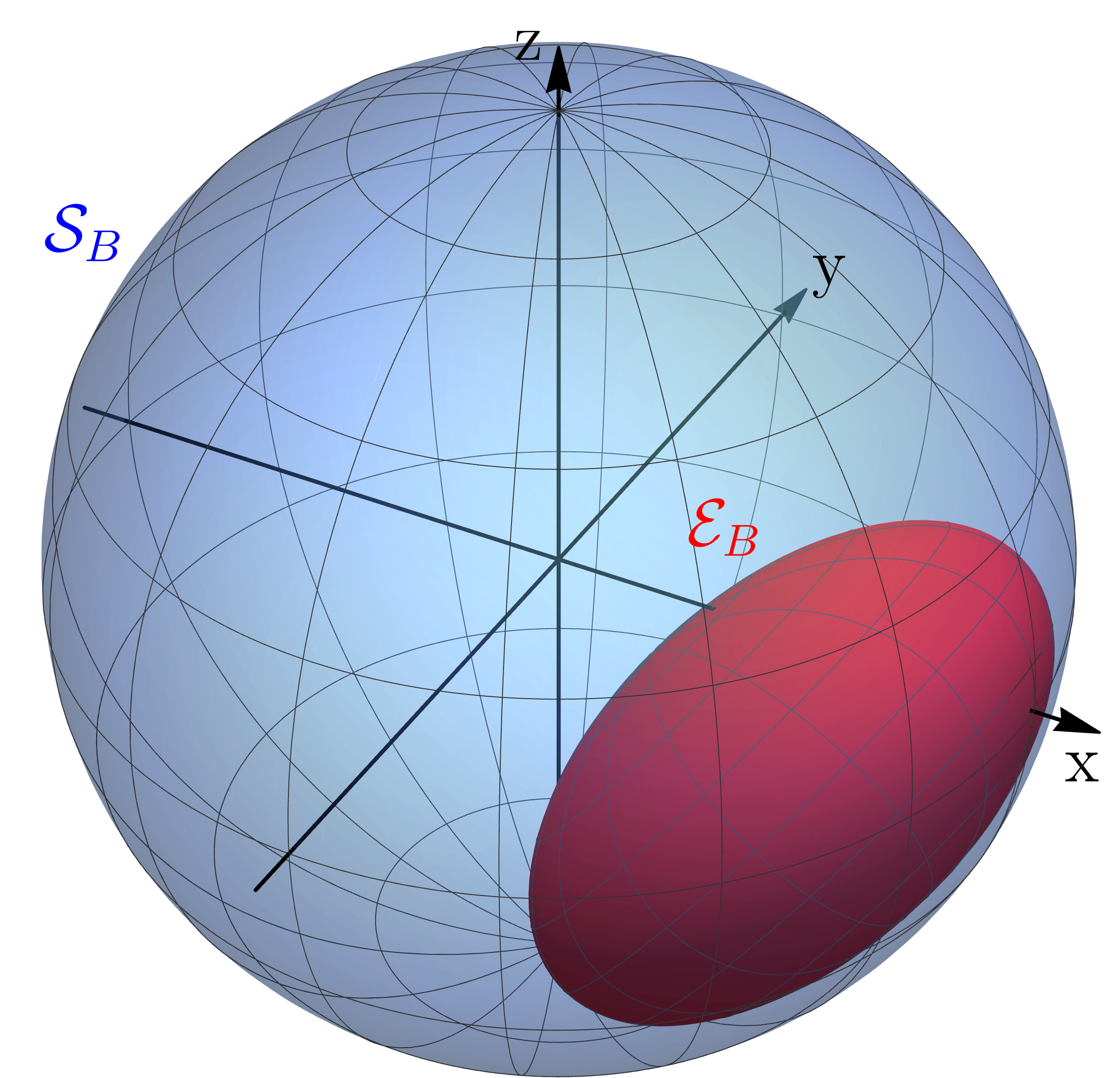}\\
    \vspace{0.075\textwidth}
    \includegraphics[width=0.35\textwidth]{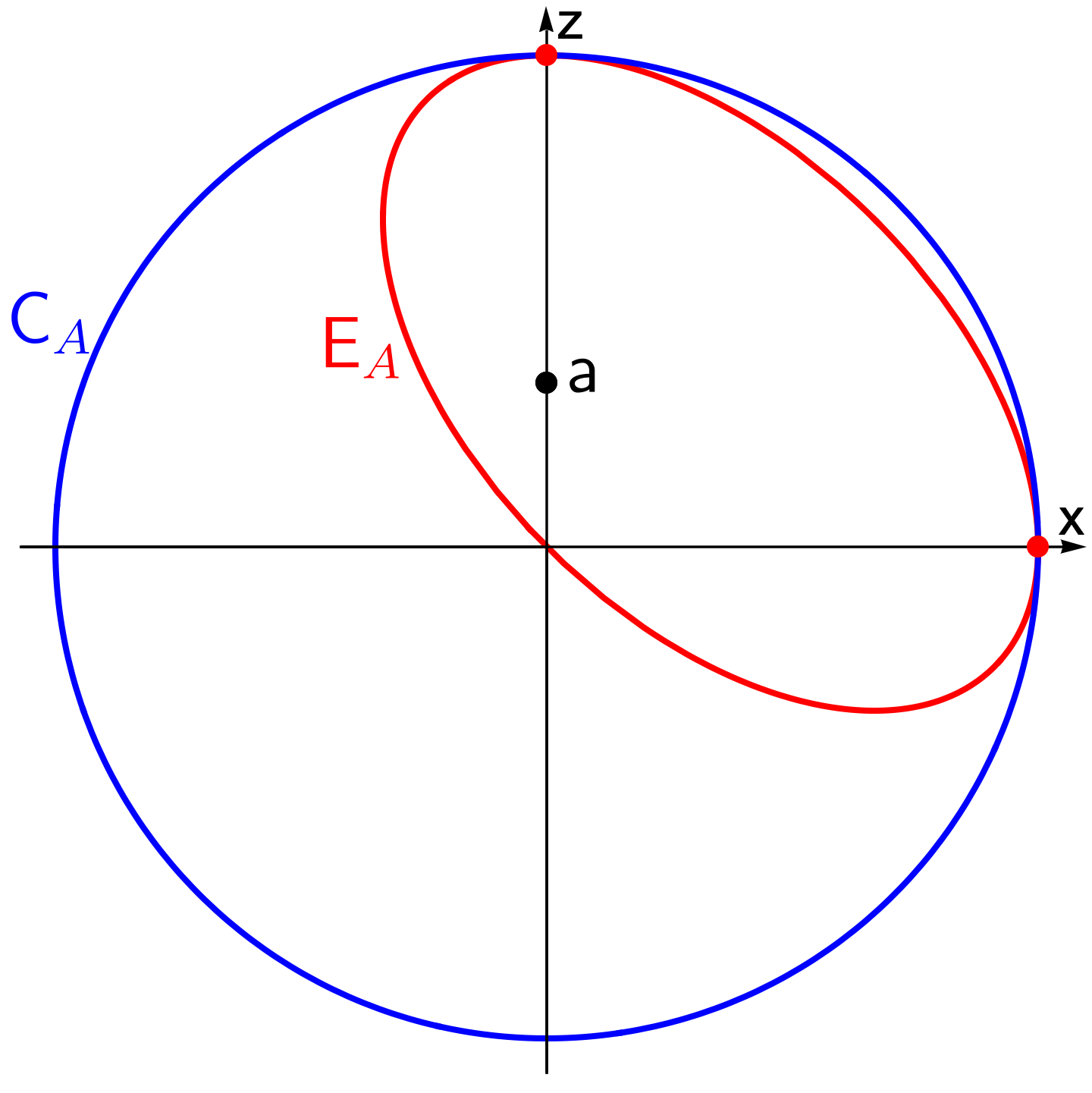}
    \hspace{0.075\textwidth}
    \includegraphics[width=0.35\textwidth]{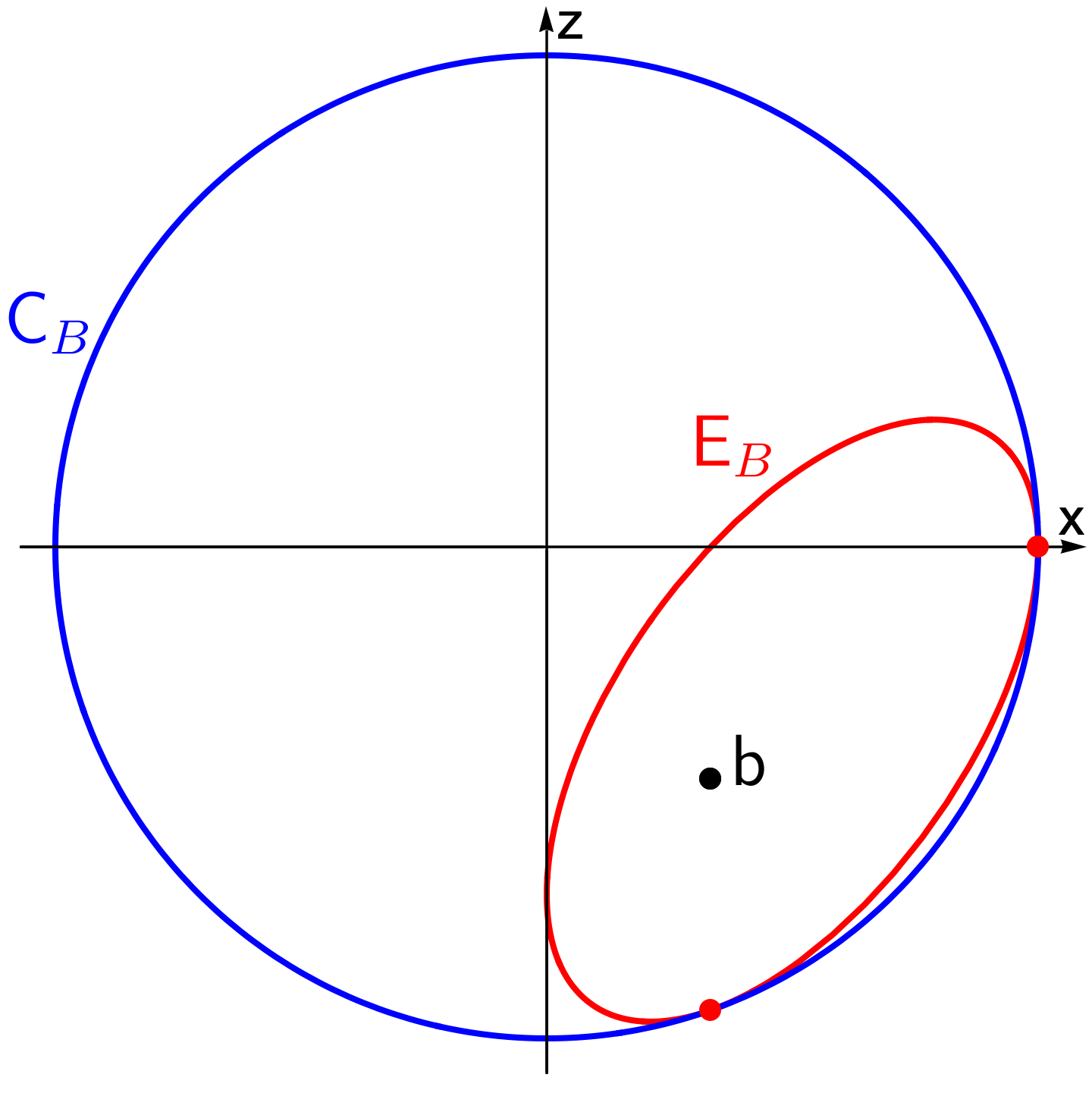}\\
    \caption{QSEs for the state $\rho_{AB}$~\eqref{exastate2p} with $x=1/2$ and $y=1/3$. Alice's QSE $\mathcal{E}_A$ is tangent to her Bloch sphere $\mathcal{S}_A$ at $|0\rangle$ and $(|0\rangle+|1\rangle)/\sqrt{2}$, while Bob's QSE $\mathcal{E}_B$ is tangent to his Bloch sphere $\mathcal{S}_B$ at $(|0\rangle+|1\rangle)/\sqrt{2}$ and $\big[(\sqrt{2}-1)|0\rangle+(\sqrt{2}+1)|1\rangle\big]/\sqrt{6}$. Points $\mathsf{a}$ and $\mathsf{b}$ denote the Bloch vectors of Alice and Bob, respectively.}
    \label{figexap4}
\end{figure*}

The first example is given by
\begin{align}\label{exastate2p}
\rho_{AB}=\left(
\begin{array}{cccc}
 \frac{1}{2} \left(\sqrt{2} X-1\right) (y-1) & \frac{1}{2 \sqrt{2}}(2 x-1) (1-y) & \frac{1}{2} X Y & \frac{1}{2} (1-x) Y \\
 \frac{1}{2 \sqrt{2}}(2 x-1) (1-y) & \frac{1}{2} \left(\sqrt{2} X+1\right) (1-y) & -\frac{1}{2} x Y & -\frac{1}{2} X Y \\
 \frac{1}{2} X Y & -\frac{1}{2} x Y & x y & X y \\
 \frac{1}{2} (1-x) Y & -\frac{1}{2} X Y & X y & (1-x) y \\
\end{array}
\right)
\end{align}
where 
\begin{align}
    X:=\sqrt{x(1-x)},\quad
    Y:=\sqrt{y(1-y)}.
\end{align}
The two pure steered states on system $A$ read
\begin{align}
    |\alpha\rangle_A = |0\rangle, \quad 
     |\alpha^{\prime}\rangle_{A} = \frac{\sqrt{1-y }}{\sqrt{1+y}}|0\rangle + \frac{\sqrt{2y}}{\sqrt{1+y}}|1\rangle.
\end{align}
The two pure steered states on system $B$ read
\begin{align}
    |\beta\rangle_B = \frac{1}{\sqrt{2}}(|0\rangle + |1\rangle), \quad 
     |\beta^{\prime}\rangle_B = \frac{\sqrt{2}-1}{\sqrt{6}}|0\rangle + \frac{\sqrt{2}+1}{\sqrt{6}}|1\rangle.
\end{align}
On both Alice's and Bob's sides, the two pure steered states cannot be obtained by a single projective measurement. An example is shown in Fig.~\ref{figexap4} for the state $\rho_{AB}$~\eqref{exastate2p} with $x=1/2$ and $y=1/3$.

\vspace{1.0\baselineskip}

The second example is defined as
\begin{align}\label{exastate2p2}
\rho_{AB}=\left(
\begin{array}{cccc}
 2K\left(1+x-\sqrt{6} X\right) & K\left(2 \sqrt{6} x+2 X-\sqrt{6}\right)  &  \left(2 - 4 x - \sqrt{6} X\right) Y &  \left(\sqrt{6} x-4 X\right) Y \\
 K\left(2 \sqrt{6} x+2 X-\sqrt{6}\right) & 2K \left(2-x+\sqrt{6} X\right) &  \left(\sqrt{6} x-4 X-\sqrt{6}\right) Y &  \left(4 x+\sqrt{6} X-2\right) Y \\
 \left(2 - 4 x - \sqrt{6} X\right) Y & \left(\sqrt{6} x-4 X-\sqrt{6}\right) Y & \frac{1}{3} (1+x) y & \frac{1}{3} X y \\
  \left(\sqrt{6} x-4 X\right) Y &  \left(4 x+\sqrt{6} X-2\right) Y & \frac{1}{3} X y & \frac{1}{3} (2-x) y \\
\end{array}
\right)
\end{align}
where
\begin{align}   
    X:=\sqrt{x(1-x)}, \quad 
    Y:=\frac{1}{6}\sqrt{y(1-y)},\quad
    K:=\frac{1-y}{6}.
\end{align}
Bob can steer system $A$ to exactly two distinct pure steered states
\begin{align}
    |\alpha\rangle_A = \sqrt{1-y}|0\rangle + \sqrt{y} |1\rangle, \quad 
     |\alpha^{\prime}\rangle_{A} = -\frac{\sqrt{1-y }}{\sqrt{1+3y}}|0\rangle + \frac{2\sqrt{y }}{\sqrt{1+3y}}|1\rangle.
\end{align}
The two pure steered states on system $B$ are uniquely given by
\begin{align}
    |\beta\rangle_B = \sqrt{x}|0\rangle + \sqrt{1-x} |1\rangle, \quad 
     |\beta^{\prime}\rangle_B = \frac{\sqrt{3(1-x)}-\sqrt{2x}}{\sqrt{5}}|0\rangle - \frac{\sqrt{2(1-x)}+\sqrt{3x}}{\sqrt{5}}|1\rangle.
\end{align}
On both Alice's and Bob's sides, the two pure steered states cannot be obtained by a single projective measurement. This feature is illustrated in Fig.~\ref{figexap5} for the state $\rho_{AB}$~\eqref{exastate2p2} with $x=2/3$ and $y=1/4$.

\begin{figure*}[t]
    \centering
    \includegraphics[width=0.35\textwidth]{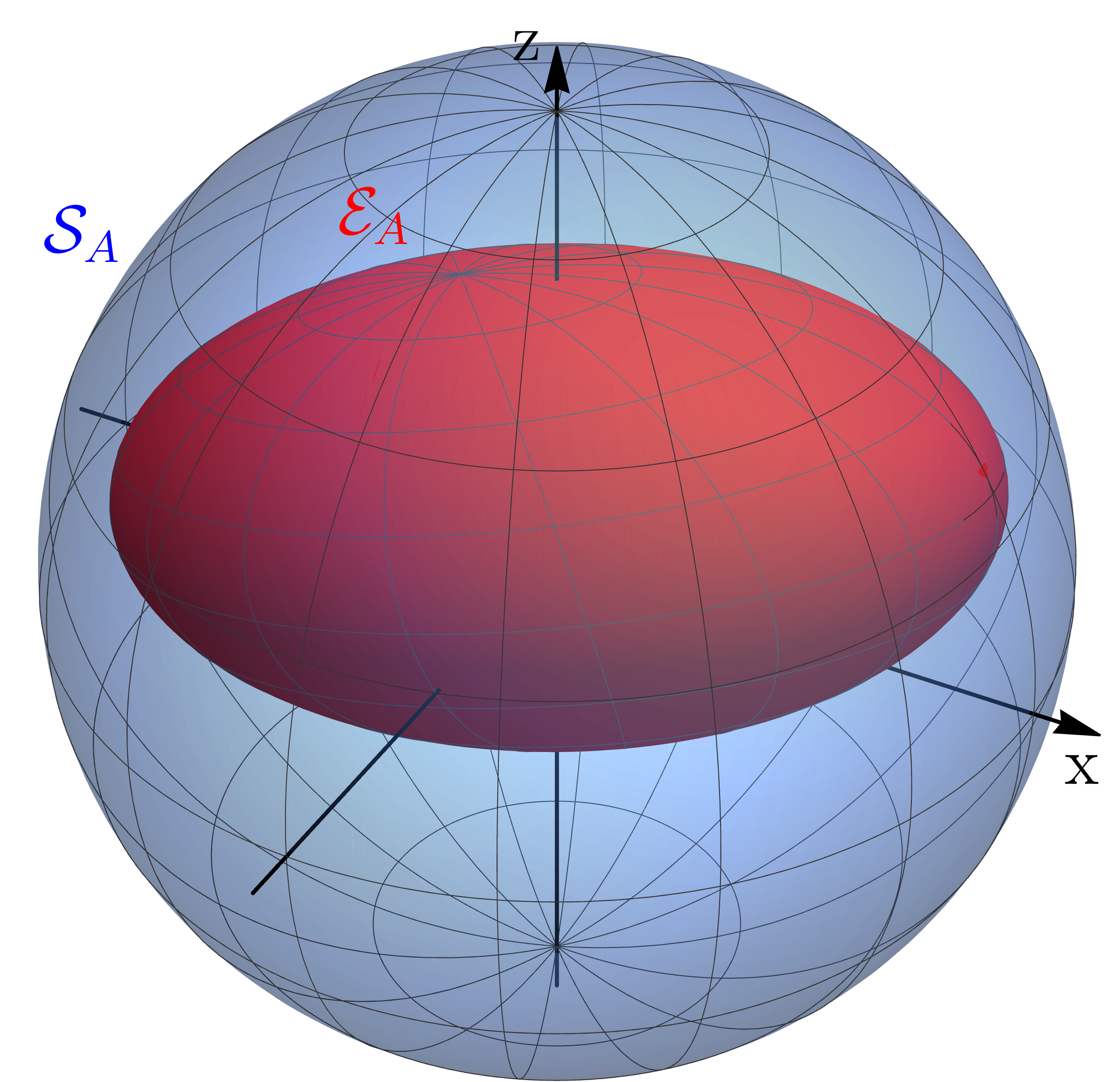}
    \hspace{0.075\textwidth}
    \includegraphics[width=0.35\textwidth]{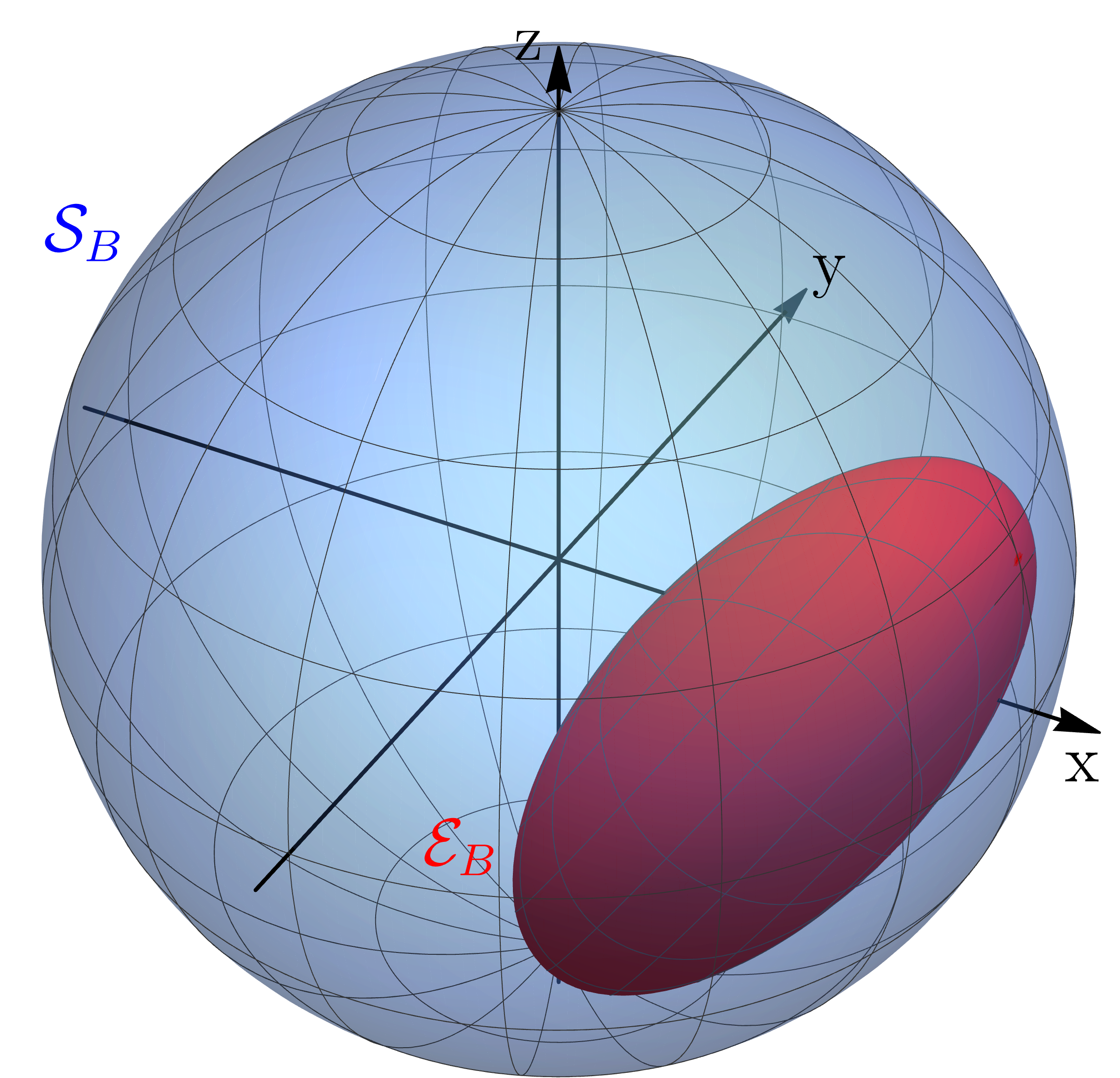}\\
    \vspace{0.075\textwidth}
    \includegraphics[width=0.35\textwidth]{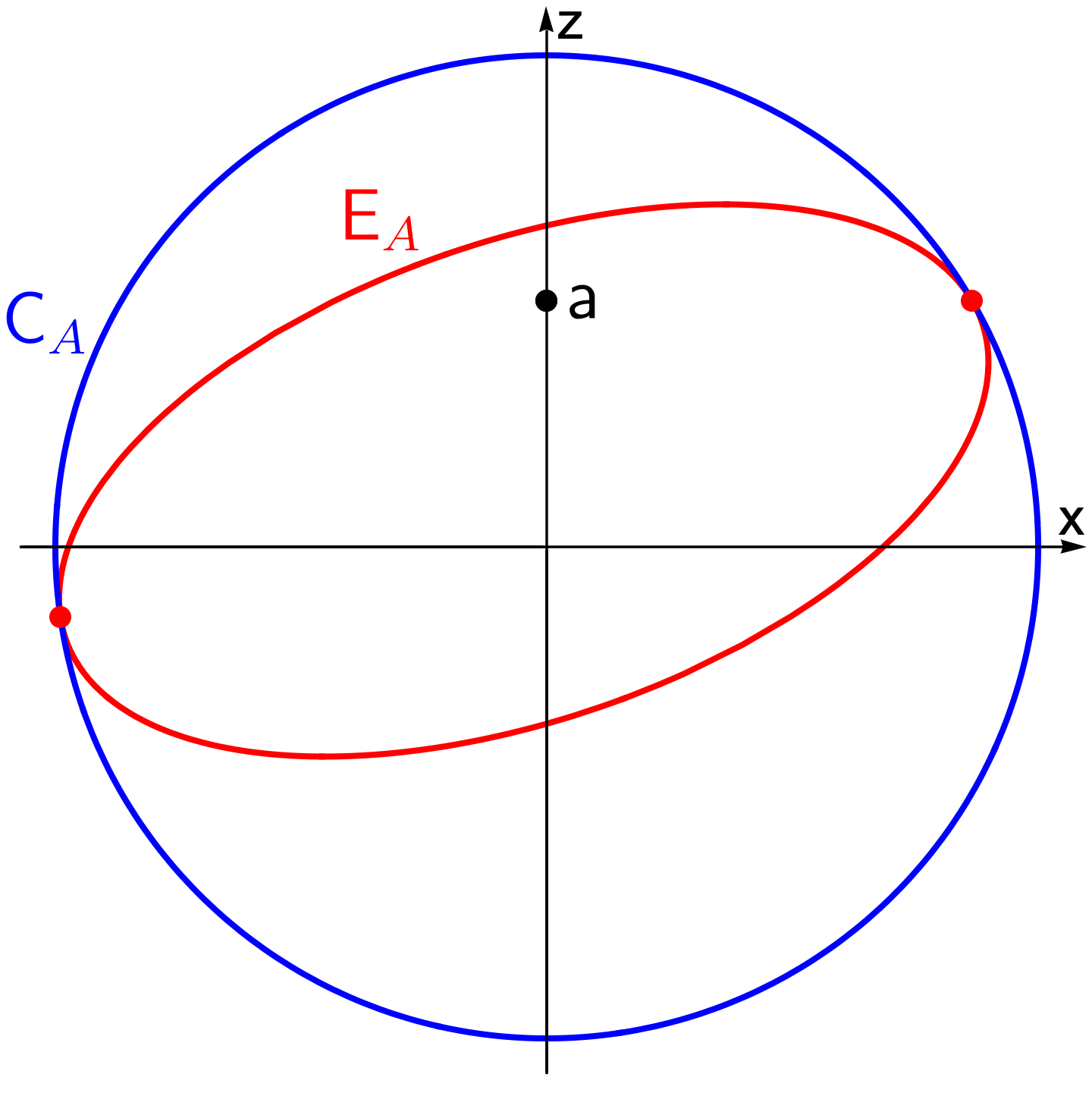}
    \hspace{0.075\textwidth}
    \includegraphics[width=0.35\textwidth]{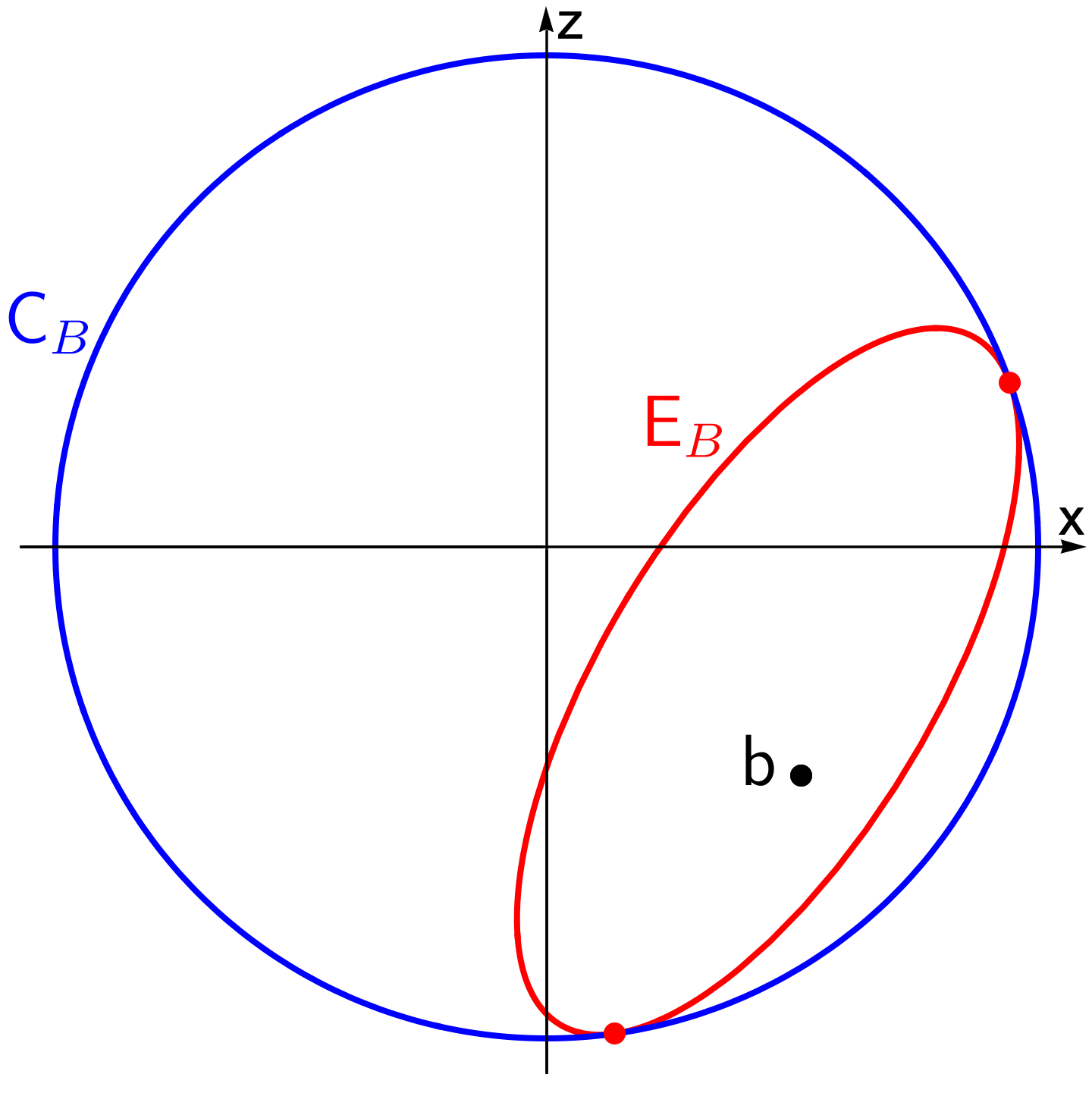}\\
    \caption{QSEs for the state $\rho_{AB}$~\eqref{exastate2p2} with $x=2/3$ and $y=1/4$. Alice's QSE $\mathcal{E}_A$ is tangent to her Bloch sphere $\mathcal{S}_A$ at $(\sqrt{3}|0\rangle + |1\rangle)/2$ and $(-\sqrt{3}|0\rangle + 2|1\rangle)/\sqrt{7}$, while Bob's QSE $\mathcal{E}_B$ is tangent to his Bloch sphere $\mathcal{S}_B$ at $(\sqrt{2}|0\rangle + |1\rangle)/\sqrt{3}$ and $\big[(\sqrt{3}-2)|0\rangle - (\sqrt{2}+\sqrt{6})|1\rangle\big]/\sqrt{15}$. Points $\mathsf{a}$ and $\mathsf{b}$ denote the Bloch vectors of Alice and Bob, respectively.}
    \label{figexap5}
\end{figure*}

Notice that the two states constructed here are both of rank-2. Recent results indicate that any rank-2 two-qubit entangled state is steerable~\cite{zhang25}. 

\end{document}